\documentclass[12pt]{article}

\usepackage{arxiv_package}

\usepackage{subfiles}

\usepackage[utf8]{inputenc} 
\usepackage{graphicx}
\usepackage{float}
\usepackage{url}            
\usepackage{booktabs}       
\usepackage{amsfonts}       
\usepackage{nicefrac}       
\usepackage{microtype}      
\usepackage{bm}
\usepackage{comment}
\usepackage{cleveref}
\usepackage{amsmath}
\usepackage{amsthm}
\usepackage{pifont}
\usepackage{multicol, multirow, colortbl}
\usepackage{caption}
\usepackage{subcaption}
\setcitestyle{square}
\usepackage{tabularx}
\usepackage{multirow}
\usepackage{makecell}
\usepackage{wrapfig}

\usepackage{appendix}

\usepackage{bbm}
\usepackage{dsfont}

\DeclareMathAlphabet{\mathbbold}{U}{bbold}{m}{n}
\newcommand*{\one}{\mathbbold{1}}

\newcommand{\quotes}[1]{``#1''}

\usepackage{hyperref}
\hypersetup{
    colorlinks=true,
    citecolor = blue,
    linkcolor=blue,
   urlcolor=blue
}

\makeatletter
\let\original@footnotemark\footnotemark
\newcommand{\align@footnotemark}{%
  \ifmeasuring@
    \chardef\@tempfn=\value{footnote}%
    \original@footnotemark
    \setcounter{footnote}{\@tempfn}%
  \else
    \iffirstchoice@
      \original@footnotemark
    \fi
  \fi}
\pretocmd{\start@align}{\let\footnotemark\align@footnotemark}{}{}
\makeatother

\def\independenT#1#2{\mathrel{\rlap{$#1#2$}\mkern2mu{#1#2}}}

\usepackage{mathtools}

\DeclarePairedDelimiter\floor{\lfloor}{\rfloor}

\newcolumntype{C}{>{\centering\arraybackslash}X}

\newtheoremstyle{boldstyle} 
  {6pt} 
  {6pt} 
  {} 
  {} 
  {\bfseries} 
  {.} 
  {.5em} 
  {} 

\theoremstyle{boldstyle}
\newtheorem{simulation}{Simulation}

\usepackage{mathtools}

\title{Posterior Conformal Prediction}

	\author[]{Yao Zhang\footnote{Department of Statistics and Data Science, National University of Singapore.} \ \ and  \  Emmanuel  J. Cand\`es\footnote{Departments of Statistics and Mathematics, Stanford University.} }
	
\begin{document}

\maketitle

\begin{abstract}
Conformal prediction is a popular technique for constructing prediction intervals with distribution-free coverage guarantees. The coverage is marginal, meaning it only holds \emph{on average} over the entire population but not necessarily for any specific subgroup. This article introduces \emph{posterior conformal prediction} (PCP), which generates prediction intervals with both marginal and approximate conditional validity for clusters (or subgroups) naturally discovered in the data.
PCP achieves these guarantees by modelling the conditional nonconformity score distribution as a mixture of cluster distributions. Compared to other methods with approximate conditional validity, this approach produces tighter intervals, particularly when the test data is drawn from clusters that are well represented in the validation data. PCP can also be applied to guarantee conditional coverage on user-specified subgroups, in which case it further ensures coverage for underrepresented individuals in each subgroup. 
When the response variable is categorical, PCP can adjust the coverage level based on the classifier's predictive probabilities,  yielding low-cardinality prediction sets if the classifier is well calibrated. 
We demonstrate enhanced performance on datasets from socioeconomics, materials science, and healthcare.

\end{abstract}

\section{Introduction}\label{sect:back}

Conformal prediction \citep{vovk2005algorithmic} is a model-free technique for generating prediction intervals with distribution-free coverage guarantees. Its flexibility has led to applications across diverse domains, such as drug discovery \citep{Ciriano2020} and election forecasting \citep{cherian2020washington}. Despite successful deployments, it has been observed that while conformal prediction techniques perform well for a test sample drawn at random from the general population, they may underperform if it is drawn from a subpopulation, e.g.~from a group for which training data is scarce. For example, in election forecasting, a predictive model trained mostly on data from urban counties may lead to conformal prediction intervals that undercover vote outcomes in rural counties \citep{cherian2020washington}.
To draw reliable conclusions, prediction intervals should ideally guarantee coverage for every individual's response.
However, achieving such individual-level guarantees often results in wide intervals. This paper explores data-adaptive guarantees balancing this trade-off.

\subsection{Marginal, conditional, or in between?}\label{sect:sect:1_1}

We begin with a quick review of conformal prediction methods. 
Imagine we hold a pre-fitted model $\hat{\mu}$ mapping features $X\in\mathcal{X}$ into a prediction of the response $Y \in\mathcal{Y}$. Imagine we also have a separate set of validation data $((X_1,Y_1), \ldots, (X_n,Y_n))$. The model $\hat{\mu}$ may be any nonparametric or machine learning model with the proviso that it is trained independently of the validation data. With this, we wish to make a prediction about the response $Y_{n+1}$ for a new feature vector $X_{n+1}$ while providing a valid prediction interval that quantifies the uncertainty of our prediction. A special instance of split conformal prediction (SCP) \citep{papadopoulos2002inductive} operates as follows: compute the {nonconformity} score $S(X_i, Y_i) = |Y_i - \hat\mu(X_i)|$ for each validation data point $i \in [n]$, and include $y$ in the prediction interval $\hat C_n^{\text{SCP}}(X_{n+1})$, if and only if 
\begin{equation}
  \label{eq:scp}
S(X_{n+1}, y) \leq Q_{1-\alpha} \left(\sum_{i=1}^{n}\frac{1}{n+1}\delta_{S(X_i, Y_i)} + \frac{1}{n+1}\delta_{+\infty}   \right). 
\end{equation}
Above, $Q_{1-\alpha}(\cdot)$ denotes the $(1-\alpha)$-quantile of the empirical distribution in its argument and $\delta_{a}$ denotes a point mass at $a$. With this choice of nonconformity score, the prediction interval for $Y_{n+1}$ takes the form 
\begin{equation}\label{equ:interval_scp}
\hat C_n^{\text{SCP}}(X_{n+1}) = 
 \left[\hat \mu(X_{n+1}) \pm Q_{1-\alpha} \left(\sum_{i=1}^{n}\frac{1}{n+1}\delta_{S(X_i, Y_i)} + \frac{1}{n+1}\delta_{+\infty}   \right)\right].
\end{equation}
In plain English, the prediction interval \eqref{equ:interval_scp} is centered at $\hat \mu(X_{n+1})$ and its width is given by the empirical residual quantile. Although our methods apply to any choice of nonconformity score $S(X_i, Y_i)$ such as those proposed in \citet{lei2018distribution,romano2019conformalized,chernozhukov2021distributional} and prediction intervals constructed as in \eqref{eq:scp}, we shall use the residual magnitude as our nonconformity score to keep things simple. 
An exception is in \Cref{sect:pcc}, which concerns categorical responses.

Set $Z_i = (X_i,Y_i) \sim P$ and assume that the sequence $(Z_1,Z_2,\ldots,Z_{n+1})$ is exchangeable (e.g.~the $Z_i$'s are drawn i.i.d.~from $P$). Then it is well known \citep[Proposition 1]{papadopoulos2002inductive} that the SCP interval \eqref{equ:interval_scp} has {\it marginal validity} in the sense that
	\begin{equation}\label{equ:marginal}
\mathbb P \big\{ {Y_{n+1}\not \in \hat C_n^{\textup{SCP}}(X_{n+1})}	\big\}  \leq \alpha.
\end{equation}
While marginal validity holds, the conditional coverage can fall below $1-\alpha$. In election forecasting, if voter turnout predictions are less accurate for low-income counties, then the coverage rate would fall below $1-\alpha$ for these. This is why we would like to construct a prediction interval with {\em conditional validity}, i.e.,
\begin{equation}\label{equ:asym_valid}
\mathbb P \big\{ {Y_{n+1}\not\in  \hat C_{n}(X_{n+1}) } \mid  X_{n+1}\big\}\leq \alpha.
\end{equation}
In the case where income is included as a feature, conditional validity ensures that the coverage rate is at least $1 - \alpha$ regardless of the county income level.  
It is however well known that, in general, any prediction interval satisfying \eqref{equ:asym_valid} must have infinite expected length in distribution-free settings \citep{vovk2012conditional,lei2014distribution}. 
This negative result is intuitive. Suppose that in our running example of election forecasting, the list of features includes 
continuous geographic, demographic, and economic variables,
and no two counties have exactly the same feature values.
If the distribution of $Y$ can depend arbitrarily on the features $X$, 
then samples with $X\neq x$ provide no information about the distribution of $Y \mid X=x$. Thus, it is not possible to achieve nontrivial conditional validity unless one is willing to make distributional assumptions, such as smoothness.

Expressed differently, 
marginal validity means that if we draw $n$ counties uniformly at random without replacement, and a test county uniformly at random from the remaining counties, then the coverage rate of the prediction interval \eqref{equ:interval_scp}---this is an average over the validation {\em and} test draws---is at least $1-\alpha$.  In contrast, suppose the features identify the test county as ``Santa Clara'', say. Then conditional validity means we need a valid prediction interval {\em for} Santa Clara. There is however only one Santa Clara county, and without making assumptions, we cannot say anything about  its vote outcome from votes in other counties.

We thus need a notion of validity which interpolates between marginal and conditional validity. Informally, we would like validity to hold for counties of the same type as Santa Clara. This line of reasoning prompts the question: what is a county type? By and large, the existing literature on {\it localized} conformal prediction methods treats counties as neighbors if their feature vectors are close to each other; see \Cref{sect:related_work} for a thorough discussion. {This local notion can become problematic in high-dimensional feature spaces, where even the two nearest points can still be far apart.}

This paper presents a radically different approach to modelling similarity between counties. {Posterior Conformal Prediction} (PCP) clusters counties into latent types, and proceeds by constructing prediction intervals that ensure coverage within each type. Crucially, these clusters are learned from data, specifically from the counties' nonconformity scores, rather than from their raw features.  In this way, two counties will be declared similar if the conditional distributions of their residuals are similar. 
This means that two counties can be of the same type even if their features are very different, as long as the residuals from a predictive model follow similar distributions. For instance, we may learn from data that what mostly drives the model error is along the rural vs.~urban axis; e.g.~the model may perform well in urban counties but poorly in rural ones. By placing the test county on this axis, we can perhaps improve 
conditional coverage across a wide range of demographic and economic conditions. 
The novel aspect of PCP is that looking at the distribution of the nonconformity scores across the feature space informs the 
construction of a prediction interval for Santa Clara, leveraging non-local information from counties that may lie far away in the feature space but exhibit similar nonconformity score distributions.
As we will see, this allows PCP to produce tighter intervals than localization methods as demonstrated in \Cref{sect:exp_preview}.

\subsection{Validity over cluster membership probabilities}\label{subsection:sv}

To formalize matters, assume for now that the (exchangeable) nonconformity scores $R_1, \ldots, R_{n+1}$, defined by $R_i = S(X_i,Y_i) = |Y_i - \hat\mu(X_i)|$, follow a finite mixture model: 
 \begin{equation}\label{equ:mixture}
	R_i \mid X_i \sim  \sum_{k=1}^{K}\pi_k(X_i)f_k,
\end{equation}
where  $\pi(X_i)$ is a $K$-dimensional cluster membership probability vector, and $f_1,\dots, f_{K}$ are $K$ distinct, potentially arbitrary, probability density or mass functions. From here on, we shall refer to the $f_k$'s as cluster distributions.
At one extreme, a model with $K=n+1$ allows each observation to belong to its own cluster,
i.e., $R_i\mid X_i \sim f_i$ for every $i\in [n+1]$, placing no assumptions whatsoever on the residual distributions. At the other extreme, $K = 1$ assumes that all residuals are conditionally drawn from the same cluster distribution, i.e., $R_i \mid X_i \sim f_1$ for all $i$.
In many applications, the conditional residual distribution can be well approximated using a mixture model with $K  \ll n$. For instance, the 
 mixture model  \eqref{equ:mixture} can model the situation discussed earlier, 
 {in which a predictive model forecasts vote outcomes more accurately for urban clusters than for rural clusters.} 

For pedagogical purposes, we first assume that the true cluster membership probabilities $\pi(X)$ are known although we shall shortly dismiss this assumption. Then to construct prediction intervals, we propose upweighting validation points whose cluster membership probabilities $\pi(X_i)$ are near those of the test point $\pi(X_{n+1})$. To do this, we start with  nonnegative weights $w_i \propto \phi(\pi(X_i),\pi(X_{n+1}))$ normalized to sum to one, where $\phi$ is a measure of similarity between probability vectors. Then set 
\begin{equation}\label{equ:oracle}
\hat C_{n}(X_{n+1}) = \left[\hat \mu(X_{n+1}) \pm Q_{1-\alpha} \left(\sum_{i=1}^{n}w_i\delta_{R_i} + w_{n+1}\delta_{+\infty}   \right)\right]. 
\end{equation}
It turns out that if this construction provides coverage conditional on the summary statistic $\pi(X_{n+1})$, then it achieves conditional validity. 
\begin{proposition}\label{prop:oracle_coverage_2}
Under the mixture model \eqref{equ:mixture}, the interval  $\hat C_{n}(X_{n+1})$ in \eqref{equ:oracle} obeys 
\begin{equation}\label{equ:oracle_coverage}
\mathbb P \big\{ Y_{n+1}\not\in   \hat C_{n}(X_{n+1}) \mid X_{n+1} \big\} = \mathbb P \big\{ Y_{n+1}\not\in   \hat C_{n}(X_{n+1}) \mid \pi(X_{n+1}) \big\}.
\end{equation}
\vspace{-2.5em}
\end{proposition}
In practice, the mixture model \eqref{equ:mixture} may of course only hold approximately, and even if it were exact, we would not know the cluster membership probabilities $\pi(X_i)$. Posterior conformal prediction (PCP) will address these issues by controlling the miscoverage rate conditional on a proxy for $\pi(X_{n+1})$. Before explaining this approximation, we preview performance.

\subsection{Preview of experimental results}\label{sect:exp_preview}

We work with the Communities and Crime dataset \citep{dua2017uci}, which we shall examine in greater detail in later sections. In a nutshell, 
we train a random forest model (as the predictive model $\hat \mu$ in \Cref{sect:sect:1_1})
to predict the per capita violent crime rate of communities using 99 features, most of which provide demographic, socioeconomic, and crime-related information.
To assess marginal coverage, we compute the coverage rate over all communities in the test portion of the dataset. 
To assess conditional coverage in finite samples, we follow \citet{cauchois2020knowing,romano2020classification} to compute the ``worst-slice'' conditional coverage. This roughly corresponds to the worst coverage over all subpopulations given by $\big\{x: v^{\top}x\in [a,b]\big\}$, where $a$, $b$, and $v$ can vary arbitrarily; when $v$ is a basis vector, the ``worst-slice'' method considers the local coverage for a particular feature in $x$. When $v$ combines multiple features (e.g., unemployment and poverty rates), it can represent the economic condition of communities, as well as social and public security conditions. In this sense, the ``worst-slice'' method also assesses the local coverage across communities under a variety of conditions.

\begin{figure}[t]
\vspace{-5pt}
     \centering
     \hspace{-5pt}
     \begin{subfigure}[b]{0.329\textwidth}
         \centering         \includegraphics[width=\textwidth]{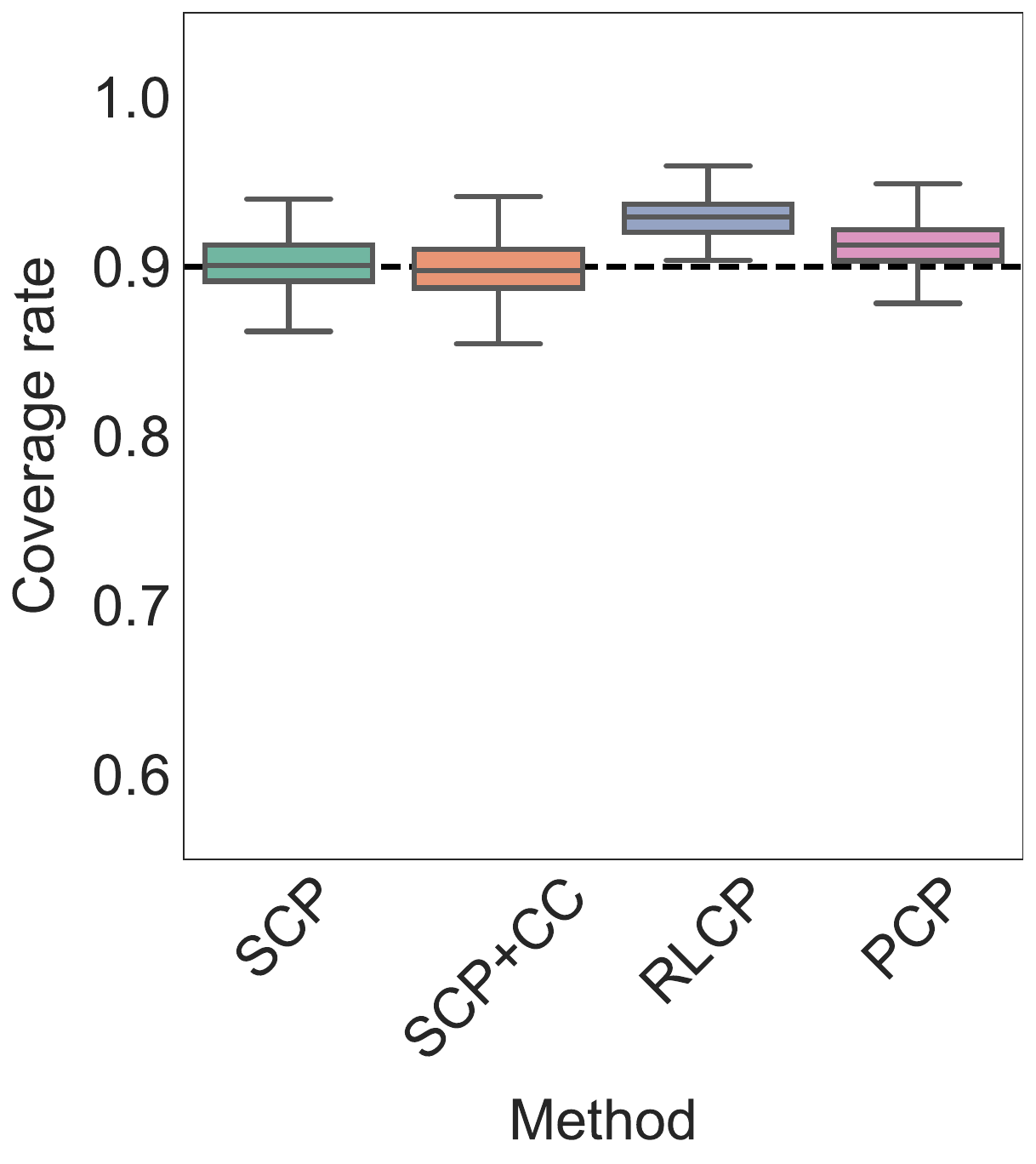}
         \caption{Marginal coverage.
         }
           \label{fig:wscr1}
     \end{subfigure}
     \begin{subfigure}[b]{0.329\textwidth}
         \centering
         \includegraphics[width=\textwidth]{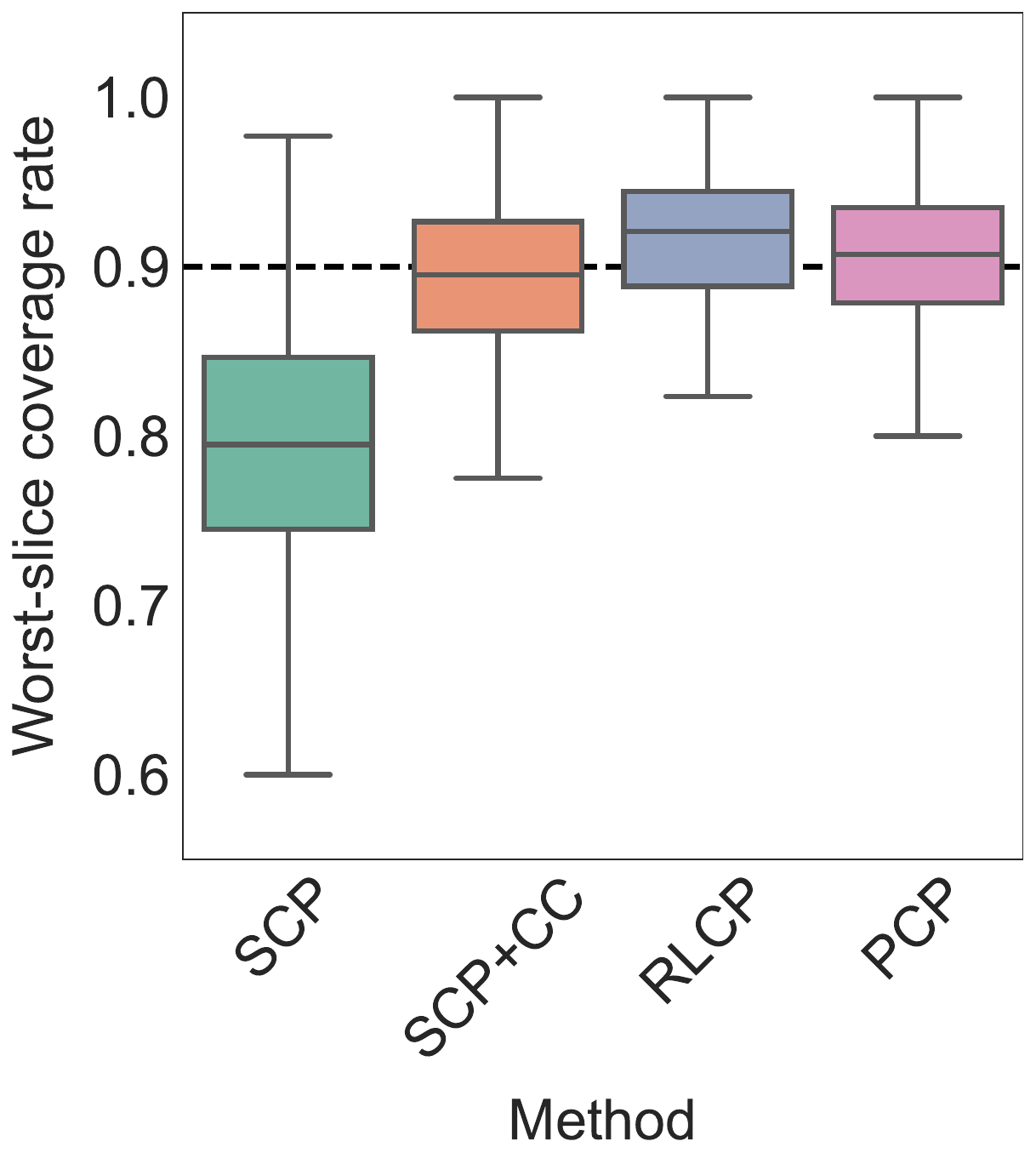}
          \caption{Worst-slice coverage.}
           \label{fig:wscr2}
     \end{subfigure}
          \begin{subfigure}[b]{0.329\textwidth}
         \centering
         \includegraphics[width=\textwidth]{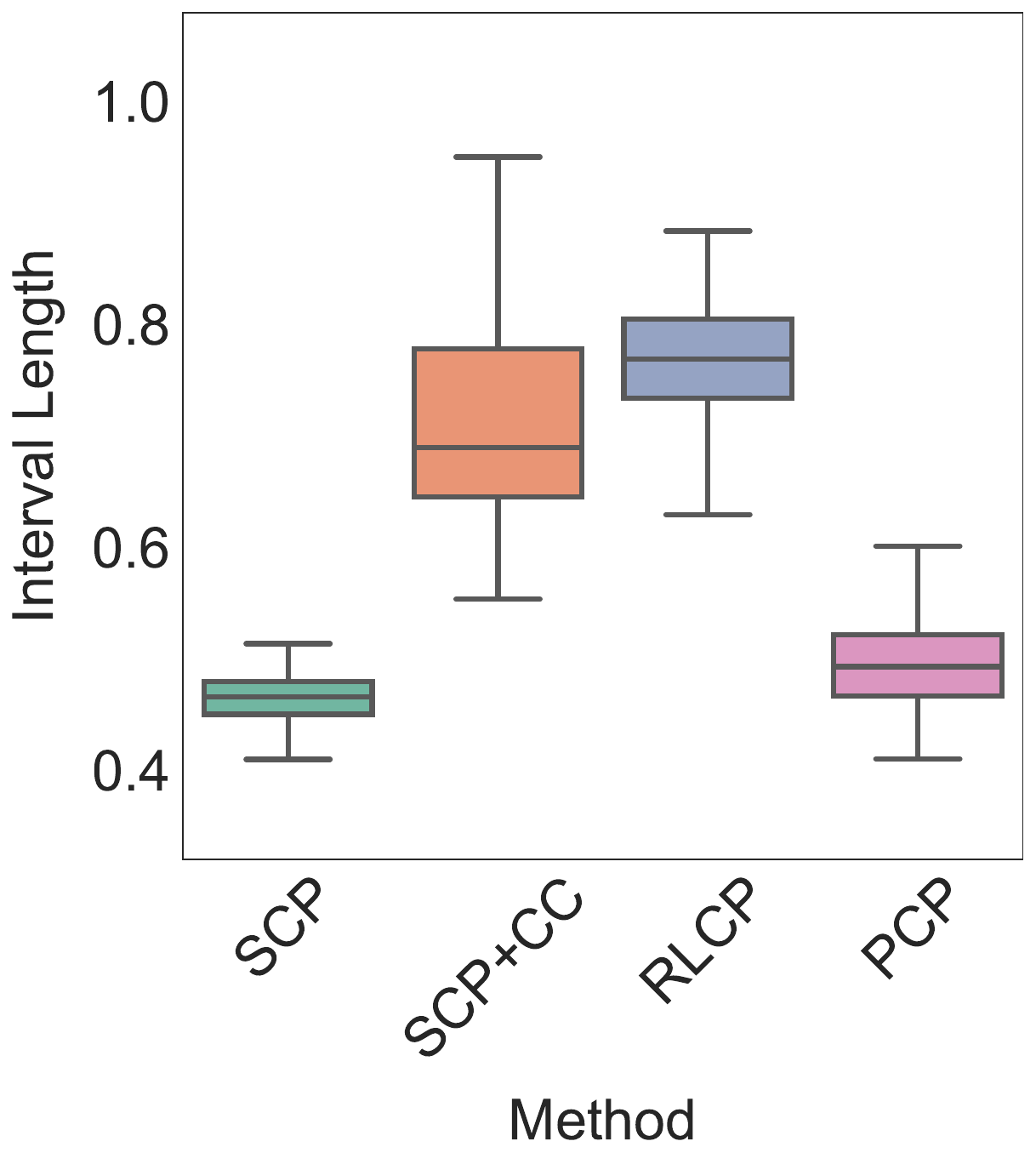}
         \caption{Interval length.
         }
          \label{fig:wscr3}
            \end{subfigure}
        \caption{Coverage rates and interval length of conformal prediction methods on the Communities and Crime dataset. If an interval is infinite, we replace its length with twice the largest absolute prediction error in the test set, which represents the range of the prediction error.     
        The results are from 200 runs of the experiments. The meaning of the labels is explained in the text.}
        \label{fig:wscr}
\end{figure}

\Cref{fig:wscr1} shows that all conformal prediction methods achieve a marginal coverage rate close to the target rate of 0.9.  Split conformal prediction (SCP) \citep{vovk2005algorithmic} has a worst-slice conditional coverage rate of around 0.8 (\Cref{fig:wscr2}). More advanced methods, including
 SCP+conditional calibration (CC) \citep{gibbs2023conformal}, randomly-localized conformal prediction (RLCP) \citep{hore2023conformal}, and our proposed method, \emph{posterior conformal prediction} (PCP), achieve a worst-slice conditional coverage rate close to 0.9. Among these advanced methods, only PCP can maintain an interval length comparable to that of SCP, as shown in \Cref{fig:wscr3}.
This improvement stems from the mixture modelling step in PCP, which identifies clusters (i.e., subpopulations) where the random forest model performs unevenly, and adjusts the interval length accordingly to ensure valid coverage in those regions. 
Code for implementing PCP and reproducing the experiments and figures in our article is available at \href{https://github.com/yaozhang24/pcp}{https://github.com/yaozhang24/pcp}.

\subsection{Related work}\label{sect:related_work}
 
We next compare PCP with existing methods that achieve approximate conditional validity, including the methods compared in the experiments above. These methods can be broadly categorized into localization and grouping. At a high level, PCP is related to both. The key difference is that PCP defines similarity between data points through their cluster membership probabilities, which are learned from the nonconformity scores rather than pre-specified based on the raw features.

Localization modifies the SCP interval in \eqref{equ:interval_scp} by using a weighted empirical quantile, where the weight on $\delta_{R_i}$ is proportional to a similarity measure between $X_i$ and $X_{n+1}$, e.g.,  $w_i \propto  \exp\big\{- \frac{1}{2\sigma^2} \|X_i- X_{n+1}\|^2\big\}$ in the localized conformal prediction (LCP) method proposed by \citet{guan2023localized}. To make this work, LCP needs to recalibrate the level $\alpha$ of the quantile to maintain marginal validity. Conceptually, we can view LCP as running SCP in the neighborhood of $X_{n+1}$, thereby improving the conditional coverage of $Y_{n+1}$. 
\citet{hore2023conformal} proposed an extension of LCP called randomly-localized conformal prediction (RLCP), where they replace 
 $X_{n+1}$ in the LCP weights by a random draw $\tilde X_{n+1}\sim \mathcal N (X_{n+1} ,\sigma^2 I)$. This randomization step makes RLCP marginally valid without any level adjustment.
Decreasing the variance parameter {$\sigma^2$} in (R)LCP widens the interval since it effectively shrinks the sample size. This strategy is different from PCP, which constructs intervals using a similarity measure between the cluster membership probabilities of $X_i$ and $X_{n+1}$.
Although $X_i$ and $X_{n+1}$ may be far apart, PCP will consider them to be close if the nonconformity score distributions conditional on $X_i$ and $X_{n+1}$ are similar.
In other words, PCP does not rely on any notion of local similarity {in the feature space}. Moreover, we shall see that PCP can also be used to solve other calibration problems, further distinguishing it from existing localization methods.

Grouping is often used to attain coverage guarantees for subgroups defined in terms of some sensitive attributes such as gender and age.
When these subgroups form a partition of the feature space $\mathcal X$, running SCP with the data points from the same group as $X_{n+1}$ can generate a group-conditional valid interval  \citep{vovk2003mondrian,romano2020malice,Kiyani2024ConformalPW,zhou2024conformal}. 
\citet{jung2022batch} introduce a more direct method called multi-calibration, which aims to achieve coverage conditional on a set of pre-specified and possibly overlapping subgroups. This method does not have finite sample coverage guarantees and can undercover in practice. 
\citet{gibbs2023conformal} propose a conditional calibration (CC) 
 method to generate intervals with exact finite-sample coverage guarantees for any pre-specified function class, subsuming a class of subgroups as a special case. 
However, using a complex function class (e.g., a large number of subgroups) may result in wide intervals (due to small sample sizes within each subgroup). We find that in contrast, PCP automatically balances the trade-off between conditional coverage and interval length.

\section{Posterior conformal prediction (PCP)}\label{sect:pcp_main}

\subsection{PCP interval with true $\pi$}\label{sect:pi_star_interval}

Recall the mixture model \eqref{equ:mixture} and let us return to the interval \eqref{equ:oracle} from \Cref{subsection:sv}, which assumes knowledge of the true cluster membership probabilities. We have seen in Proposition \ref{prop:oracle_coverage_2} that conditional validity is implied by the coverage guarantee conditional on $\pi(X_{n+1})$. There is however a major outstanding issue: it is {entirely possible} that the $\pi(X_i)$'s are all distinct across $i\in [n+1]$. When this occurs, knowledge of the values of $\pi(X_i)$ perfectly identifies each data point, and then nontrivial conditional inference is again impossible, as in the Santa Clara discussion surrounding \eqref{equ:asym_valid}.

A natural idea is to obfuscate the value of $\pi(X_{n+1})$ by introducing randomness as not to reveal the identity of the test point among the unordered collection of calibration samples and test sample. An elegant way of doing this is via sampling. We draw a sample of size $m$ from the multinomial distribution with parameters 
$\pi (X_{n+1})$ and return empirical frequencies: 
\begin{equation}\label{equ:disc}
L^\star  = (L_1^\star ,\dots , L_K^\star )\sim \text{Multi}(m,\pi(X_{n+1}))  \text{ and set }    \pi^\star(X_{n+1}) =  L^\star/m.
\end{equation}
We refer to the sample size $m$ as the precision parameter below.

Upon revealing $\pi^\star(X_{n+1})$, 
we would not be sure which of the $n+1$ multinomial distributions $\text{Multi}(m,\pi(X_{i}))$ was used to generate it---even in the case where the $\pi(X_{i})$'s are all distinct. Formally, let 
$E_{\bm z}$ denote the event that the set $\{Z_1, \dotsc, Z_{n+1}\} = \{z_i = (x_i, y_i): i \in [n+1]\}$. This says that the vector $(Z_1, \ldots, Z_{n+1})$ is a permutation of the $z_i$'s but we do not know which one. Bayes' rule gives that for each $i \in [n+1]$, 
\begin{equation}\label{equ:likelihood}
\mathbb P \big\{Z_{n+1} = z_i\mid E_{\bm z}, \pi^\star(X_{n+1})	 \big\}\propto \prod_{k=1}^{K}\big[ \pi_k(x_{i})\big]^{ L_k^\star 	}.
\end{equation}
It is clear from above that we are not sure about the identity of the test point. 
We now construct intervals as in \eqref{equ:oracle} with weights given by the right-hand side of \eqref{equ:likelihood}:
\begin{equation}\label{equ:w_i_star}
w_i^\star=\frac{\prod_{k=1}^{K}\big[\pi_k(X_i)\big]^{L_k^\star}}{\sum_{j=1}^{n+1}\prod_{k=1}^{K}\big[\pi_k(X_j)\big]^{L_k^\star}} = \frac{\prod_{k=1}^{K}\big[\pi_k(X_i)\big]^{m  \pi_k^\star(X_{n+1}) }}{\sum_{j=1}^{n+1}\prod_{k=1}^{K}\big[\pi_k(X_j)\big]^{m  \pi_k^\star(X_{n+1}) }}.
\end{equation}
\begin{theorem}\label{thm:validity_1}
Assume $(Z_1,\ldots,Z_{n+1})$, with $Z_i=(X_i,Y_i)$, is exchangeable. The prediction interval $\hat C_n(X_{n+1})$ \eqref{equ:oracle}, with weights $w_i^\star$ defined in \eqref{equ:w_i_star},
satisfies 
\begin{equation}\label{equ:tilde}
\mathbb P \big\{ Y_{n+1}\not\in   \hat C_{n}(X_{n+1})  	
\mid	\pi^\star(X_{n+1}) \big\}\leq \alpha.
\end{equation}
Marginalizing out $\pi^\star(X_{n+1})$ also establishes that $\hat C_{n}(X_{n+1})$  has marginal validity. 
\end{theorem}
Returning to our election-forecasting example, the empirical frequencies \(\pi^\star(X_{n+1})\) describe how often Santa Clara is assigned to each cluster over the \(m\) draws in \eqref{equ:disc}. With this in mind, the guarantee in \eqref{equ:tilde} means that the interval \(\hat C_n(X_{n+1})\) has valid coverage for counties whose cluster-assignment frequencies are similar to those of Santa Clara (similarity is thus defined through the cluster-assignment frequencies rather than the raw covariates). Moreover, as the precision parameter \(m\) grows, \(\pi^\star(X_{n+1})\) concentrates around \(\pi(X_{n+1})\). Proposition \ref{prop:oracle_coverage_2} then implies that this cluster-based guarantee approaches conditional coverage at the test point itself, so \(\hat C_n(X_{n+1})\) achieves approximately valid coverage conditional on \(X_{n+1}\).

To demonstrate the necessity of randomization, we prove in \Cref{sect:random_nonrandom} that our interval becomes invalid even if it uses similar but non-randomized weights.
`Noising' data is of course a common strategy in the literature for constructing valid intervals as can be seen in \cite{tian2018selective,barber2022testing,hore2023conformal,panigrahi2024exact}. 
PCP randomizes the interval by sampling a multinomial draw $L^\star$ in \eqref{equ:disc}, which yields empirical frequencies $\pi^\star(X_{n+1})=L^\star/m$ preserving a cluster membership interpretation. In addition, when $\pi(X_{n+1})$ is binary with one entry equal to one, i.e. the membership of $X_{n+1}$ is known with no uncertainty, this sampling step introduces no error since  $\pi^\star(X_{n+1})=\pi(X_{n+1})$.

\textbf{RLCP vs.~PCP.} We have seen that RLCP adds noise by drawing $\tilde X_{n+1}$ from $\mathcal N (X_{n+1},\sigma^2 I)$ and constructs intervals with a coverage guarantee conditional on $\tilde X_{n+1}$. When the value of $\sigma^2$ is small, $X_{n+1}$ may effectively become the only neighbor of $\tilde X_{n+1}$ among all $X_i$'s and the RLCP interval becomes infinite. This happens as soon as the weight $w_{n+1} \propto  \exp\big\{- \|X_{n+1} - \tilde X_{n+1}\|^2/2\sigma^2 \big\}$ assigned to $\delta_{+\infty}$ exceeds $\alpha$. Conversely, using a large value of $\sigma^2$ increases the distance between $\tilde X_{n+1}$ and $X_{n+1}$, weakening the guarantee conditional on $\tilde X_{n+1}$. Thus, while introducing noise may be a simple way to discard information about $X_{n+1}$, it may not effectively {balance the trade-off between conditional coverage and interval length.}

In contrast, let $D_{\textup{KL}}(\hspace{-1pt}\cdot \hspace{-0.5pt} \| \hspace{-0.5pt}\cdot \hspace{-1pt} )$ denote the Kullback--Leibler (KL) divergence. 
Then, for large \(m\), the weights \(w_i^\star\) in \eqref{equ:w_i_star} can be approximated by
\begin{equation}\label{equ:likelihood_2}
w_i^\star
\approx
\frac{\exp\left\{- m D_{\textup{KL}}\big(\pi(X_{n+1})\| \pi(X_i)\big)\right\}}
{\sum_{j=1}^{n+1}\exp\left\{- m D_{\textup{KL}}\big(\pi(X_{n+1})\| \pi(X_j)\big)\right\}}.
\end{equation}
The weights \(w_i^\star\) can differ substantially from the local weights used by RLCP, especially when the validation points are far from the test point in the raw feature space, yet have similar membership probabilities. In such settings, the interval constructed from $w_i^\star$ can remain finite, that is, \(w_{n+1}^\star < \alpha\), even when the precision parameter $m$ is large. The next theorem illustrates this phenomenon through an example in which the probability that the interval is finite can be characterized explicitly. 

\begin{theorem}\label{thm:length}
Assume the model \eqref{equ:mixture} holds, $\pi(X_{n+1})\sim \textup{Dir}(\Lambda := (\Lambda_1,\Lambda_2,\Lambda_3))$ and $\bar \Lambda = \sum_{k=1}^3\Lambda_k$.
 The interval $\hat  C_{n}(X_{n+1})$ described in \Cref{thm:validity_1} satisfies 
\vspace{2pt}
\begin{align*}
& \mathbb P \big\{  |\hat C_{n}(X_{n+1})| < \infty  	\mid  X_1,\dots,X_{n},\pi^\star(X_{n+1}) = \Lambda/\bar \Lambda \big\}  \\
 = &  \min \left\{\frac{\alpha}{1-\alpha}\sum_{i=1}^{n}\exp \left(-m D_{\textup{KL}}\big( \pi^\star(X_{n+1}) \| \pi(X_{i}) \big)  \right),1\right\}^{(m+\bar \Lambda)/m }+ O(1/m).\footnotemark 
\end{align*}
\footnotetext{
The proof in \Cref{sect:length} is based on an interesting observation that, after some transformation, the non-normalized weight of $\delta_{+\infty}$ in $\hat  C_{n}(X_{n+1})$  approximately follows an exponential distribution. We also provide a set of  experiments to verify this theory across various $\Lambda$.}
\end{theorem}
The result above gives an example in which membership probability vectors on the simplex follow a Dirichlet distribution and are concentrated around the mean \(\Lambda/\bar\Lambda\). By conditioning on \(\pi^\star(X_{n+1})=\Lambda/\bar\Lambda\), we place the randomized membership vector of the test point at a high-density location, so that many validation points can also have membership probabilities \(\pi(X_i)\) close to \(\pi^\star(X_{n+1})\) in KL divergence. If a nontrivial fraction of these points satisfy
$D_{\textup{KL}}\big(\pi^\star(X_{n+1})\|\pi(X_i)\big)=O(1/m)$,
then the sum in the theorem can reach \((1-\alpha)/\alpha\), implying that \(\hat C_n(X_{n+1})\) is finite with high probability. In this sense, \Cref{thm:length} illustrates that increasing \(m\) can improve conditional coverage without necessarily forcing the interval to become infinite.

\subsection{PCP interval with estimated $\pi$}
\label{sect:e_validity}

In reality, we do not have access to the true cluster membership probabilities, and we shall instead learn them from data. The estimated
membership probabilities, and hence the interval weights, may depend on the responses.

For the conformal machinery to go through, the learning algorithm for the membership probabilities must treat the validation and test samples symmetrically. 
However, the test response \(Y_{n+1}\) is unobserved. Thus, for each candidate value \(y\), we impute \(Y_{n+1}=y\) and relearn the
membership probabilities from the imputed dataset
\[
\mathcal D_{[n+1]}^y
=
\{(X_1,Y_1),\dots,(X_n,Y_n),(X_{n+1},y)\}.
\]
Concretely, fix a number \(J\) of clusters. Our learning algorithm
\(\mathcal A\) maps \(\mathcal D_{[n+1]}^y\) to the estimated membership probabilities \(\hat\pi^y\).
We require \(\mathcal A\) to satisfy
\begin{equation}\label{equ:sym}
\mathcal A ((x_{\sigma (1)},y_{\sigma (1)}),\dots,
(x_{\sigma (n+1)},y_{\sigma (n+1)}))
\stackrel{d}{=}
\mathcal A ((x_{1},y_{1}),\dots,(x_{n+1},y_{n+1})),
\end{equation}
for all permutations \(\sigma\) on \([n+1]\). The equality in distribution
allows \(\mathcal A\) to be randomized, provided that its randomization is
independent of the data order.

\textbf{Notation.} Before defining PCP with estimated membership probabilities, \Cref{tab:pcp_notation} summarizes the notation changes from true to estimated probabilities.
We retain the superscript \(y\) throughout the definition of the prediction set, since \(y\) ranges over candidate responses. For the coverage statements
and subsequent discussion, when the imputed value is the realized response
\(Y_{n+1}\), we suppress the superscript and write
\begin{equation}\label{equ:shorthand}
\begin{aligned}
\mathcal D_{[n+1]} & = \mathcal D_{[n+1]}^{Y_{n+1}}, 
&\qquad
\hat\pi(X_i) &= \hat\pi^{Y_{n+1}}(X_i), \\
\hat L^\star &= \hat L^{\star,Y_{n+1}},
&\qquad
\hat\pi^\star(X_{n+1})
&=
\hat\pi^{\star,Y_{n+1}}(X_{n+1})
=
\hat L^\star/m .
\end{aligned}
\end{equation}
The same convention applies to the interval weights
\(\hat w_i^\star(\mathcal D_{[n+1]})\).

\begin{table}[t]
\centering
\footnotesize
\setlength{\tabcolsep}{10pt}
\renewcommand{\arraystretch}{2.5}
\begin{tabular}{lcc}
\toprule
Quantity & True & Estimated with imputed \(y\) \\
\midrule
Membership probabilities
& \(\pi(X_i)=[\pi_1(X_i),\dotsc,\pi_K(X_i)]\)
& \(\hat\pi^y(X_i)=[\hat\pi_1^y(X_i),\dotsc,\hat\pi_J^y(X_i)]\) \\

Multinomial draw
& \(L^\star \sim \mathrm{Multi}(m,\pi(X_{n+1}))\)
& \(\hat L^{\star,y} \sim \mathrm{Multi}(m,\hat\pi^y(X_{n+1}))\) \\

Randomized probabilities
& \(\pi^\star(X_{n+1})=L^\star/m\)
& \(\hat\pi^{\star,y}(X_{n+1})=\hat L^{\star,y}/m\) \\

\bottomrule
\end{tabular}
\vspace{2pt}
\caption{Notation for true and estimated membership probabilities in PCP. The superscript \(y\) indicates that the probabilities are computed after imputing the test response as \(y\), and the superscript \(\star\) denotes the randomized multinomial version.}
\label{tab:pcp_notation}
\end{table}

\textbf{PCP interval.} The prediction interval of PCP is defined as 
\begin{equation}\label{equ:posterior_p}
\hat C_{n}^{\textup{PCP}}(X_{n+1}) = \left\{y\in \mathcal{Y}: P_{n}^{\textup{PCP}}(y) := \sum_{i=1}^{n+1}\hat w_{i}^\star(\mathcal D_{[n+1]}^y)\one {\left\{R_i\geq R_{n+1}^{y}\right\}}  >\alpha \right\},
\end{equation}
where $R_i = |Y_i - \hat \mu(X_i)|, R_{n+1}^{y} = |y-\hat \mu(X_{n+1})|$, and the model $\hat \mu(\cdot)$ is as before. The weights $\hat w_i^\star(\mathcal D_{[n+1]}^y)$ are defined using $\hat L^{\star,y} \sim \text{Multi}(m,\hat \pi^{y}(X_{n+1}))$ by 
\[
\hat  w_{i}^\star(\mathcal D_{[n+1]}^y)
=
\frac{
\prod_{k=1}^{J}\big[ \hat \pi_k^y(X_{i})\big]^{m\hat \pi_k^{\star,y}(X_{n+1})}
}{
\sum_{j=1}^{n+1} \prod_{k=1}^{J}\big[ \hat \pi_k^y(X_{j})\big]^{m \hat \pi_k^{\star,y}(X_{n+1})}
},
\qquad
\hat \pi^{\star,y}(X_{n+1}) = \hat L^{\star,y}/m.
\]
When computing the interval in \eqref{equ:posterior_p}, we need to estimate \(\hat \pi^y\) for all \(y\in\mathcal Y\). We address this challenge by
developing an efficient algorithm \(\mathcal A\), which has linear complexity in the sample size \(n\).
Since this algorithm is not required for the applications of PCP in \Cref{sect:fair,sect:pcc}, we defer its description to \Cref{sect:alg}. 

\vspace{5pt}
\begin{proposition}\label{thm:pcp_cover}
Using the shorthand notation in \eqref{equ:shorthand}, suppose that the sequence \((Z_1,\ldots,Z_{n+1})\), with \(Z_i=(X_i,Y_i)\), is exchangeable.
Assume that the learning algorithm \(\mathcal A\) satisfies the symmetry condition in \eqref{equ:sym}. Then, for any positive integer \(m\), the
posterior conformal prediction interval \eqref{equ:posterior_p} satisfies
\begin{equation}\label{equ:pcp_guarantee}
\mathbb P \Big\{ Y_{n+1}\not\in \hat C_{n}^{\textup{PCP}}(X_{n+1}) \mid  \hat\pi^\star(X_{n+1}) \Big\} \leq \alpha.
\end{equation}
Furthermore, if all the residuals $R_{[n+1]}$ are distinct with probability 1, 
\begin{equation}\label{equ:pcp_guarantee_2}
\mathbb P \Big\{ Y_{n+1}\not\in \hat C_{n}^{\textup{PCP}}(X_{n+1}) \mid  \hat\pi^\star(X_{n+1}) \Big\}
\geq \alpha - \mathbb E \Big\{ \max_{i\in [n+1]} \hat w_{i}^\star(\mathcal D_{[n+1]}) \mid  \hat\pi^\star(X_{n+1}) \Big\}.
\end{equation}
\vspace{-1.5em}
\end{proposition}
If the maximum weight in \eqref{equ:pcp_guarantee_2} is small, our interval is finite and has a coverage rate close to $1-\alpha$. This occurs when sufficiently many $\hat \pi(X_{i})$'s are close to  $\hat \pi^{\star}(X_{n+1})$, which resembles the condition for obtaining a finite interval in \Cref{thm:length}. To control the magnitude of the weights, we select the number $J$ of cluster distributions and the precision parameter $m$ using the dataset originally used to fit the model $\hat\mu$. Specifically, we generate a separate set of residuals from this dataset via cross-validation and use them to construct PCP intervals and choose $J$ and $m$ so that most intervals remain finite and stable; see \Cref{sect:hyper} for more detail.

\begin{figure}[t]
\vspace{-35pt}
     \centering
     \begin{subfigure}[b]{0.45\textwidth}
     \hspace{-15pt}
         \centering
         \includegraphics[width=0.91\textwidth]{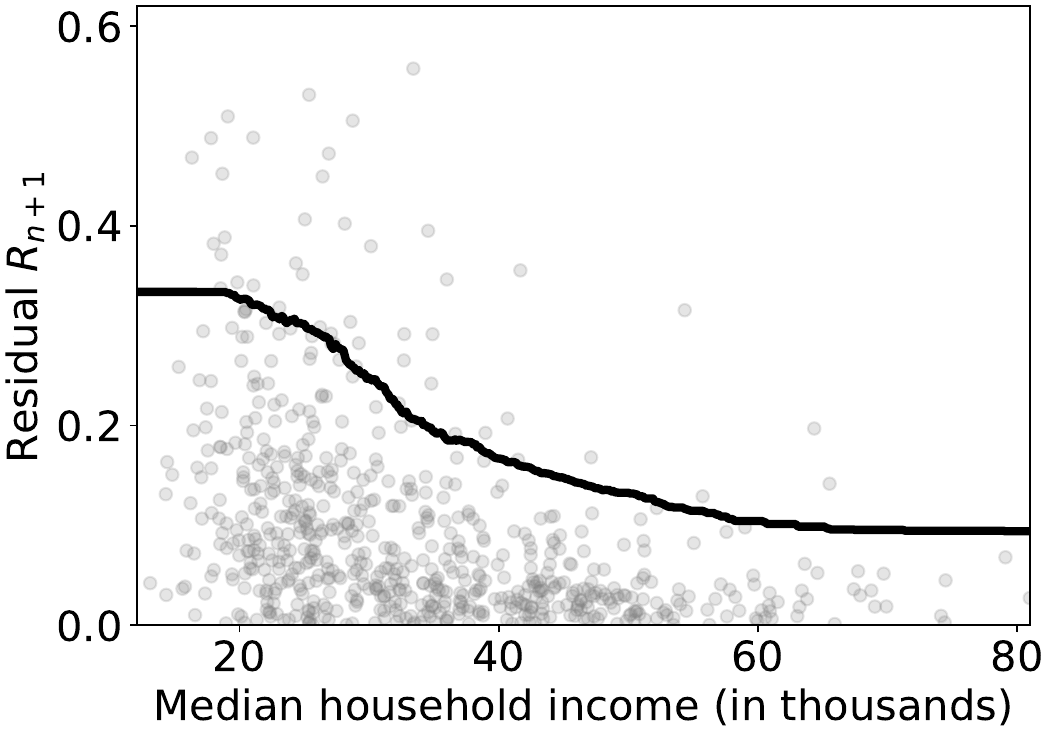}
          \caption{Residuals and PCP quantile.}
          \vspace{15pt}
           \label{fig:g1}
     \end{subfigure}
 \begin{subfigure}[b]{0.45\textwidth}
         \centering
         \includegraphics[width=0.91\textwidth]{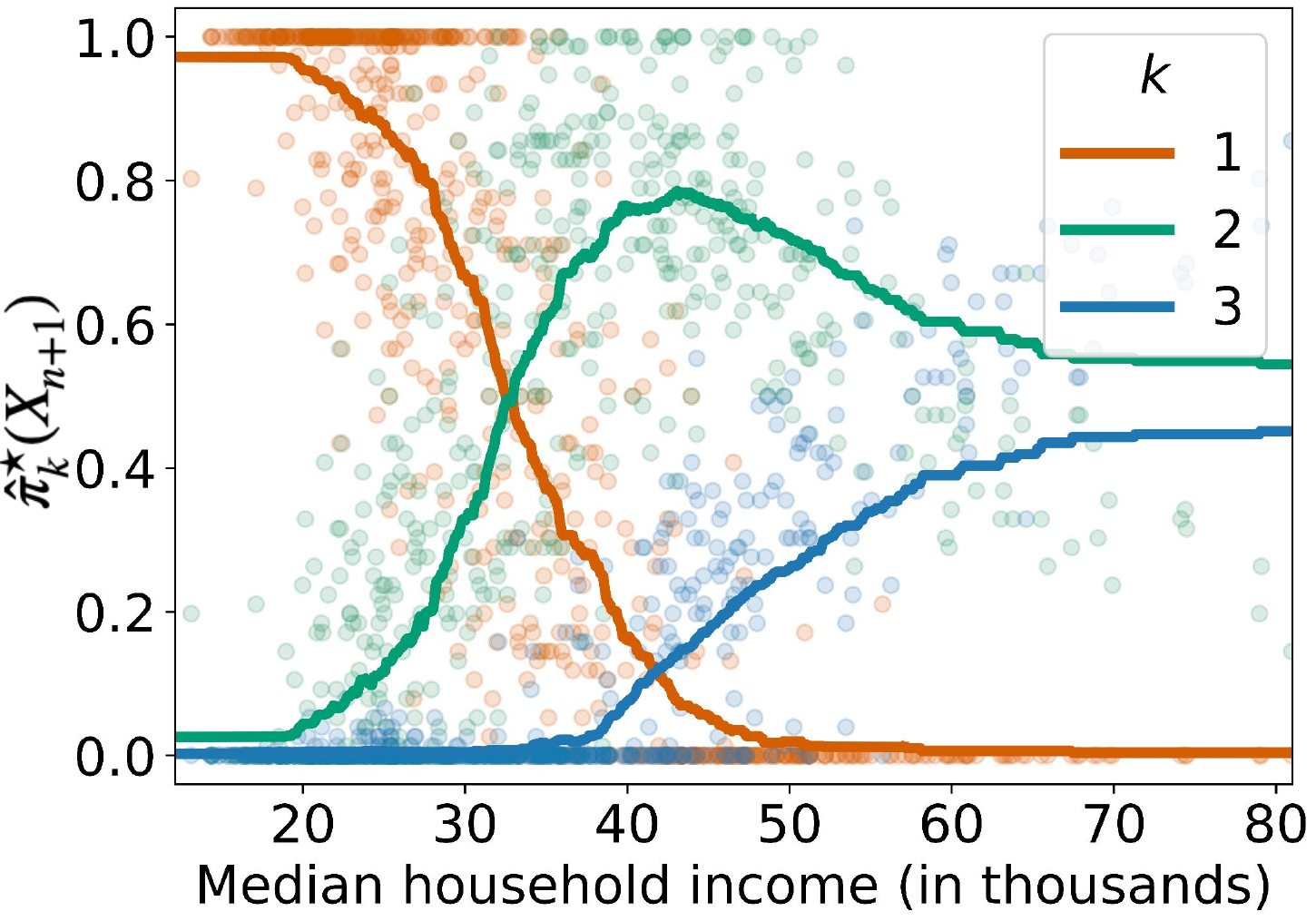}
         \caption{Membership probabilities.}
            \vspace{15pt}
          \label{fig:g2}
            \end{subfigure}    
     \begin{subfigure}[b]{0.45\textwidth}
     \hspace{-15pt}
         \centering
         \includegraphics[width=0.91\textwidth]{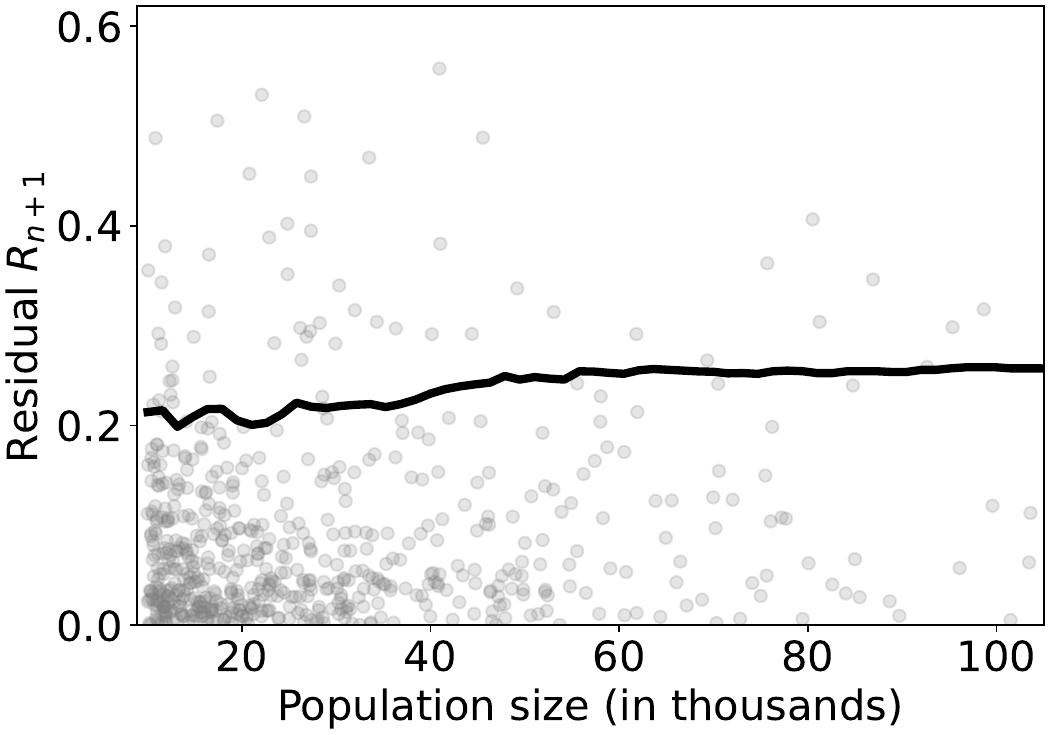}
          \caption{Residuals and PCP quantile.}
          \vspace{5pt}
     \label{fig:g3}
     \end{subfigure}
	\begin{subfigure}[b]{0.45\textwidth}
         \centering
         \includegraphics[width=0.91\textwidth]{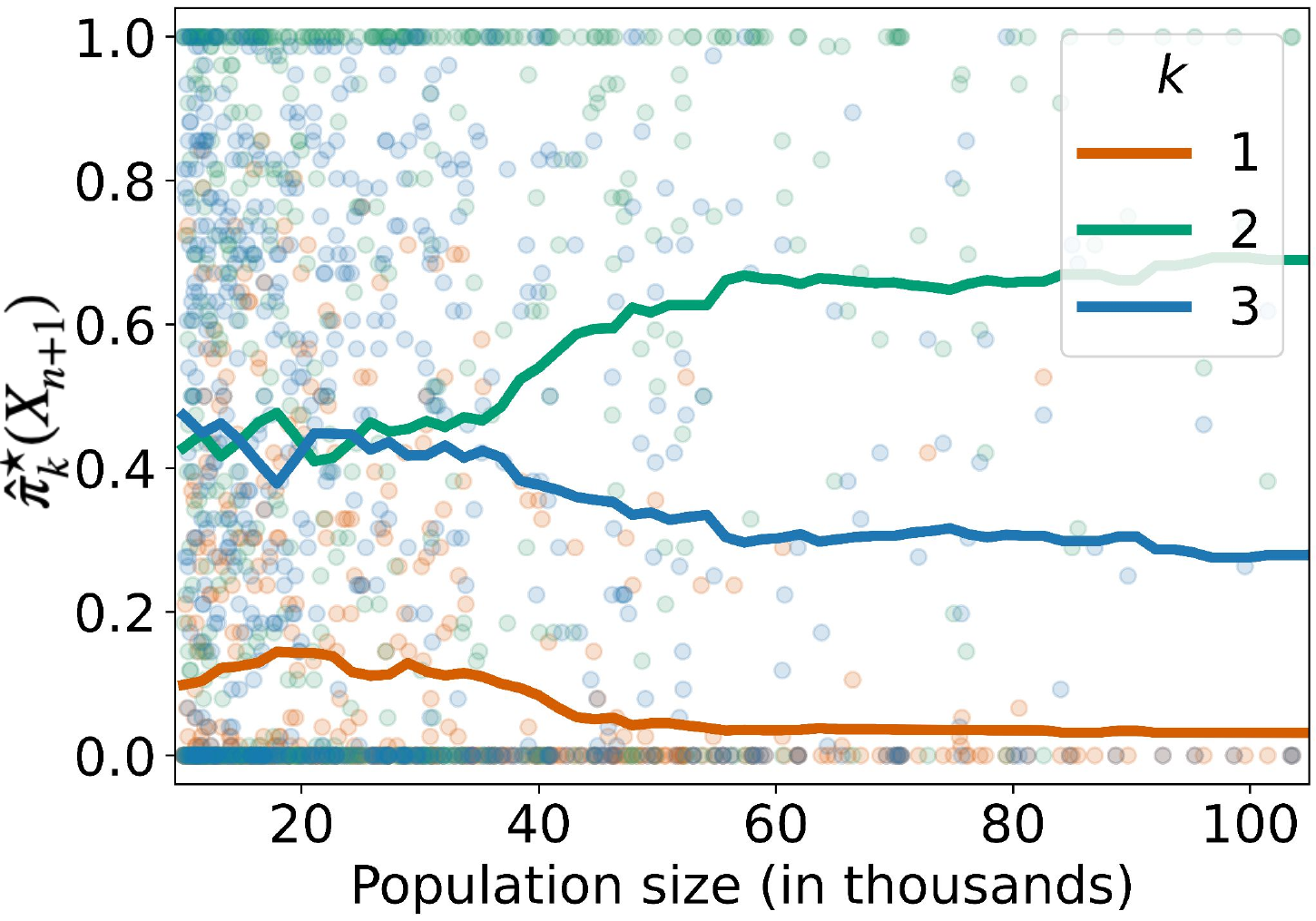}
         \caption{Membership probabilities.}
            \vspace{5pt}
          \label{fig:g4}
            \end{subfigure}
\caption{Illustration of PCP quantiles and membership probabilities over the population size and income features in the Communities and Crime dataset. Each grey dot in panels (a) and (c) represents the residual of a test point. The black curves in the same panels show the local average of the upper residual bounds of PCP, computed over the 100 neighboring test points along the horizontal feature axis, where the upper residual bound is the upper endpoint of the PCP interval for $Y_{n+1}$ minus the prediction $\hat{\mu}(X_{n+1})$.
In panels (b) and (d), each test point contributes three dots, one for each cluster membership probability $\hat \pi_k^\star(X_{n+1})$, $k=1,2,3$, with cluster identity indicated by color. The curves show the corresponding local average probabilities, again computed over the 100 neighboring test points along the horizontal feature axis.
}
\label{fig:guarantee_pi}
\end{figure}

One way to interpret the guarantee in \eqref{equ:pcp_guarantee} is to view the probabilities $\hat \pi^\star(X_{n+1})$ as a ``soft'' prediction of the cluster membership of $X_{n+1}$. 
When $\hat \pi_k^\star(X_{n+1}) \approx 1$ for some $k$, PCP assigns more weight to data points that are likely to belong to cluster $k$. 
When the membership is uncertain and the values of $\hat \pi_k^\star(X_{n+1})$ are nearly equal across $k$, PCP instead upweights data points with similarly uncertain membership. 
This conditioning can serve as a useful proxy for conditioning on $X_{n+1}$ when the residual distribution is well described by  
$R \mid X \sim \sum_{k=1}^{J} \hat \pi_k(X) g_k$ 
for some cluster distributions $g_1,\dotsc,g_J$, as discussed above.

To illustrate, we show further results on the Communities and Crime dataset \citep{dua2017uci} introduced in \Cref{sect:exp_preview}. 
\Cref{fig:g1,fig:g3} show how the test residuals vary across two features, \quotes{median household income} and \quotes{population size}, respectively.
\Cref{fig:g2,fig:g4} depict the membership probabilities for three clusters as a function of the same two features. When the residuals increase in lower-income communities, as shown in \Cref{fig:g1}, the likelihood of belonging to the first cluster also rises with decreasing income; see \Cref{fig:g2}. 
The first cluster primarily consists of low-income communities, while the other two comprise higher-income communities. 
In contrast, the residual distribution is nearly invariant across population size (\Cref{fig:g3}), and consequently, $\hat \pi^\star(X_{n+1})$ only exhibits minor variations across population sizes (\Cref{fig:g4}).  We can thus see that 
$\hat \pi^\star(X_{n+1})$ summarizes the feature information related to the residual distribution while marginalizing out the irrelevant factors. Consequently, PCP has approximately valid coverage conditional on these features.

\Cref{thm:pcp_cover} offers a guarantee based on a mixture model learned from the data. We next assume that the residuals follow a true mixture model \eqref{equ:mixture}, and then bound the conditional coverage gap of our interval, without requiring the number of cluster distributions to be specified correctly. Our bounds are inspired by those in \citet{barber2022conformal}.  In the theorems below, we consider the KL limit of our interval in \eqref{equ:likelihood_2}.

\vspace{8pt}
\begin{theorem}\label{thm:approximate_valid_2}
If the mixture model \eqref{equ:mixture} holds, the PCP interval \eqref{equ:posterior_p} using weights 
$\hat w_{\textup{KL},i} \propto \exp\{- m D_{\textup{KL}}(\hat \pi(X_{n+1})\| \hat \pi(X_{i}))\}
$ for some fixed $\hat \pi$  obeys 
\vspace{-0.1em}
\[
\mathbb P \big\{ Y_{n+1}\not \in  \hat C_{n}^{\textup{PCP}}(X_{n+1}) 	\mid	 X_{1},\dots, X_{n+1} \big\}	
\leq  \alpha  + 2\sum_{i=1}^{n+1} \hat w_{\textup{KL},i} \sum_{k=1}^{K} |\pi_k(X_{i})- \pi_k(X_{n+1})|.
\]
\end{theorem}
As demonstrated in \Cref{fig:guarantee_pi}, PCP fits a mixture model to the residuals, allowing the membership probabilities $\hat \pi$ to identify the key features driving the differences in the conditional distribution of the residuals.
\Cref{thm:approximate_valid_2} shows that the weight $\hat w_{\textup{KL},i}$ used in the PCP interval can reduce the conditional coverage gap when the KL divergence $D_{\textup{KL}}(\hat \pi(X_{n+1})\| \hat \pi(X_{i}))$ increases with the true membership discrepancy $\sum_{k=1}^{K} |\pi_k(X_{i})- \pi_k(X_{n+1})|$; 
that is when data points $X_i$ whose membership probabilities differ from those of the test point $X_{n+1}$ receive smaller weights $\hat w_{\textup{KL},i}$.

\clearpage

\begin{theorem}\label{thm:approximate_valid_3}
In the setting of \Cref{thm:approximate_valid_2}, assume additionally that
$
\pi_k(X_i)>0, \forall i\in[n+1],\ k\in[K],
$
almost surely. Then it holds that
\[
\mathbb P \big\{ Y_{n+1}\not \in  \hat C_{n}^{\textup{PCP}}(X_{n+1}) \mid X_1,\dots, X_{n+1} \big\} 
\leq 
\alpha \sum_{k=1}^{K}
\frac{ \pi_k(X_{n+1}) }
{\sum_{i=1}^{n} \hat w_{\textup{KL},i}\pi_k(X_i)}.
\]
\vspace{-2em}

\end{theorem}
When $\pi$ is sparse, \Cref{thm:approximate_valid_3} yields a tighter upper bound that merely 
depends on the nonzero components of $\pi$. For instance, if $\pi_k{(X_{n+1})} \approx 1$ and $\pi_{k'}{(X_{n+1})} \approx 0$ for $k'\neq k,$ the bound can be written as 
\[\alpha /\sum_{i=1}^{n} \hat w_{\textup{KL},i}\pi_k{(X_{i})}\approx \alpha +\alpha \left[\pi_k{(X_{n+1})} -\sum_{i=1}^{n} \hat w_{\textup{KL},i}\pi_k{(X_{i})} \right]\approx \alpha,
\]
which avoids dependence on all $K$ components of $\pi(X_{n+1})$ as in
\Cref{thm:approximate_valid_2}.

PCP may be less effective when the coverage gaps in the theorems are difficult to minimize. For example, when the predictive model $\hat \mu$ performs unevenly across many regions of the feature space, the conditional residual distribution may only be represented by a complex mixture model \eqref{equ:mixture} with a large number
$K$ of cluster distributions and/or with nonsmooth membership probabilities $\pi(X_i)$.
To improve coverage in such cases, one would need to increase the precision
parameter $m$ or the dimension of $\pi$.
However, doing so increases $\hat w_{\textup{KL},n+1}$ and widens
the intervals.

\section{Empirical performance}\label{sect:empirical}

This section introduces experiments comparing PCP with the other methods from \Cref{fig:wscr}. 
All implementation details can be found in \Cref{sect:D}.

\subsection{Experiments on synthetic data}  

We consider two settings, the first described here and the second in the Appendix. We begin by creating a 6-dimensional feature vector $X$ by sampling each feature from the uniform distribution on $[0,8].$ We let $V$ denote the first feature in $X$. We then define a  mean function 
$f(V) = -3V + V^{2} - 5 V\sin (V)$ and use it to generate the response,
\[
Y = f(V)  + [4+2(V- 2)^2]\epsilon_i, \quad \epsilon_i\sim \mathcal{N}(0,1).
\]
We generate training, validation, and test sets, each consisting of 5000 independent copies of $(X,Y)$. We fit a random forest model $\hat \mu$ on the training set and use it to predict the responses in the validation and test sets. Because $\hat \mu$ accurately approximates $f(V)$, most of the variation in the residual $R = |Y-\hat \mu(X)|$ given $X$ comes from the variance function of $Y$. Consequently, $R$ is approximately distributed according to the mixture model in \eqref{equ:mixture} with a small number of clusters. This is precisely the setting in which the PCP interval can remain finite even when a large precision parameter $m$ is used to reduce the coverage gap in \Cref{thm:approximate_valid_2}.

\Cref{fig:sim_1_coverage} shows the local average coverage rates of four conformal prediction methods over the feature $V$, computed using the 250 nearest test points. PCP is the only method that maintains a coverage rate close to the target level $1-\alpha=0.9$ across the range of $V$. In contrast, SCP constructs intervals without using the feature $V$, and its local coverage rate deviates substantially from 0.9. SCP+CC constructs intervals using a linear quantile regression model for the residuals, and loses coverage for small and large values of $V$, where the model fails to capture the nonlinear changes in the conditional residual distribution. RLCP also suffers coverage loss, suggesting that matching on all features may overlook the residual variation driven by $V$.

\begin{figure}[t] 
	\centering  
         \includegraphics[width=\textwidth]{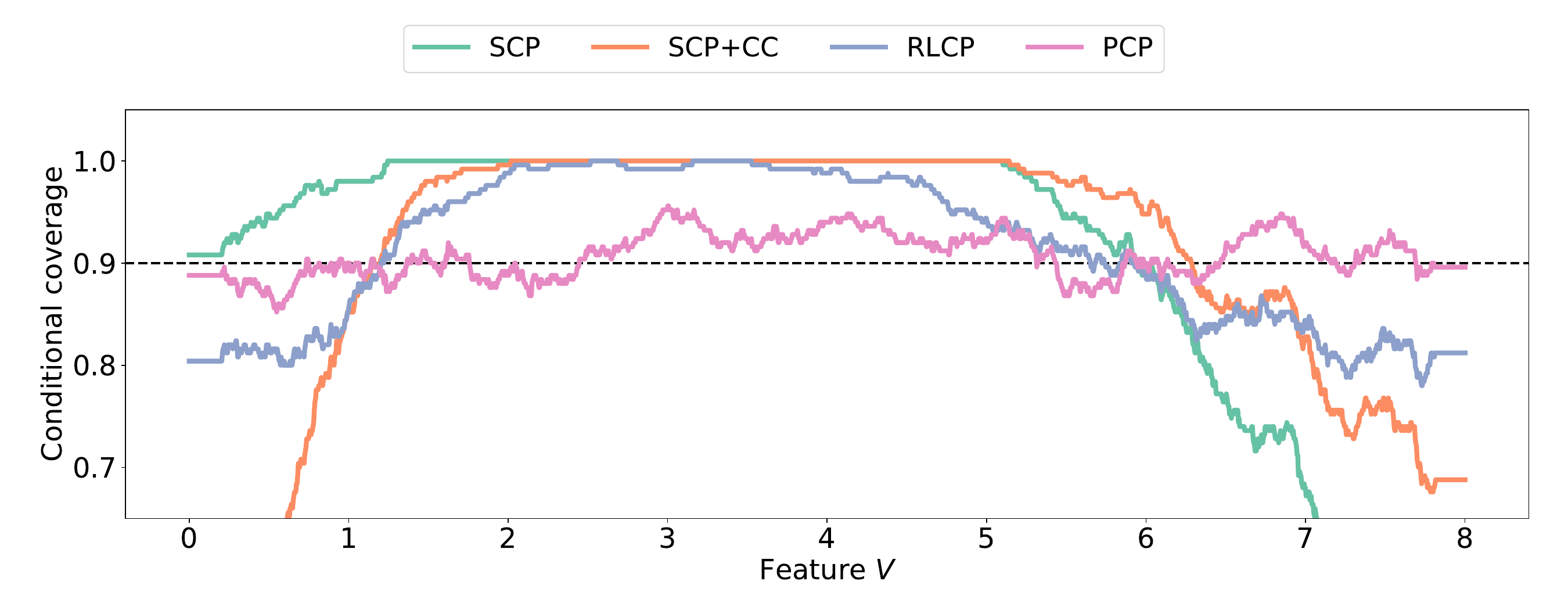}
         	\hspace{-16pt}
\caption{Local average coverage rates of conformal prediction methods in Setting 1.}
\label{fig:sim_1_coverage}
\end{figure}

For a closer look, \Cref{fig:sim_1_interval} visualizes the test responses in grey together with the corresponding prediction intervals. In \Cref{fig:sim_1_RCLP}, the RLCP intervals fluctuate irregularly across adjacent values of $V$, moving repeatedly in and out of the cloud of grey points. By contrast, in \Cref{fig:sim_1_PCP}, the PCP intervals track the local spread of the grey points more closely. This behavior can be attributed to the PCP hyperparameters \((J,m)=(3,276)\) selected by the procedure in \Cref{sect:hyper}.
\Cref{fig:pi_sim1} (the same type of plot as \Cref{fig:g4}) shows that the membership probabilities of PCP vary smoothly over $V$ while fluctuating randomly over the second feature $W$ in $X$. This suggests that the mixture model in PCP identifies $V$ as the main feature governing the conditional residual distribution, allowing PCP to adapt its interval length to the systematic residual variation along $V$. In \Cref{sect:synthetic_data}, we further show that PCP can also handle nonsmooth variations.

\begin{figure}[t]
\vspace{-20pt}
     \centering
  \begin{subfigure}[b]{1\textwidth}      
	\centering 
         \includegraphics[width=0.99\textwidth]{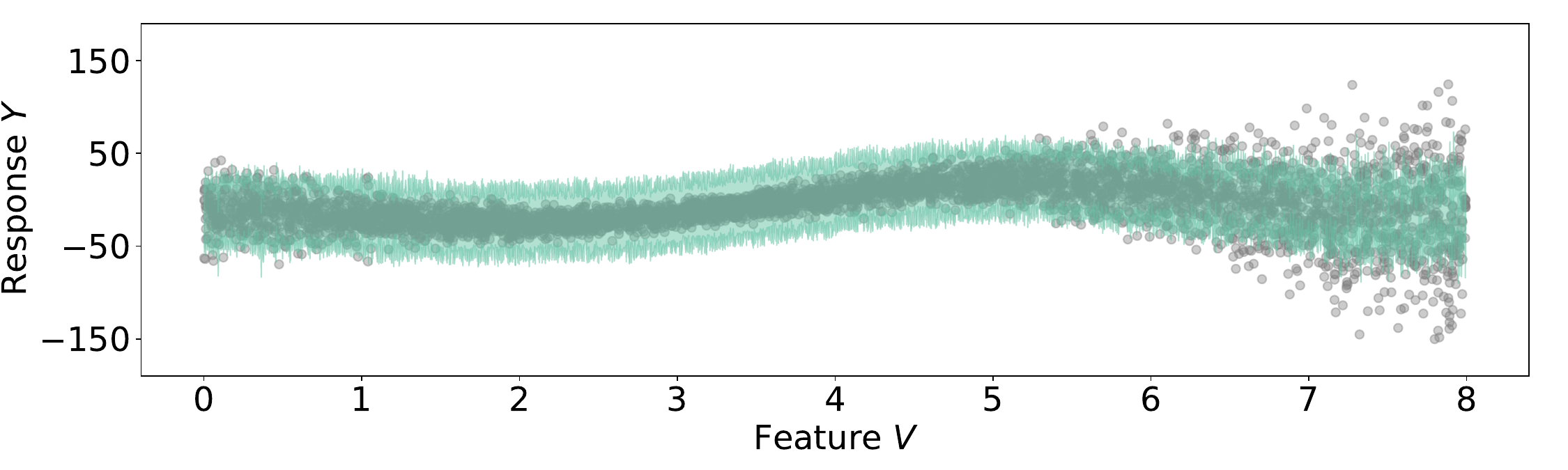}
          \caption{Split conformal prediction (SCP).}
          \vspace{5pt}
          \label{fig:sim_1_SCP}
\end{subfigure}
   \begin{subfigure}[b]{1\textwidth}    
	\centering  
         \includegraphics[width=0.99\textwidth]{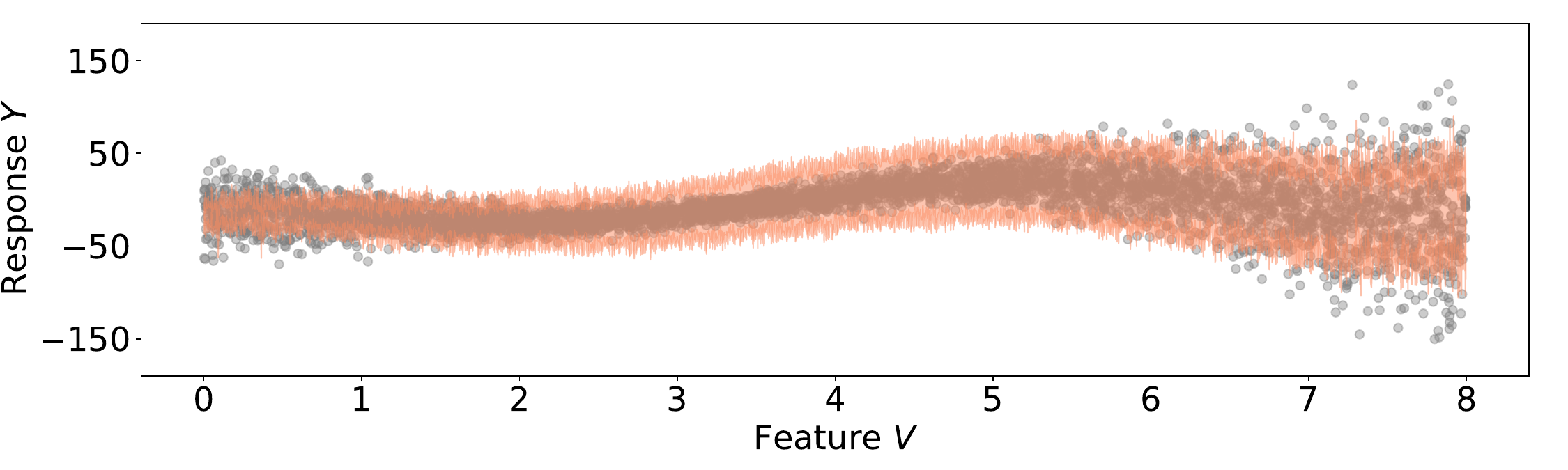}
          \caption{SCP+conditional calibration (SCP+CC).}
                \vspace{5pt}
             \label{fig:sim_1_CC}
	\end{subfigure}	
   \begin{subfigure}[b]{1\textwidth}    
	\centering  
         \includegraphics[width=0.99\textwidth]{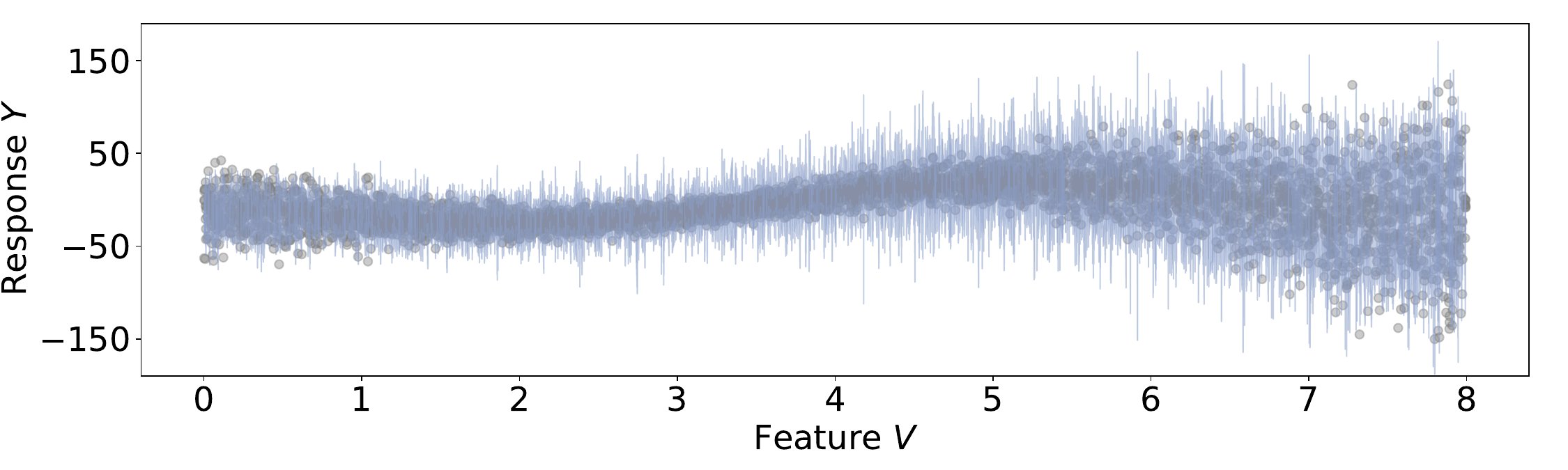}
          \caption{Randomly-localized conformal prediction (RLCP).}
                \vspace{5pt}
             \label{fig:sim_1_RCLP}
	\end{subfigure}	
   \begin{subfigure}[b]{1\textwidth}    
	\centering  
         \includegraphics[width=0.99\textwidth]{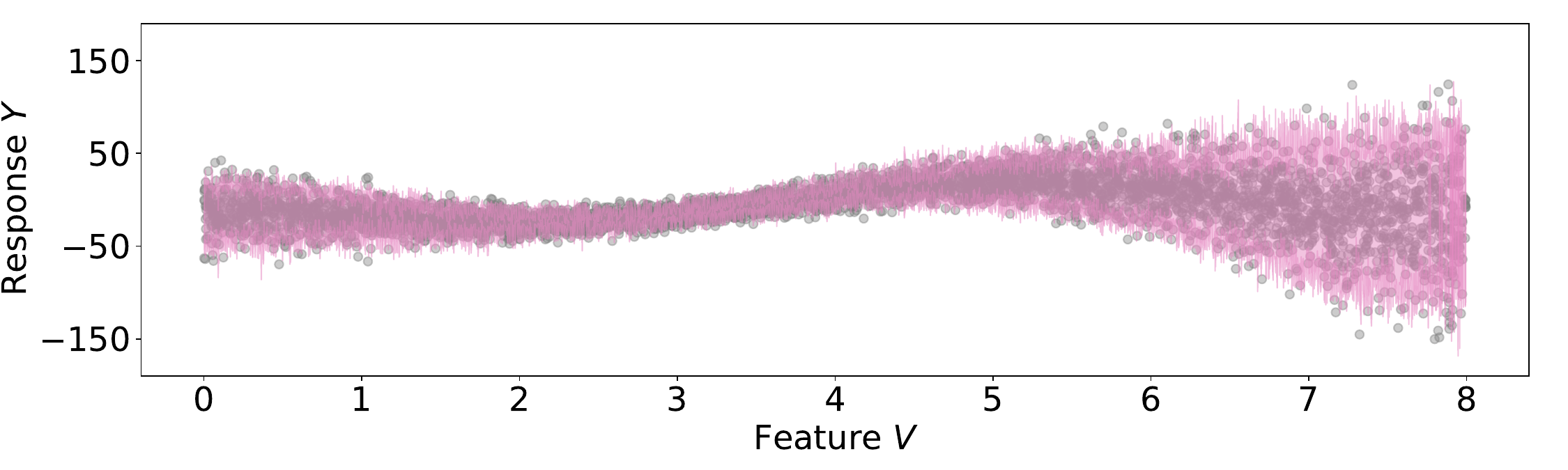}
          \caption{Posterior conformal prediction (PCP).}
                \vspace{5pt}
             \label{fig:sim_1_PCP}
	\end{subfigure}	
\caption{Prediction intervals of conformal prediction methods in Setting 1.}
\label{fig:sim_1_interval}
\end{figure}

\clearpage

\begin{figure}[t]
     \centering
       \vspace{5pt}
  \begin{subfigure}[b]{0.49\textwidth}      
  \hspace{-28pt}
	\centering 
         \includegraphics[width=0.97\textwidth]{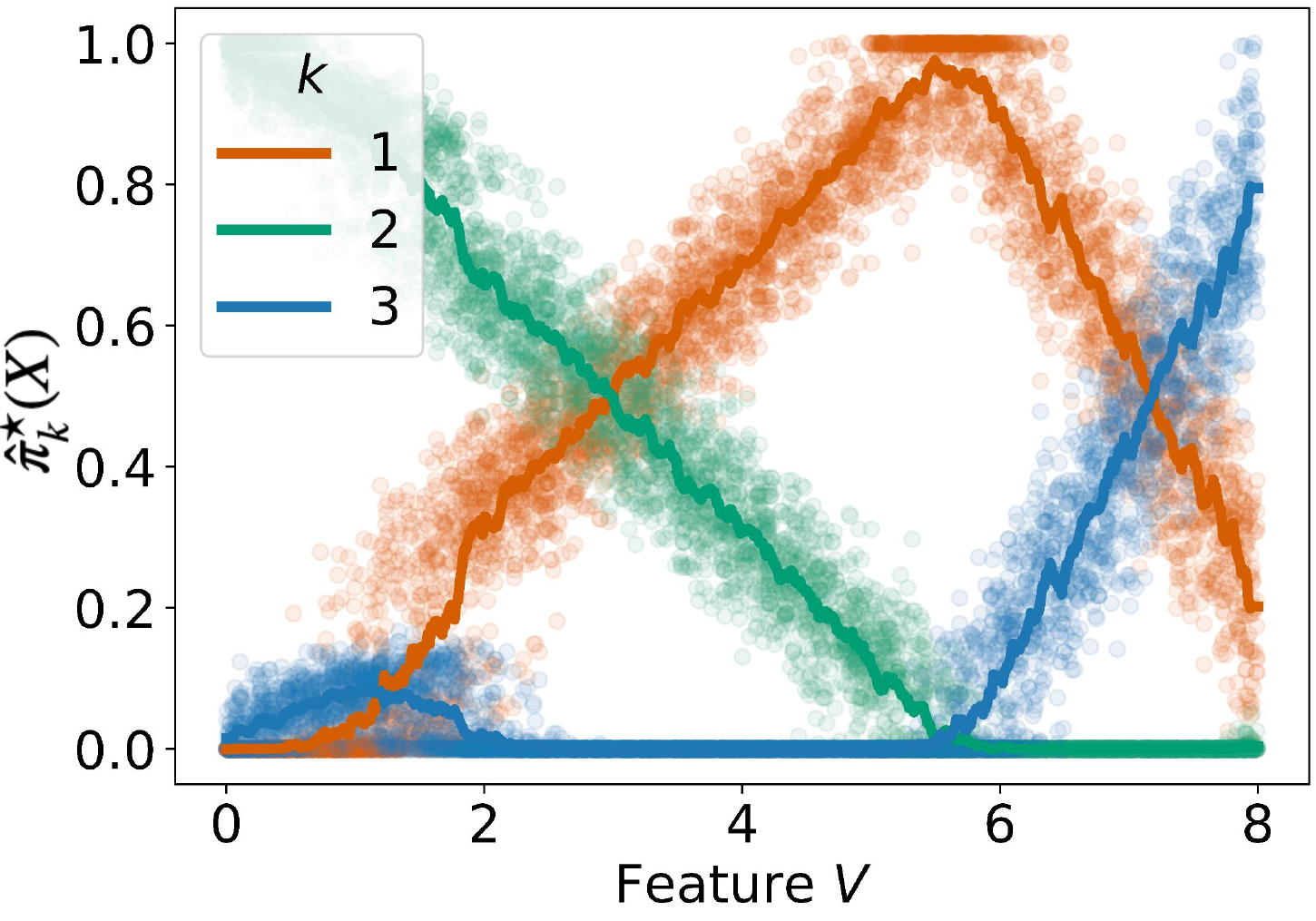}
          \caption{}
          \label{fig:pi_sim_1_1}
\end{subfigure}
  \begin{subfigure}[b]{0.49\textwidth}      
	\centering 
         \includegraphics[width=0.97\textwidth]{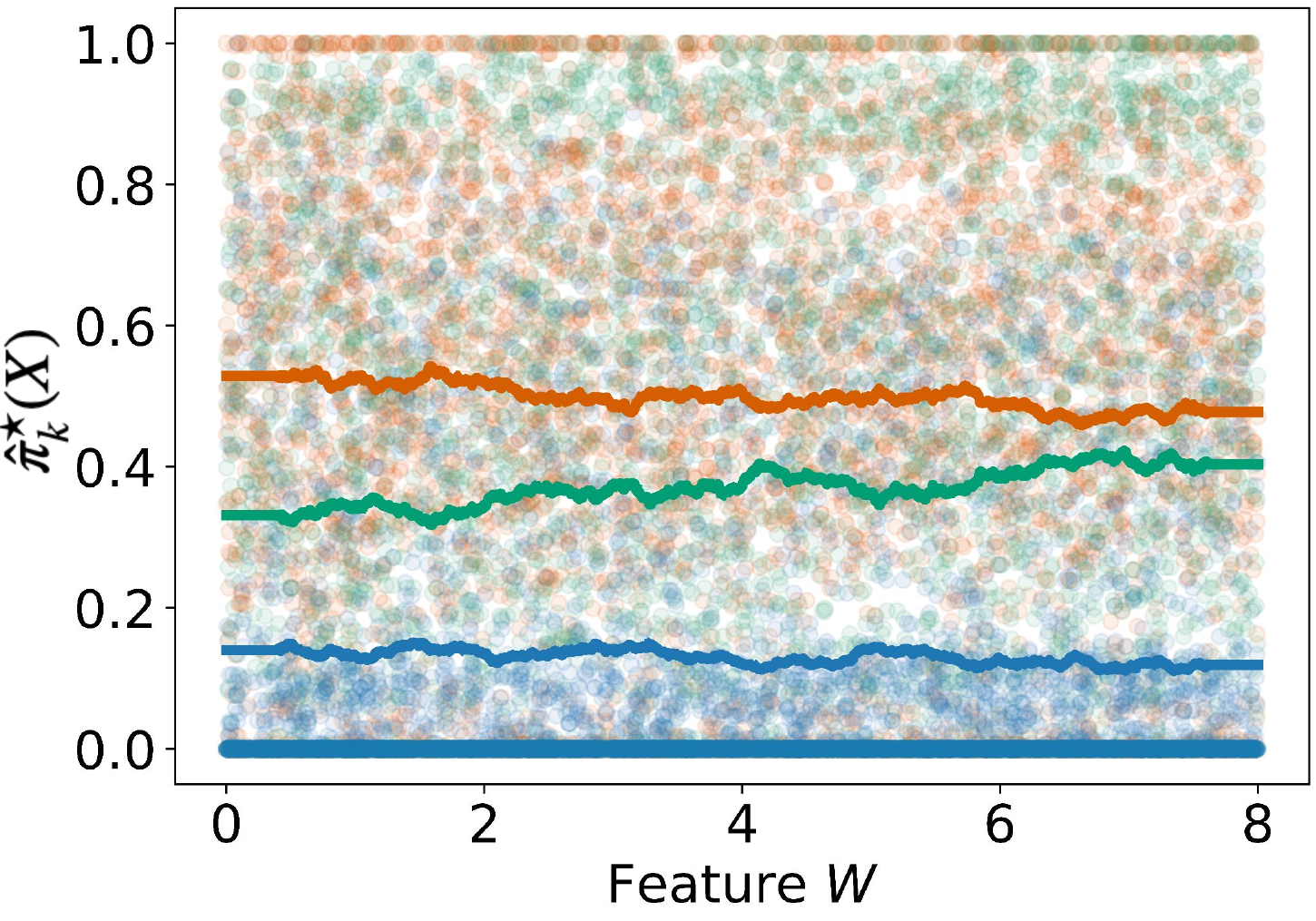}
          \caption{}
          \label{fig:pi_sim_1_2}
\end{subfigure}
\caption{PCP membership probabilities in Setting 1.}
\label{fig:pi_sim1}
\vspace{5pt}
\end{figure}

\subsection{Experiments on real data}

We next discuss experiments on the following two datasets:
\begin{itemize}
\item The online news popularity dataset  \citep{misc_online_news_popularity}  includes 58 features describing the content of articles, such as word counts and the number of images, to predict the number of shares each article receives on social networks.
\item The superconductivity dataset  \citep{misc_superconductivty} contains 81 features about materials and their properties, such as mean atomic mass and radius, to predict the critical temperature (Tc) at which a material becomes superconductive. 
\end{itemize}

We reduce the dimensionality of both datasets to 30 using principal component analysis (PCA) and use ridge regression as the predictive model $\hat \mu$.  We compare the coverage rates and interval length of the conformal prediction methods across 200 independent runs of our experiments, each using a subsample of 2,000 data points.

The results for the popularity dataset are reported in \Cref{fig:wscr_shares}. All methods achieve a marginal coverage rate near the target rate of  0.9. However, the worst-slice conditional coverage rate (WSCR) of SCP falls significantly below 0.9, while the WSCRs of RLCP, SCP+CC, and PCP are approximately 0.9. In terms of interval length, RLCP produces much wider intervals than SCP+CC and PCP. 
Although SCP+CC generates shorter intervals than PCP, it does so at the cost of slightly lower WSCR.

The results for the superconductivity dataset are shown in  \Cref{fig:wscr_conductor}. Here, the WSCR of SCP is slightly below 0.9, leaving a small gap for other methods to fill. 
Among the advanced methods, only PCP can close the gap without significantly increasing interval length. Using a linear quantile regression model, SCP+CC achieves a pre-specified coverage guarantee across all linear functions of the features. In comparison,  PCP improves coverage more effectively by adapting its guarantee to the residuals.

\begin{figure}[t]
     \vspace{-5pt}
     \centering
     \hspace{-5pt}
     \begin{subfigure}[b]{0.329\textwidth}
         \centering         \includegraphics[width=\textwidth]{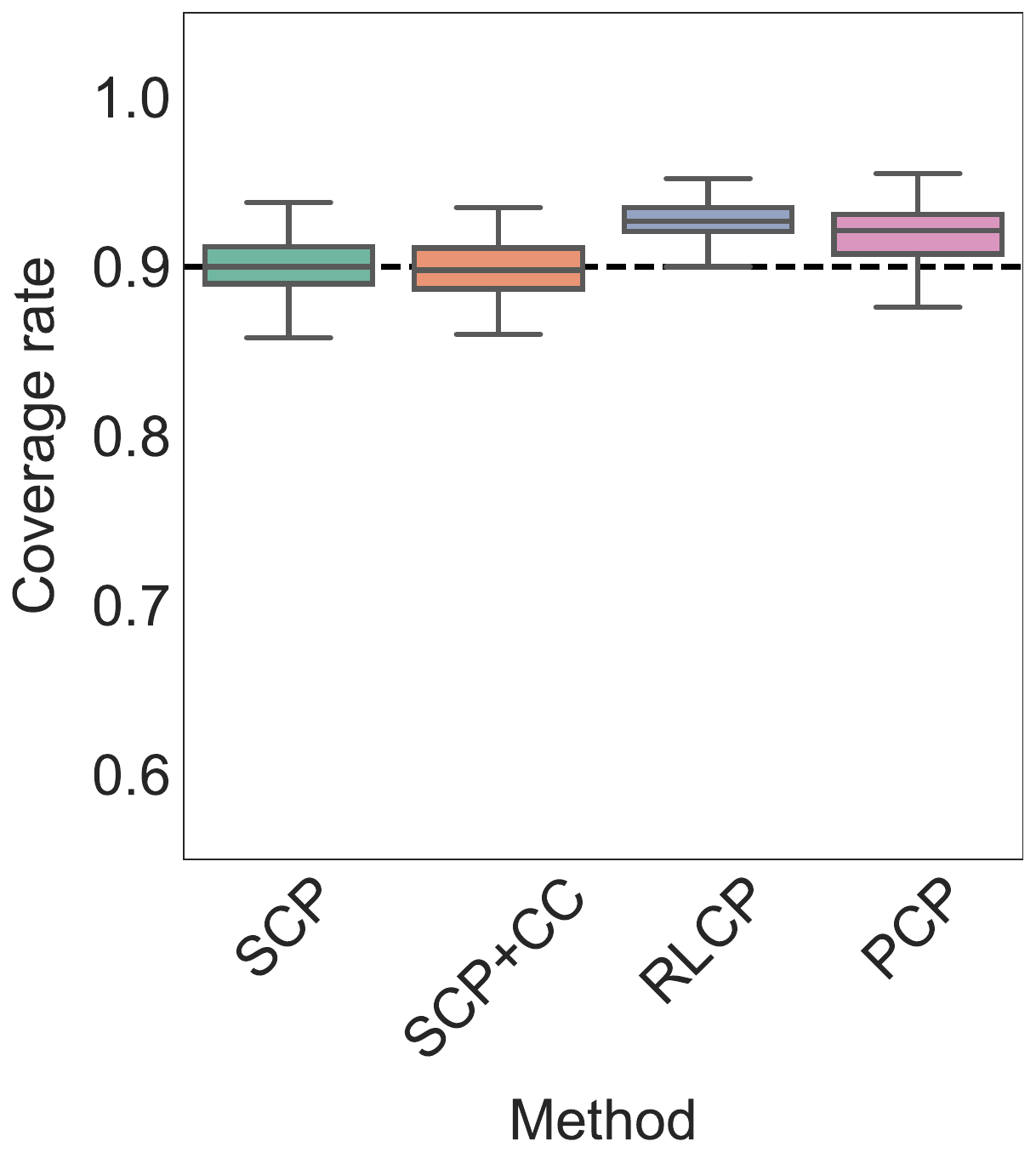}
         \caption{Marginal coverage.}
           \label{fig:wscr1_shares}
     \end{subfigure}
     \begin{subfigure}[b]{0.329\textwidth}
         \centering
         \includegraphics[width=\textwidth]{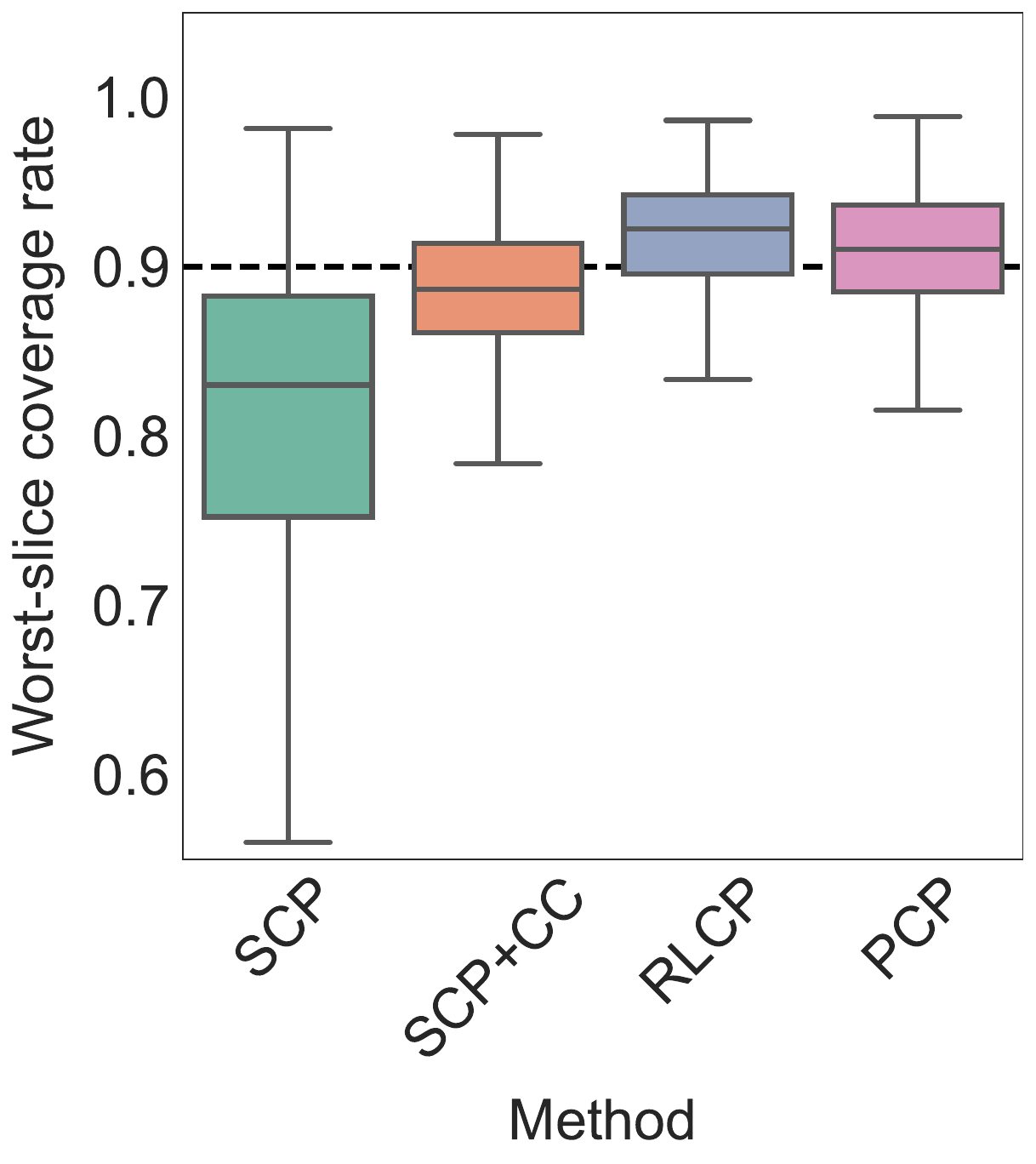}
          \caption{Worst-slice coverage.}
           \label{fig:wscr2_shares}
     \end{subfigure}
          \begin{subfigure}[b]{0.329\textwidth}
         \centering
         \includegraphics[width=\textwidth]{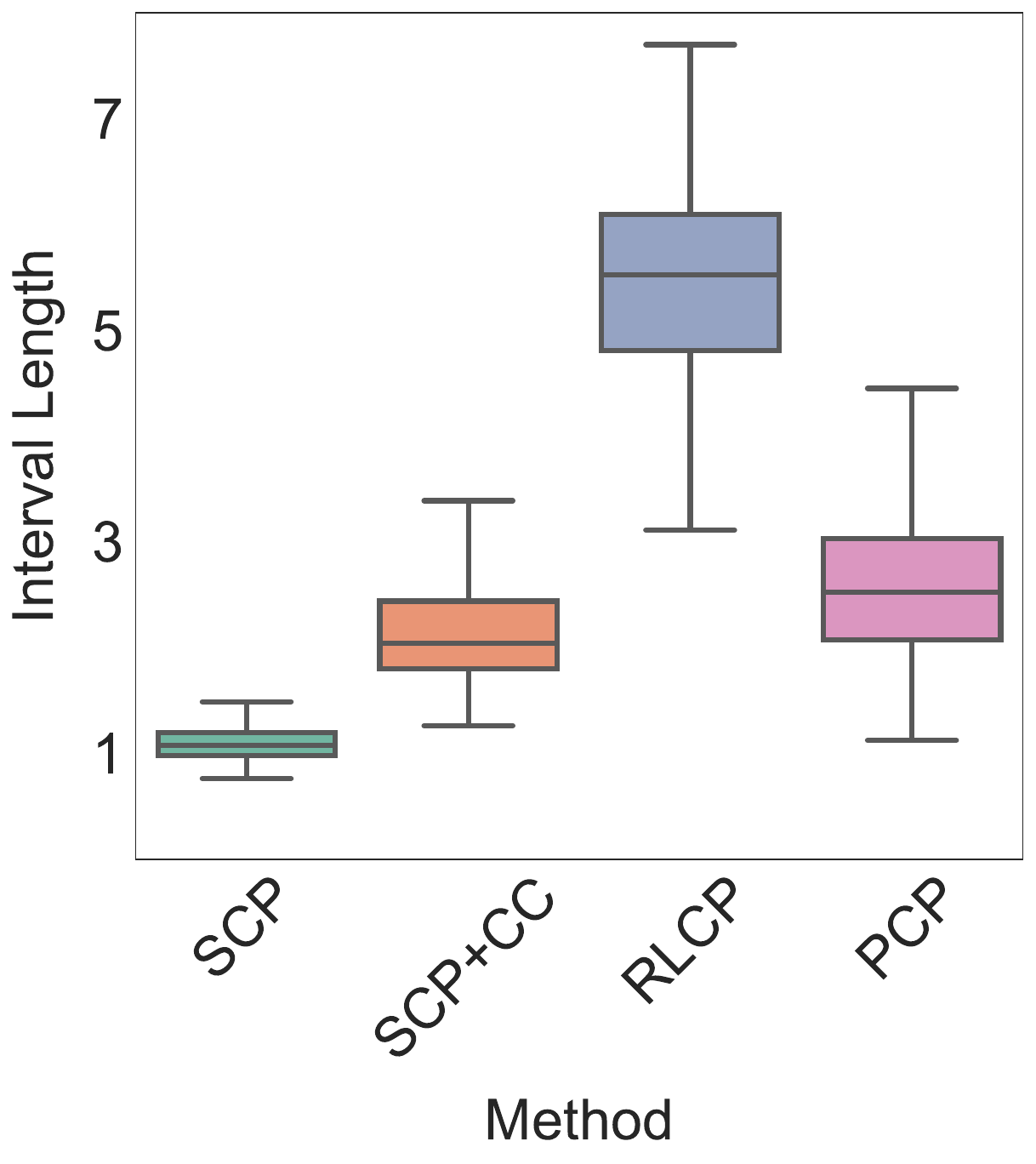}
         \caption{Interval length.}
          \label{fig:wscr3_shares}
            \end{subfigure}
        \caption{Comparison of conformal prediction methods on the popularity data. }
        \label{fig:wscr_shares}
\end{figure}

\begin{figure}[t]
\vspace{3pt}
     \centering
     \hspace{-5pt}
     \begin{subfigure}[b]{0.329\textwidth}
         \centering         \includegraphics[width=\textwidth]{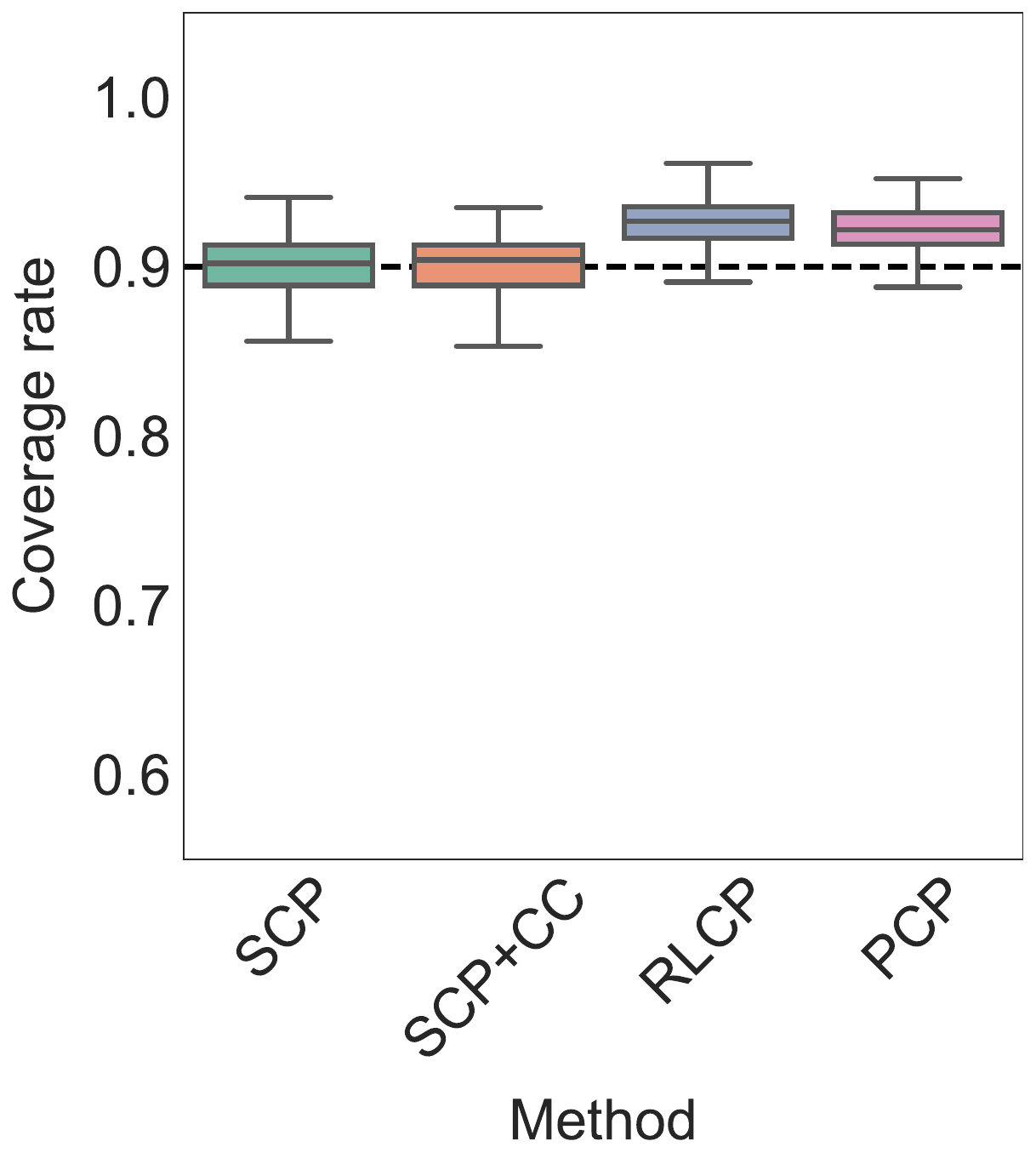}
         \caption{Marginal coverage.}
           \label{fig:wscr1_conductor}
     \end{subfigure}
     \begin{subfigure}[b]{0.329\textwidth}
         \centering
         \includegraphics[width=\textwidth]{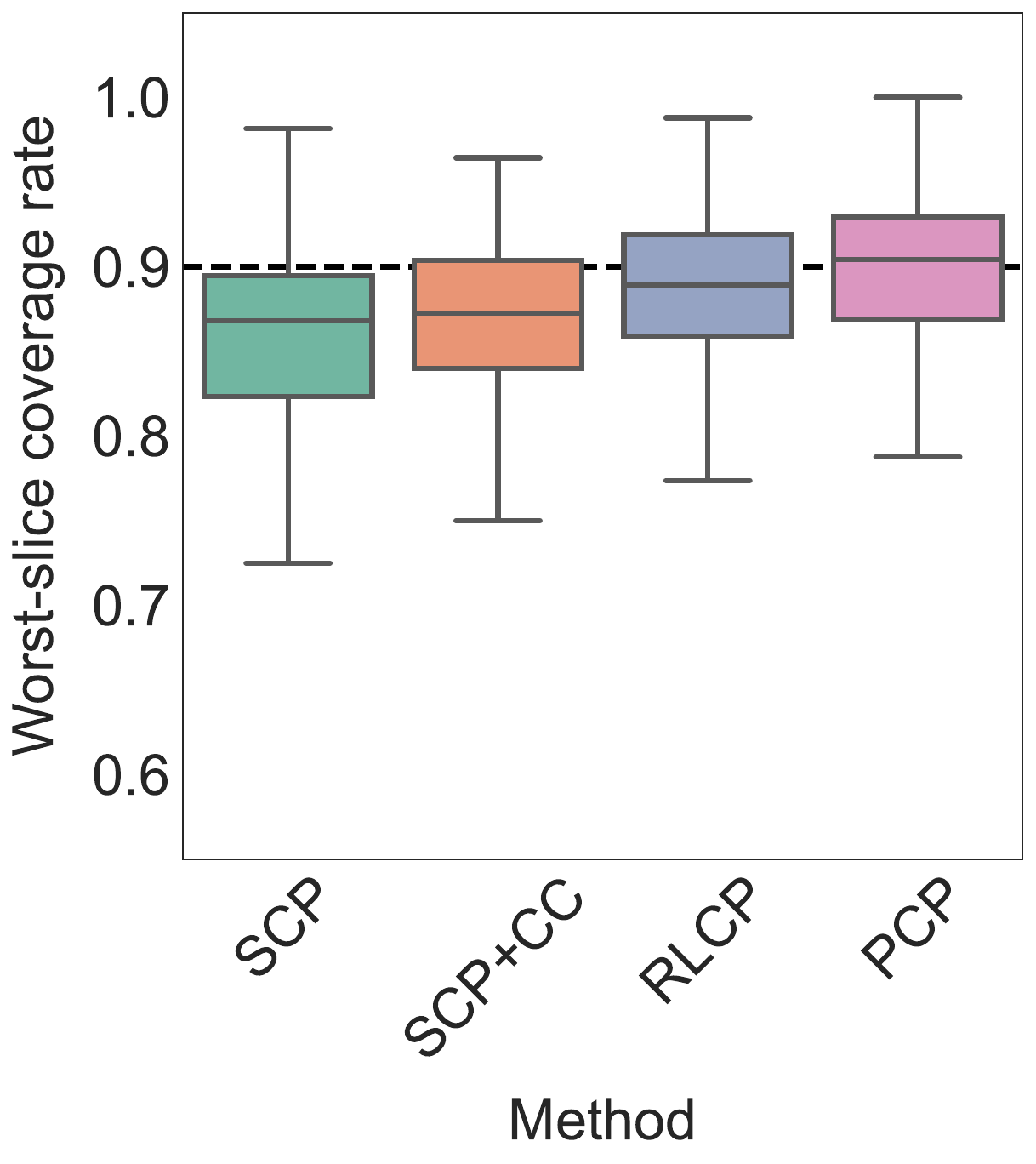}
          \caption{Worst-slice coverage.}
           \label{fig:wscr2_conductor}
     \end{subfigure}
          \begin{subfigure}[b]{0.329\textwidth}
         \centering
         \includegraphics[width=\textwidth]{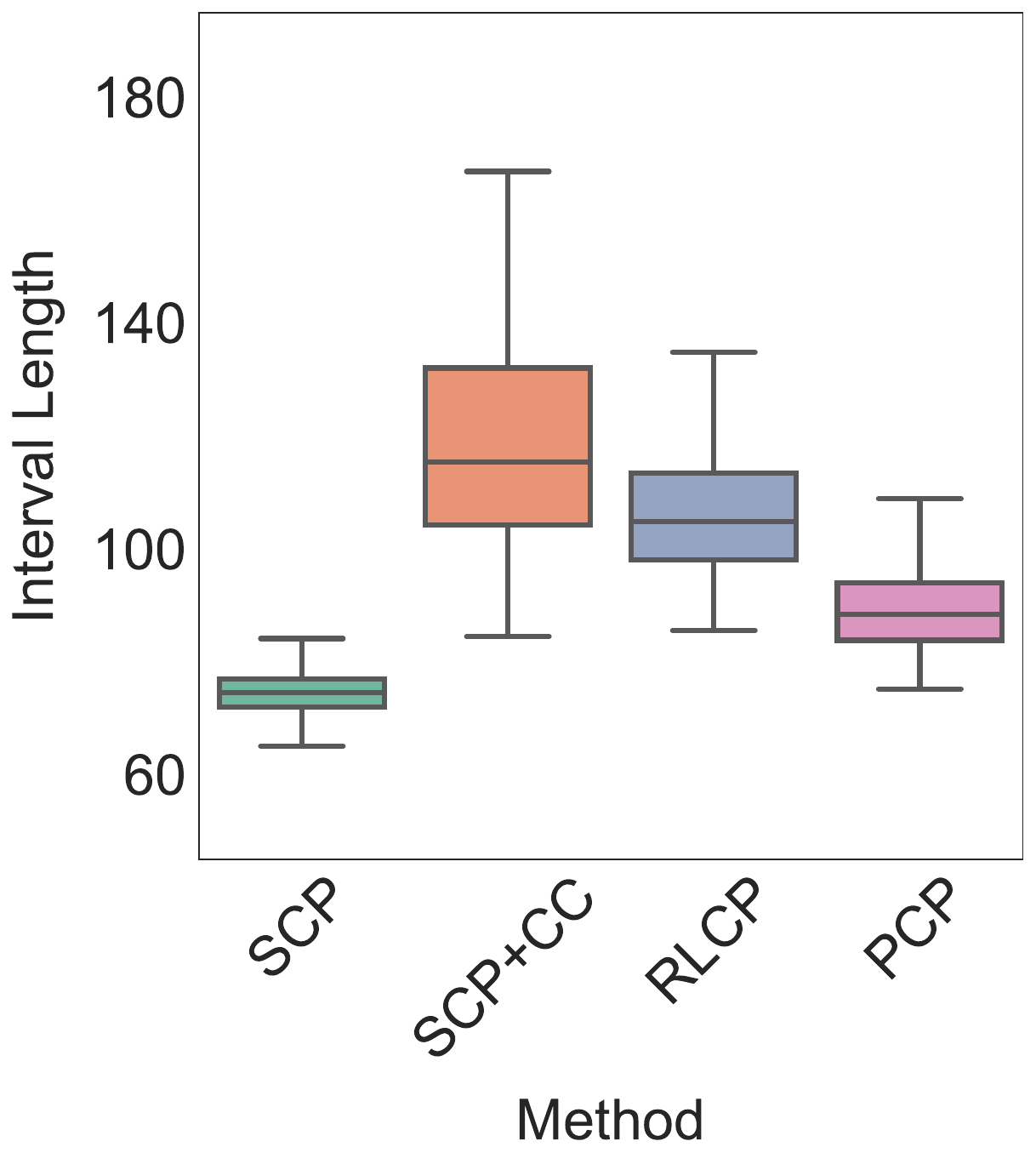}
         \caption{Interval length.}
          \label{fig:wscr3_conductor}
            \end{subfigure}
        \caption{Comparison of conformal prediction methods on the superconductivity data. }
        \label{fig:wscr_conductor}
\end{figure}

\section{Equalized conditional coverage}\label{sect:fair}

In conformal prediction, the partition method \citep{lei2014distribution,romano2020malice} is often used to generate prediction intervals with group-conditional coverage guarantees.  We next show that this partition-based approach may lead to a loss of coverage for some individuals within the subgroups.  In response, we apply PCP to achieve the same group-conditional coverage guarantees while making the interval nearly independent of the test point's subgroup membership.

\begin{figure}[t]
\hspace{-10pt}
     \centering
     \begin{subfigure}[b]{0.44\textwidth}
         \centering
         \includegraphics[width=0.9\textwidth]{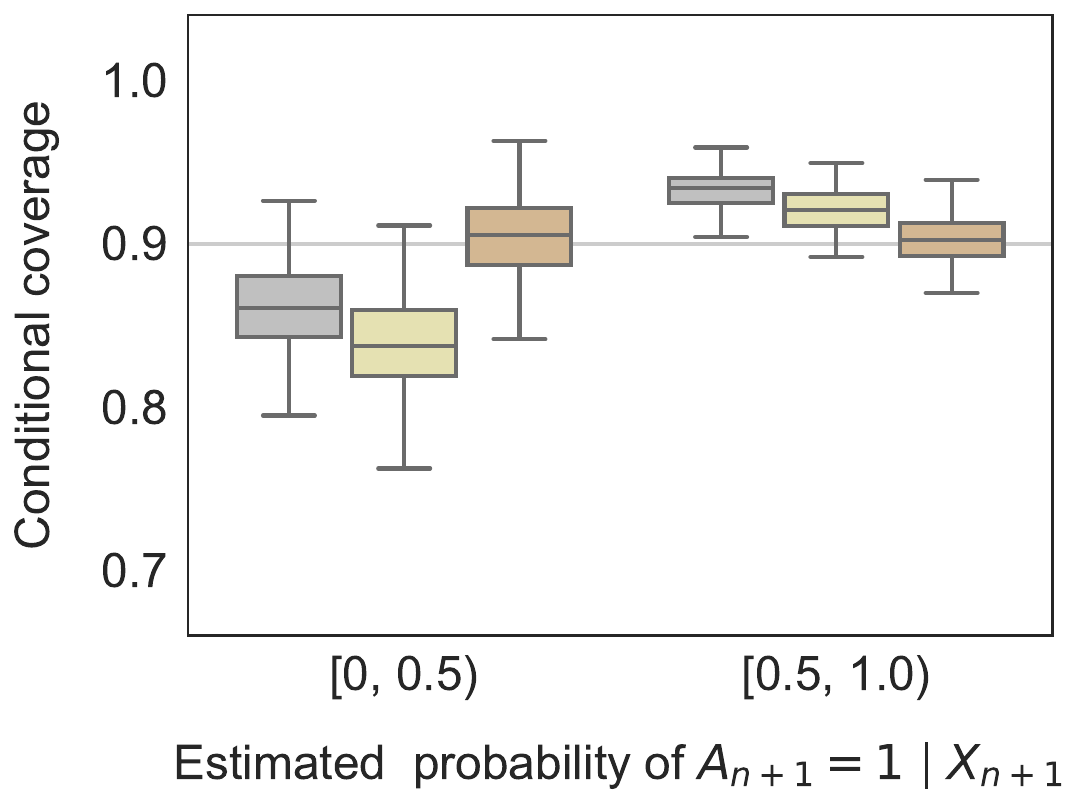}
          \caption{
           \begin{minipage}[t]{0.85\textwidth}
           Coverage rate given $A_{n+1}=1$  and 
             $\hat e^\star(X_{n+1}) <$ or $\geq 0.5$.   
        \end{minipage}
          }
           \label{fig:fair1}
     \end{subfigure}
      \vspace{10pt}
\hspace{5pt}
 \begin{subfigure}[b]{0.44\textwidth}
         \centering
         \includegraphics[width=0.9\textwidth]{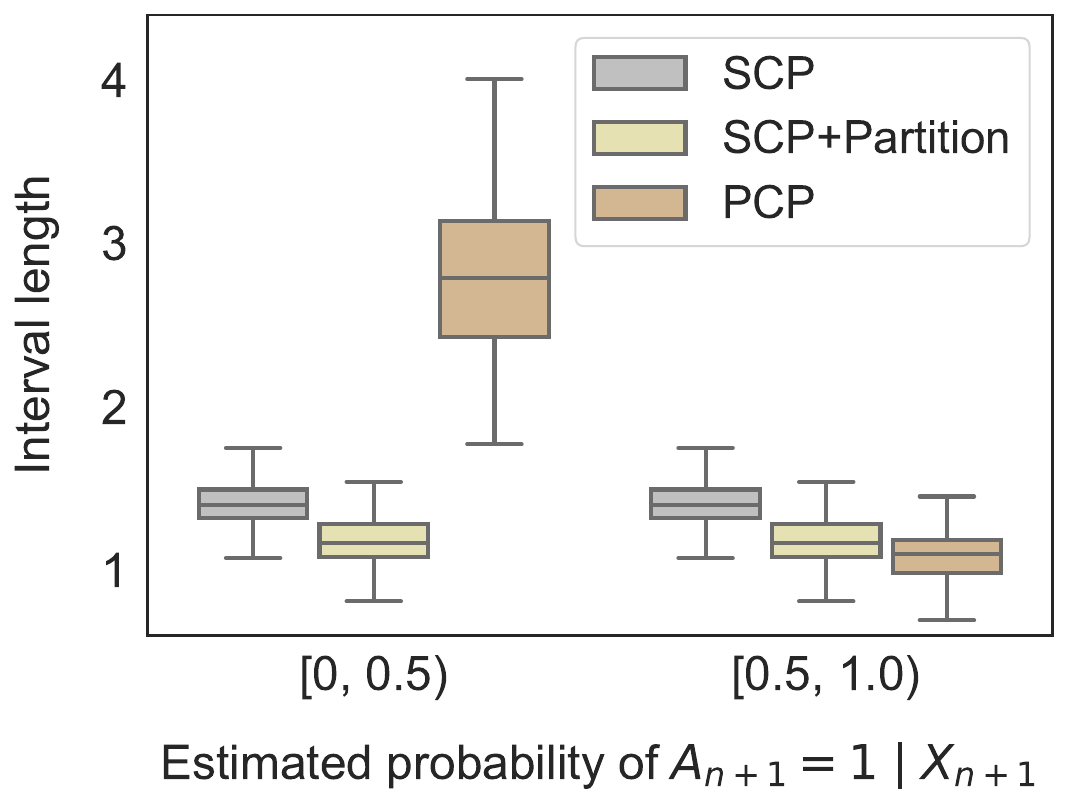}
         \caption{
             \begin{minipage}[t]{0.85\textwidth}
            Interval length given $A_{n+1}=1$  and 
              $\hat e^\star(X_{n+1}) <$ or $\geq 0.5$.   
        \end{minipage}
         }
          \label{fig:fair2}
            \end{subfigure}
       \vspace{10pt}
\hspace{-10pt}
     \begin{subfigure}[b]{0.44\textwidth}
         \centering
         \includegraphics[width=0.9\textwidth]{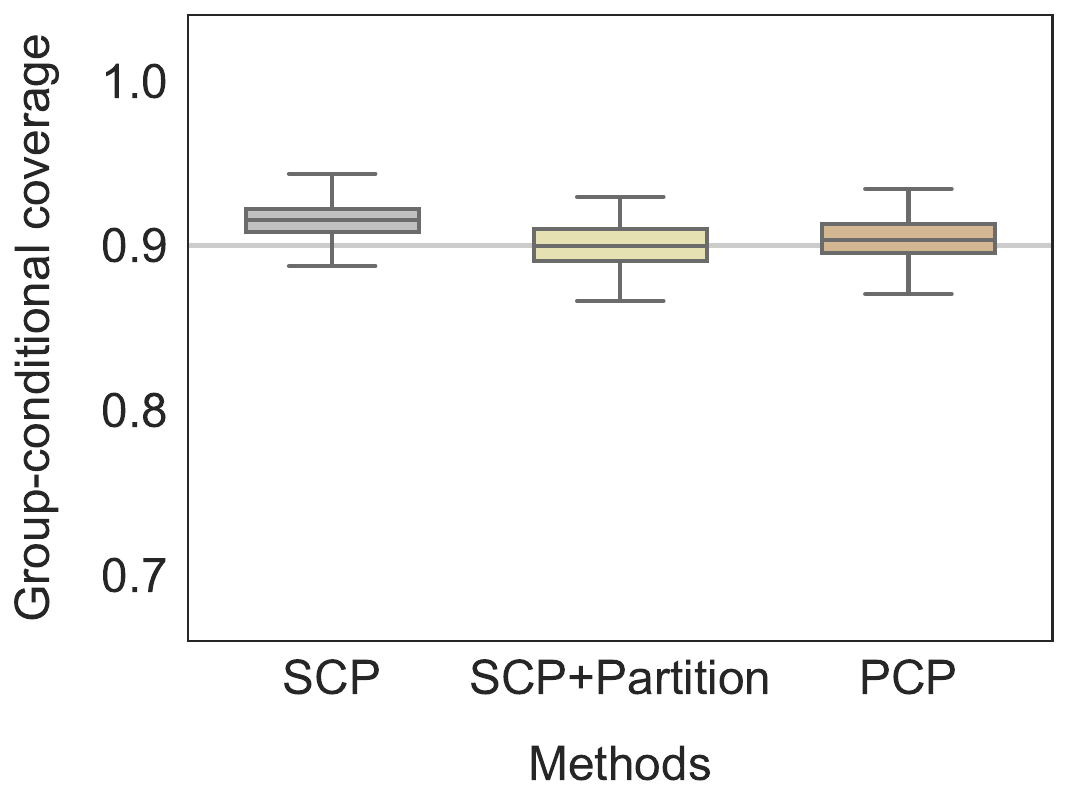}
          \caption{           \begin{minipage}[t]{0.85\textwidth}
          Group-conditional coverage rate  given $A_{n+1}=1$  (female).     
           \end{minipage}          
        } 
     \label{fig:fair3}
     \end{subfigure}
    \hspace{5pt}
	\begin{subfigure}[b]{0.44\textwidth}
         \centering
         \includegraphics[width=0.9\textwidth]{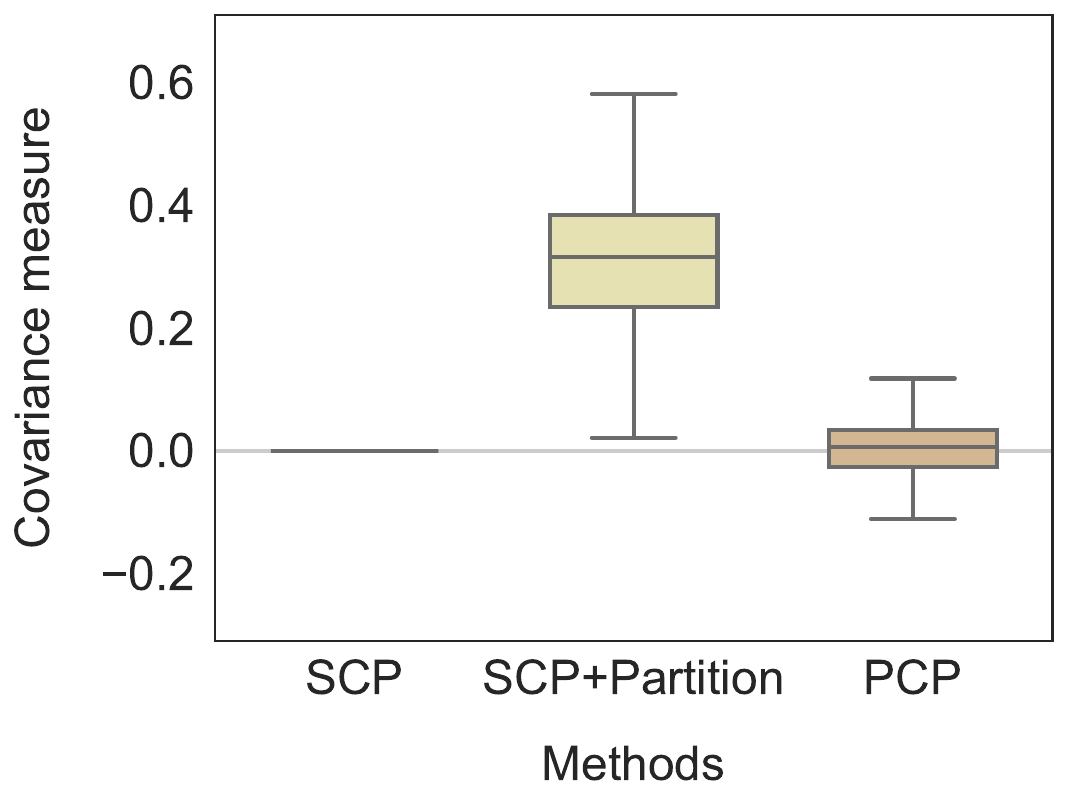}
          \caption{           \begin{minipage}[t]{0.85\textwidth}
         Measure of dependence between intervals and  $A_{n+1}$ given $X_{n+1}$. 
        \end{minipage}          
        } 
          \label{fig:fair4}
            \end{subfigure}
                \vspace{-5pt}
        \caption{Comparison of conformal prediction methods on the Medical Expenditure Panel Survey (MEPS) 19. The results are from 200 runs of the experiments. 
        {As discussed in the text, we divide the female group (\(A_{n+1}=1\)) into two subgroups: a minority subgroup with \(\hat e^\star(X_{n+1})<0.5\), and a majority subgroup with \(\hat e^\star(X_{n+1})\ge 0.5\).}
        }
        \label{fig:fair_gender}
\end{figure}

To begin with a concrete example, imagine that we have a medical dataset from a study of a disease where male patients are predominant. In this dataset, it is crucial to obtain a coverage guarantee regardless of gender $A_{n+1}\in \{0,1\}$. This is far from automatic since the predictive model may underperform for females due to limited training data. In this situation, we often run SCP for males and females separately. 
Specifically,  we let $\mathcal I_n= \{i\in [n]:A_{i} = A_{n+1} \}$ be all the individuals who have the same gender as the test point.
For this partition interval, let \(R_i=|Y_i-\hat\mu(X_i,A_i)|\).
Running SCP solely on $Z_{\mathcal I_n}$ produces the interval,
\[
\hat C_{n}^{\text{Partition}}(X_{n+1},A_{n+1}) = \left[\hat \mu(X_{n+1},A_{n+1}) \pm Q_{1-\alpha}\left(\sum_{i\in \mathcal I_n}\frac{1}{|\mathcal I_n|+1} \delta_{R_i}	+ \frac{1}{|\mathcal I_n|+1} \delta_{+\infty} \right)		\right],
\]
which satisfies the group-conditional coverage guarantee,
\begin{equation}\label{equ:weak_coverage}
\mathbb P \big\{ {Y_{n+1}\not\in  \hat C_{n}^{\text{Partition}}	(X_{n+1},A_{n+1}) }	\mid A_{n+1} \big\}\leq \alpha.
\end{equation}
Despite the guarantee in \eqref{equ:weak_coverage}, generating intervals based on a partition may decrease the coverage rate for some minorities within the subgroups. This general phenomenon can be illustrated using the probabilities of $A =1$  (female),
\[
\theta = \mathbb P\{A =1\}, \qquad e(X) = \mathbb P\{A =1 \mid X\}.
\] 
Assuming  $\theta\in (0,1)$, it holds that
\begin{equation}\label{equ:p_difference}
\underbrace{\mathbb P\{e(X)\ge \theta \mid A=1\}}_{\text{Female population}}
-
\underbrace{\mathbb P\{e(X)\ge \theta \}}_{\text{Full population}}
=
\mathbb E\!\left[\one\{e(X)\ge \theta\}\left(e(X)/\theta-1\right)\right]
\ge 0.
\end{equation}
\Cref{equ:p_difference} highlights that the female population $(A=1)$ is more likely to have $e(X)\geq \theta $ than the full population. In other words, a test point with $A_{n+1}=1$ but $e(X_{n+1})< \theta $ is less likely in the female population than in the full population.  For instance, females who are taller than the average height of the full population are even more underrepresented within the female population than in the full population. Consequently, the SCP+Partition interval  $\hat C_{n}^{\text{Partition}}(X_{n+1},A_{n+1})$ defined above will achieve a worse coverage rate for tall females, compared to the original SCP interval $\hat C_n^{\text{SCP}}(X_{n+1}) $ for $Y_{n+1}$ in \eqref{equ:interval_scp} computed using all $n$ validation samples.

To illustrate, we compare the methods on the Medical Expenditure Panel Survey (MEPS) 2019 dataset \citep{romano2019conformalized,romano2020malice}, provided by the Agency for Healthcare Research and Quality. In this experiment, we predict individuals' utilization of medical services from personal information such as age, race, poverty status, and health status. We randomly draw 6000 individuals from the MEPS dataset and divide them into three folds. We fit two random forests,   \(\hat \mu(\cdot)\) and \(\hat e(\cdot)\), using the first fold. The female proportion \(\theta\) in \eqref{equ:p_difference} is approximately \(0.5\) in the MEPS dataset. We further divide the female group (\(A_{n+1}=1\)) in the test set into two subgroups according to whether the randomized estimator \(\hat e^\star(X_{n+1})\) defined in \eqref{equ:e_tilde}  is greater than \(0.5\). \Cref{fig:fair1} shows that SCP+Partition achieves a lower coverage rate than the original SCP in the subgroup with \(A_{n+1}=1\) and \(\hat e^\star(X_{n+1})<\theta=0.5\), which is a minority subgroup within the female population.\footnote{In \Cref{sect:repeat_exp}, we repeat the same experiment, with $A_{n+1}$  indicating whether a person is white or non-white. The results there are consistent with the ones here.}

We apply PCP to remedy the coverage loss of SCP+Partition. We sample $  \hat A_{1}^\star ,\dots, \hat A_m^\star\stackrel{\text{i.i.d.} }{\sim }  \text{Ber}(\hat e(X_{n+1}) )$, and use them to define a randomized version of $\hat e(X_{n+1})$:
\begin{equation}\label{equ:e_tilde}
\hat e^\star(X_{n+1}) =  \hat L^\star/m,
\qquad
\hat  L^\star = \sum_{s=1}^m  \hat A_s^\star.
\end{equation}
After excluding gender from \(\hat\mu\), we redefine \(R_i=|Y_i-\hat\mu(X_i)|\) and modify the SCP+Partition interval as
\begin{equation}\label{equ:pcp_fair}
\hat C_n^{\textup{PCP}}(X_{n+1},A_{n+1}) = \left[\hat \mu(X_{n+1}) \pm Q_{1-\alpha}\left(\sum_{i\in \mathcal I_n}\hat w_{i}^\star \delta_{R_i}	+\hat w_{n+1}^\star \delta_{+\infty} \right)		\right],
\end{equation}
where \(\hat w_i^\star\propto
[\hat e(X_i)]^{\hat L^\star}[1-\hat e(X_i)]^{m-\hat L^\star}\),
with the normalization taken over \(i\in\mathcal I_n\cup\{n+1\}\). This interval satisfies an augmented group-conditional guarantee:
\begin{proposition}\label{prop:fairness}
\vspace{5pt}
Suppose that the data points $(X_i,A_i, Y_i), i\in [n+1]$, are i.i.d. Then the prediction interval in  \eqref{equ:pcp_fair} satisfies
\begin{equation}\label{equ:c_fairness_0}
\mathbb P \big\{ {Y_{n+1}\not\in  \hat C_{n}^{\textup{PCP}}(X_{n+1},A_{n+1}) }			\mid  A_{n+1}, \hat e^\star(X_{n+1}) \big\}\leq \alpha.
\end{equation}
Assume that, with probability 1, \(e(x)\in(0,1)\) for all \(x\in\mathcal X\), and \(\sup_{x\in\mathcal X}|\hat e(x)-e(x)|\to 0\) as \(n\to\infty\). Then, as $n\rightarrow\infty$ and $m(n)\rightarrow\infty,$ we have
\begin{equation}\label{equ:c_fairness}
\hat C_{n}^{\textup{PCP}}(X_{n+1},A_{n+1}) \independent A_{n+1} \mid X_{n+1}.
\end{equation}
\end{proposition} 
The guarantee in \eqref{equ:c_fairness_0} follows from \Cref{thm:validity_1}. Imagine that a female group mainly consists of individuals with large values of $ \hat e^\star(X_{n+1})$, while minorities in the group, such as tall females in the example above, have smaller values of $ \hat e^\star(X_{n+1})$. The smaller $ \hat e^\star(X_{n+1})$ is, the less represented an individual is within the female group. In this sense, the guarantee in \eqref{equ:c_fairness_0} means that PCP has valid coverage across individuals regardless of \emph{how represented they are in their gender group}.

This guarantee is demonstrated through the results in \Cref{fig:fair1}, where PCP has coverage of approximately 0.9 for both subgroups. We also see from \Cref{fig:fair2} that the interval length of PCP is highly adaptive: it increases significantly to close the coverage gap of SCP+Partition for the underrepresented individuals with $\hat e^\star(X_{n+1})<0.5$, while decreasing for the well-represented individuals with $\hat e^\star(X_{n+1})\geq0.5$, where $\hat \mu$ is relatively accurate due to sufficient training data.
Finally, we note that, by \eqref{equ:c_fairness_0}, PCP also has the group-conditional guarantee of SCP+Partition, as shown in \Cref{fig:fair3}.  

Asymptotically, the guarantee in \eqref{equ:c_fairness_0} leads to the conditional independence in \eqref{equ:c_fairness},  by the consistency assumption of \(\hat e(\cdot)\) and $A_{n+1}\independent X_{n+1}\mid e(X_{n+1})$. The independence means PCP will generate the same interval for two individuals who differ only in gender but share all other features, thereby eliminating the effect of partition on any individual's interval. 
Furthermore, if $Y_{n+1}\independent A_{n+1}\mid X_{n+1}$, the independence in \eqref{equ:c_fairness} implies that the coverage rate of our interval is also independent of $A_{n+1}$ given $X_{n+1}$.  

To empirically verify the independence in \eqref{equ:c_fairness}, we measure the generalized covariance \citep[Section 3]{shah2020hardness}\footnote{We additionally sample 2000 data points to estimate the conditional expectations in the covariance. The estimators are two random forest models applied to compute the empirical covariance.} between 
the upper bounds of the prediction intervals and the gender feature.  
As shown in  \Cref{fig:fair4}, the covariance for PCP is approximately 0, which is significantly smaller than the covariance for SCP+Partition. 
The covariance for SCP is 0 because its interval does not depend on gender. 
This result confirms that the dependence of our interval on gender can vanish asymptotically.

Finally, we note that the equalized coverage problem studied in this section is related to the covariate shifts problem in \citet{tibshirani2019conformal}. However, the difference of PCP is that it maintains the group-conditional guarantee of SCP+Partition and provides the new guarantees in \Cref{prop:fairness}.

\vspace{-5pt}

\section{Level-adaptive conditional coverage}\label{sect:pcc}
        
Next, we consider classification problems and apply PCP to create level-adaptive prediction sets. The adaptivity ensures that the top-class accuracy of our prediction set aligns with the classifier's top-class predictive probability. This allows us to lower the target coverage rate and obtain singleton or
small prediction sets, which can be useful in certain applications, such as recommending a more focused set of products to users when the classifier is
uncertain about their preferences.

Suppose that a pre-fitted classifier predicts $Y_{n+1} \in \mathcal{Y}=\{1,2,\dots, K\}$  based on a probability vector $\hat \eta (X_{n+1})$, where $\hat \eta_k(X_{n+1})$ is the predictive probability of $Y_{n+1}=k$.
 We define the top-class prediction and the top-class  predictive probability as
\[
\hat \mu(X_{n+1}) = \argmax_{k\in [K]} \hat \eta_k(X_{n+1}) \  \text{ and } \ \hat p(X_{n+1}) = \max_{k\in [K]} \hat \eta_k(X_{n+1}).
\]
We call the classifier top-class calibrated  \citep{gupta2021top} if it obeys
\begin{equation}\label{equ:cal}
\mathbb P \big\{ Y_{n+1}=	 	\hat \mu(X_{n+1}) \mid \hat p(X_{n+1}),\hat \mu(X_{n+1}) \big\} =  \hat p(X_{n+1}),
\end{equation}
i.e., the true probability of $Y_{n+1} =  \hat \mu(X_{n+1})$ is exactly $ \hat p(X_{n+1})$ when the classifier predicts $Y_{n+1} = \hat  \mu(X_{n+1})$ holds with probability $ \hat p(X_{n+1})$. 
In other words, the accuracy of the top-class prediction  $ \hat \mu(X_{n+1})$ stays consistent with the top-class probability $\hat p(X_{n+1})$.
Before calibration, the original classifier often overfits the training data and does not obey \eqref{equ:cal}.  Calibration methods adjust the predictive probabilities $\hat \eta (X_{n+1})$ by fitting nonparametric regression models to validation data. This can be done e.g.~via histogram binning or isotonic regression \citep{zadrozny2001obtaining}. These methods allow the top-class prediction $ \hat \mu(X_{n+1})$ to obey \eqref{equ:cal} asymptotically. Additional calibration methods and guarantees can be found in \citet{guo2017calibration,kull2019beyond,gupta2021top} and the references therein.

In comparison, existing conformal prediction sets can offer marginal and class-conditional guarantees for covering $Y_{n+1}$ in finite samples; see examples in \cite{sadinle2019least,romano2020classification,stutz2021learning,ding2023class}.
However, when the classifier has low predictive probability for the top class, these prediction sets would need to include many classes to achieve a target coverage rate of 90\%. In medical diagnostics, a large prediction set has limited utility as it offers little information on the true outcome $Y_{n+1}$ and may take a long time for experts to review. 
A calibrated classifier satisfying \eqref{equ:cal} can offer more information about $Y_{n+1}$ through its top-class prediction $\hat \mu(X_{n+1})$ and predictive probability $\hat p(X_{n+1})$, which are easier for experts to review.
In this section, we shall use PCP to generate smaller, often single label prediction sets, combining the benefits of top-class calibration and conformal prediction.

We define our prediction set using the nonconformity score proposed by \citet{romano2020classification}. For simplicity, 
we omit discussing the randomization step involved and note here that it occasionally generates an empty set to keep the coverage close to the target rate. We however include this randomization  step in both SCP and PCP in our experiments below. 

We let $\hat \eta^{()}(X_{i})$ denote the sorted version of $ \hat \eta(X_{i})$ from the largest to the smallest, and $K_i$ denote the rank of $  \hat \eta_{Y_i}(X_{i})$  in  $ \hat  \eta^{()}(X_{i})$. For example, if $Y_i = 3$ and $ \hat  \eta(X_{i}) = [0.2,0.5,0.3]$, 
\[
 \hat \eta^{()}(X_{i}) = [0.5,0.3,0.2]  \ \text{ and } \ K_i = 2.
\]
Similarly, we let 
$K_{n+1}^y$ denote the rank of $ \hat \eta_{y}(X_{n+1})$ in $ \hat \eta^{()}(X_{n+1})$  for any value $y\in\{1,2,\dots, K\}$. Suppose the example above holds with $X_i$ replaced by $X_{n+1}$. We have
\[
K_{n+1}^1=3,\ K_{n+1}^2=1  \ \text{ and } \ K_{n+1}^3=2.
\]
We next define the nonconformity score $S_i$ as the probability we need to add on top of the top-class probability $ \hat p(X_{n+1})$ for the top-$K_i$ prediction set to cover $Y_i$:
\[
	S_i =  \sum_{k\in [K_i]} \hat \eta_k^{()}(X_{i}) -  \hat  p(X_{n+1})  \quad  \text{and}\quad 
S_{n+1}^y = \sum_{k\in [K_{n+1}^y]}\hat\eta_k^{()}(X_{n+1}) -  \hat p(X_{n+1}). 
\]
Our PCP prediction set for $Y_{n+1}$ is defined as
\begin{equation}\label{equ:class_pcp}
\hat C_n^{\textup{PCP}}(X_{n+1}) = \left\{ y\in \mathcal{Y}: S_{n+1}^y \leq Q_{1-\hat \alpha^\star (X_{n+1})}\left(\sum_{i=1}^{n}\hat w_i^{\star} \delta_{S_i} + \hat w_{n+1}^{\star} \delta_{1} \right) \right\},
\end{equation}
where \(\hat L^\star \sim \textup{Multi}\big(m,\hat\eta(X_{n+1})\big)\),
\[
\hat\eta^\star(X_{n+1})=\hat L^\star/m,\qquad
\hat p^\star(X_{n+1})=\max_{k\in[K]}\hat\eta_k^\star(X_{n+1}),
\qquad
\hat\alpha^\star(X_{n+1})=1-\hat p^\star(X_{n+1}),
\]
and the weights \(\hat w_i^\star\) are defined by
\begin{equation}\label{equ:w_i_eta}
\hat w_i^\star =
\frac{\prod_{k=1}^{K}\big[ \hat \eta_k(X_i)\big]^{m\hat \eta_k^\star (X_{n+1})}}
{\sum_{j=1}^{n+1}\prod_{k=1}^{K}\big[ \hat \eta_k(X_j)\big]^{m\hat \eta_k^\star (X_{n+1})}}.
\end{equation}
Here $m$ is a user-specified precision parameter, and $\hat \eta^\star (X_{n+1})$ is the randomized predictive probability vector. If $\hat \eta(X_i)$ and $\hat \eta^\star (X_{n+1})$ are similar, the score $S_i$ receives a larger weight $\hat w_i^\star $ in the interval \eqref{equ:class_pcp}. 
In other words, PCP adjusts its interval mainly according to validation points whose predictive probabilities are similar to the randomized predictive probabilities of the test point. 
The proposition below follows directly from \Cref{thm:validity_1} after a change of notation.
\begin{proposition}\label{prop:class_pcp}
Assume that $(Z_1,\ldots,Z_{n+1})$, with $Z_i=(X_i,Y_i)$, is an exchangeable sequence.  Then the posterior conformal prediction (PCP) set \eqref{equ:class_pcp} satisfies
\begin{equation}\label{equ:prob_guarantee}
\mathbb P \big\{ Y_{n+1} \not\in \hat C_{n}^{\textup{PCP}}(X_{n+1}) \,\big|\, \hat \eta^\star (X_{n+1}) \big\}\leq \hat \alpha^\star (X_{n+1}).
\end{equation}
\end{proposition}
The adaptive level $\hat\alpha^\star (X_{n+1})$ defined above
 is determined by $\hat \eta^\star (X_{n+1})$ and can therefore be chosen adaptively when constructing the prediction set $\hat C_{n}^{\textup{PCP}}(X_{n+1})$. When the target coverage rate $\hat p^\star (X_{n+1}) = 1-\hat\alpha^\star (X_{n+1})$ is low, the prediction set tends to be small. If the classifier is well calibrated so that the accuracy of its top-class prediction for $Y_{n+1}$ matches the target coverage rate $\hat p^\star (X_{n+1})$, the prediction set is often a singleton that contains the true value of $Y_{n+1}$ with probability at least $\hat p^\star (X_{n+1})$. In this sense, \Cref{prop:class_pcp} approximately yields the top-class calibration guarantee \eqref{equ:cal}.

\begin{figure}[t]
     \centering
     \begin{subfigure}[b]{0.48\textwidth}
     \hspace{-20pt}
         \centering
         \includegraphics[width=0.9\textwidth]{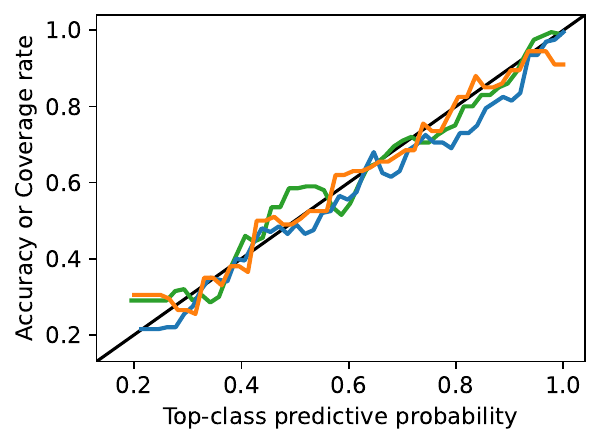}
          \caption{Calibration curves (CNN).}
           \label{fig:cal1}
    \vspace{10pt}
     \end{subfigure}
\begin{subfigure}[b]{0.48\textwidth}
         \centering
         \includegraphics[width=0.9\textwidth]{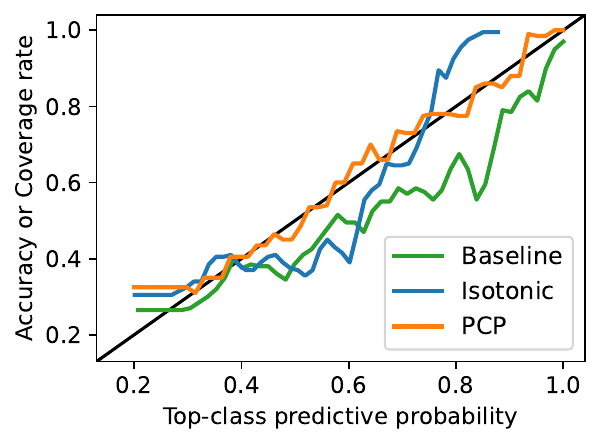}
         \caption{Calibration curves (logistic regression).}
          \label{fig:cal2}
          \vspace{10pt}
            \end{subfigure}
     \begin{subfigure}[b]{0.47\textwidth}
         \centering
                   \hspace{-20pt}
         \includegraphics[width=0.9\textwidth]{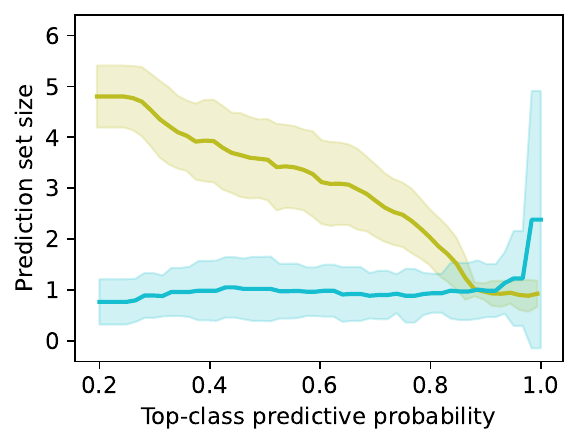}
          \caption{Prediction set size (CNN).}
     \label{fig:cal3}
     \end{subfigure}
      \hspace{0.7pt}
	\begin{subfigure}[b]{0.47\textwidth}
         \centering
         \includegraphics[width=0.9\textwidth]{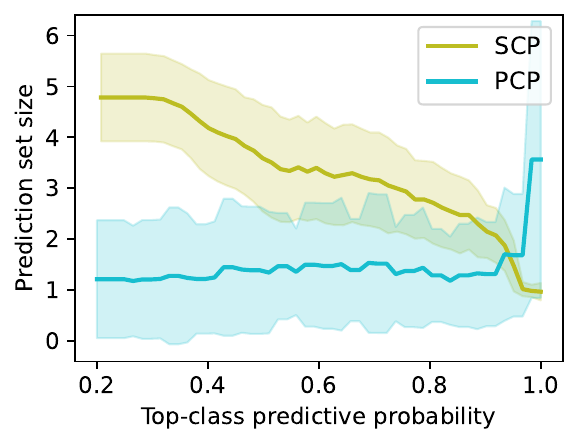}
         \caption{Prediction set size (logistic regression).}
          \label{fig:cal4}
            \end{subfigure}
        \caption{Comparison of calibration and conformal prediction methods for CNN and logistic regression models on the HAM10000 image dataset.}
        \label{fig:calibration_1}
\end{figure}

In the following experiment,  we compare PCP against the isotonic regression-based calibration method \citep{zadrozny2001obtaining} and the SCP method for classification  \citep{romano2020classification}. Our evaluation is based on the HAM10000 dataset \citep{tschandl2018ham10000}, which consists of 10,015 dermatoscopic images with seven classes of pigmented skin lesions. We consider two baseline classifiers: a convolutional neural network (CNN) and logistic regression.  We split the dataset into three folds: one for training the classifier, another for validation, and the last for testing the methods.

First, we compare calibration curves of PCP and isotonic regression. To construct each curve, we evenly place 50 points within the support of the top-class predictive probability and then calculate the local average accuracy at each point, using the 200 nearest test points. The original CNN model is relatively well-calibrated, since the green curve approximately follows the diagonal line in \Cref{fig:cal1}. The logistic regression model is poorly calibrated, as shown by the green curve in \Cref{fig:cal2}. 
We can see from the same figure that our prediction sets always keep the local coverage rate close to the top-class predictive probability $\hat p(X_{n+1})$. In contrast, the 
isotonic regression method fails to calibrate the logistic regression model.

By definition, our guarantee in \eqref{equ:prob_guarantee} is a good approximation of the calibration guarantee in \eqref{equ:cal}, if our prediction set is most often a singleton. This is the case for the CNN model when its top-class predictive probability $\hat p(X_{n+1})\leq0.9$  in \Cref{fig:cal3}. When  $\hat p(X_{n+1})>0.9$, our method enlarges the prediction set to meet the target coverage rate set to $\hat p(X_{n+1})$. \Cref{fig:cal4} shows that our prediction set is larger for the less-calibrated logistic regression model. Nevertheless, it is smaller than the prediction set SCP produces when the top-class predictive probability $\hat p(X_{n+1})\leq 0.9$.

\section{Discussion}\label{sect:end}
This paper introduced posterior conformal prediction (PCP), a general method to construct prediction intervals with conditional coverage guarantees. We showed that PCP improves conditional coverage by modeling residual distributions/nonconformity scores as mixtures, strengthens subgroup coverage by reweighting individuals within subgroups, and generates smaller prediction sets by adapting the target coverage level based on classifiers’ predictive probabilities.
In the following, we discuss a few limitations and extensions of PCP.

\textbf{PAC guarantee.} 
For large $n$, we may also want a stronger version of marginal validity, namely that guarantees such as \eqref{equ:marginal} and \eqref{equ:pcp_guarantee} hold conditional on the realized validation data $Z_{[n]}$ with high probability.
This would rule out the possibility that validity holds only on average over validation samples while failing badly for many particular realizations. Formally, we ask that
\begin{equation}\label{eqi:pac}
\mathbb P \left\{ \alpha(Z_{[n]}) := \mathbb P\big\{Y_{n+1}\not\in \hat C_n(X_{n+1}) \mid Z_{[n]}\big\} \leq  \alpha + o(1) \right\}\geq 1-o(1).
\end{equation}
A prediction set $\hat C_n$ satisfying \eqref{eqi:pac} is called \quotes{Probably Approximately Correct} (PAC) \citep{wilks1941,vovk2012conditional}.
When marginal validity fails to hold, i.e., $\mathbb E[\alpha(Z_{[n]})]> \alpha$, it is less likely that $\alpha(Z_{[n]})\leq \alpha$ can hold with high probability as in \eqref{eqi:pac}. On the contrary, when marginal validity holds, it is relatively straightforward to prove the PAC guarantee in \eqref{eqi:pac}. For PCP intervals with data-dependent weights, the proof requires a mild stability assumption on our mixture learning algorithm. We refer to \Cref{sect:pac_appendix} for details.

\textbf{Distributional shifts.} 
There has been growing interest in extending conformal prediction to applications with distributional shifts, such as election forecasting \citep{cherian2020washington,barber2022conformal} and time series modelling  \citep{xu2021conformal,gibbs2022conformal,zaffran2022adaptive}. PCP  can be extended to these applications, by fitting a mixture model to residuals, or fitting a classifier to predict the variable that causes the distributional shift. The advantage of PCP is that it guarantees marginal validity when there is no distributional shift.
When using our method in online or dynamic settings, we can let the number of cluster distributions grow as more data distributions emerge over time or space.

\textbf{Reward maximization.} In \Cref{sect:pcc}, we applied our method to generate level-adaptive prediction sets. This adaptivity is useful when increasing uncertainty is costly. For example, suppose $h(X)$ is a decision rule that determines whether recommending the movies in our prediction set $\hat C_n^{\textup{PCP}}(X)$ to a user with features $X$. If the user randomly selects one movie to watch, the expected benefit of our recommendation is
\[
\mathbb E \left\{ h(X)	\times  \bigg[ r_1(X)\frac{\one \big\{Y \in  \hat C_n^{\textup{PCP}}(X)\big\}	}{\big|\hat C_n^{\textup{PCP}}(X)\big|} 	\\
+  r_2(X)  \bigg( 1-\frac{\one \big\{Y \in  \hat C_n^{\textup{PCP}}(X)\big\}	}{\big|\hat C_n^{\textup{PCP}}(X)\big|} \bigg) \bigg]		\right\},
\]
where $r_1(X)$ is the reward for watching the favourite movie and $r_2(X)$ is the reward for watching the other movies.
Generating a large prediction set may not maximize the expected reward; in this
setting, the level-adaptive prediction sets of PCP can lower the target coverage rate to define a decision rule $h(X)$ that favors smaller recommendation sets.

\section{Acknowledgements}

Y.Z thanks John Cherian, Ying Jin, Daniel LeJeune, Jinzhou Li, and Tijana Zrnic for their valuable feedback on this article. Y.Z and E.J.C. were supported by the Office of Naval Research under Grant No. N00014-24-1-2305 and the National Science Foundation under Grant No. DMS2032014. Additionally, E.J.C. acknowledges support from the Office of Naval Research under Grant No. N00014-24-1-2305, the Simons Foundation under Award 814641, and the Army Research Office under Grant No. 2003514594.

\bibliographystyle{plainnat}
\bibliography{references}

\clearpage

\appendix
\appendixpage

\section{Randomized versus non-randomized intervals}\label{sect:random_nonrandom}

This section first presents a counterexample to prove that non-randomized intervals may lose marginal validity. Following this, we discuss the difference between the randomized and non-randomized PCP interval as the precision parameter $m$ increases.

Previously, Proposition 2 in \citet{guan2023localized} demonstrated the loss of marginal validity in localized conformal prediction intervals. In the example there, every feature $X_{ij}\in \{-1,0, 1\}$ for $j\in [d]$.
The residuals at the origin $(X_i=0 )$ are 0. When $X_i=0$ holds with high probability, the interval mainly uses zero residuals and achieves a poor coverage rate. The loss of coverage in this example is primarily due to the choice of weight function $w_i\propto\exp (-d(X_i,X_{n+1}))$ in the interval, where 
\[
d(X_i,X_{n+1})  = 
\begin{cases}
 0,  \ \   \ \ \text{ if } X_i\in \{0,X_{n+1}\},\\
+\infty,  \  \text{otherwise}.
\end{cases}
\]
Equality in the distance metric $d(\cdot,\cdot)$ is non-transitive, i.e.,
\[
d(0,X_{n+1}) = 0 \text{ and } d(0,X_{n+2}) = 0 \ \not\Rightarrow\  d(X_{n+1},X_{n+2}) =  0,
\]
whereas most distance metrics are transitive in measuring similarity between data points.
Transitivity is also essential for other metric properties, such as the triangle inequality.
In the following, we provide a new example to illustrate that non-randomized intervals using standard distance metrics may still fail to achieve marginal validity. In distribution-free settings, we assume that the feature space can be divided into two disjoint subsets with significantly different residuals.

\begin{example}\label{example_1}
Without loss of generality, we assume that the residuals for $X_i>0$ are larger than those for  $X_i \leq 0$ with probability 1. 
We also assume that $(X_1,R_1),\dotsc, (X_{n+1},R_{n+1})$ are distinct with probability 1.
We define a prediction interval with weights only depending on whether $X_i,X_{n+1}>0$ as follows:
\begin{equation}\label{equ:hat_c_n}
\hat C_{n} (X_{n+1}) = \left[\hat \mu(X_{n+1}) \pm Q_{1-\alpha} \left(\sum_{i=1}^{n}w_i\delta_{R_i} + w_{n+1}\delta_{+\infty}   \right)\right],
\end{equation}
where 
$w_i = \phi(X_i,X_{n+1})/\sum_{j=1}^{n+1}\phi(X_j,X_{n+1})$ for $\phi(X_{n+1},X_{n+1})=1+c$ and
\begin{equation}\label{equ:phi_a_b}
\phi(X_i,X_{n+1}) = 
\begin{cases}
    a, \hspace{-1pt} \text{ if } X_i \leq 0, X_{n+1} > 0, \\
    b, \text{ if } X_i > 0, X_{n+1} \leq  0, \\
    1, \text{ otherwise}.
\end{cases}
\end{equation}
where $a,b\in [0,1]$. Choosing different values for $a$ and $b$ can make the distance metric in $\phi$ asymmetric.
In \eqref{equ:phi_a_b}, the $X_i$'s with the same sign as $X_{n+1}$ receive greater weight than those with a different sign. If $c$ is positive, $X_{n+1}$ itself receives a larger weight than others. If $c=-1$, we remove the point mass $\delta_{+\infty}$  in \eqref{equ:hat_c_n},  which defines the interval $\hat C_{n} (X_{n+1})$ solely using the residuals $R_{[n]}$. 
In the proposition below,  $\rho = \mathbb P \{X_i > 0\}$ and we omit terms that vanish exponentially as the sample size $n$ increases.
\end{example}

\begin{figure}[t]
     \centering
  \begin{subfigure}[b]{1\textwidth}      
	\centering 
         \includegraphics[width=0.95\textwidth]{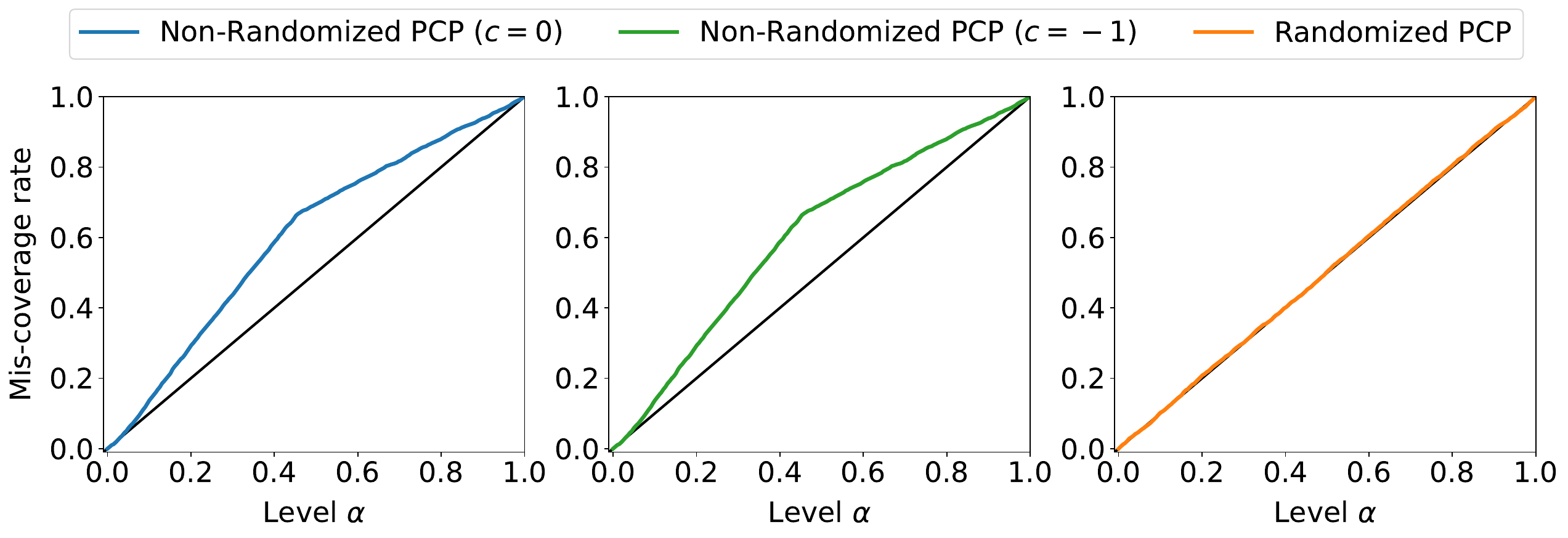}
          \caption{Asymmetric similarity measure $\phi$.}
          \label{fig:kl_1}
          \vspace{10pt}
\end{subfigure}
   \begin{subfigure}[b]{1\textwidth}    
	\centering  
         \includegraphics[width=0.95\textwidth]{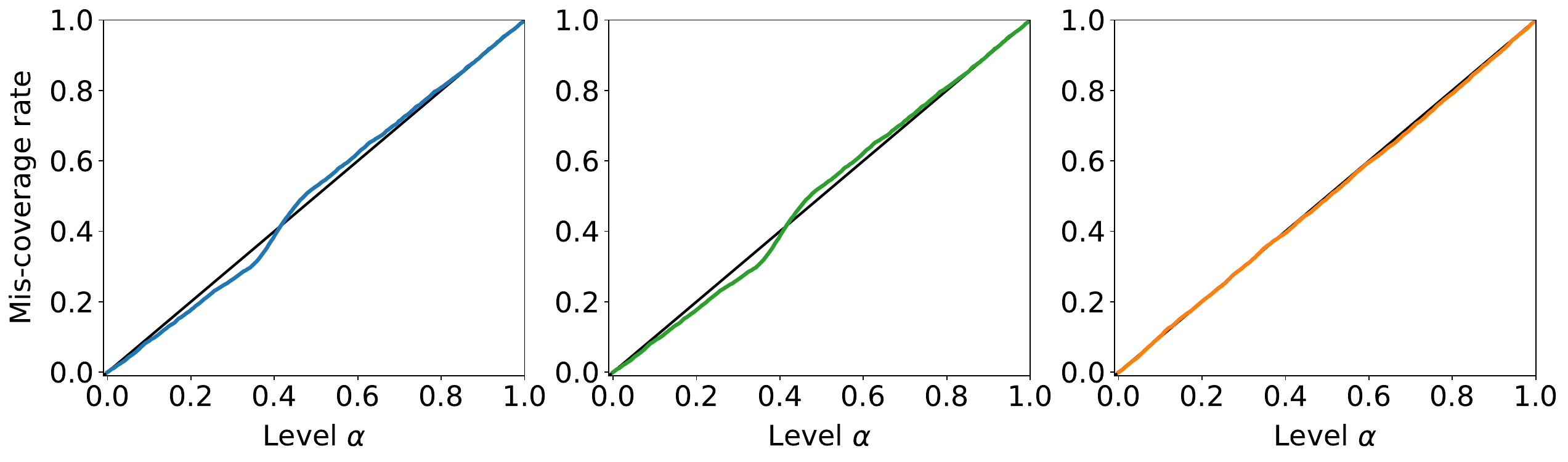}
          \caption{Symmetric similarity measure $\phi$.}
             \label{fig:kl_2}
	\end{subfigure}	
\caption{Comparison of the randomized and non-randomized PCP intervals in \Cref{sim:pcp_sim}.}
\vspace{-3pt}
\label{fig:kl}
\end{figure}

\begin{proposition}\label{prop:example}
Suppose that \(Z_i=(X_i,Y_i)\), \(i\in[n+1]\), are i.i.d.
In the setup of \Cref{example_1}, assume that \(a,b\in(0,1)\), \(c\ge -1\), and 
\(\rho=\mathbb P\{X_i>0\}\in(0,1)\).
Then the interval \(\hat C_n(X_{n+1})\) in \eqref{equ:hat_c_n}, using the weights in \eqref{equ:phi_a_b}, satisfies, up to terms that vanish exponentially fast in \(n\),
\begin{equation}\label{equ:coverage_rate}
 \mathbb P\big\{ Y_{n+1}\not\in \hat C_{n}(X_{n+1}) \big\} 
 \geq 
 \alpha + \alpha a(1-\rho) - (1-\alpha)b\rho  
 - \frac{2(1-\alpha)c+2}{n+1}.
\end{equation}
\end{proposition}
The proof of the proposition can be found in  \Cref{sect:prop_example}. In the lower bound,
the term $\alpha a(1-\rho)$ is the coverage loss for $X_{n+1}>0$ where the test residual $R_{n+1}$ is large, while $(1-\alpha)b\rho$ is the coverage gain for $X_{n+1}\leq 0$ where $R_{n+1}$ is small. The mis-coverage rate in \eqref{equ:coverage_rate} is above $\alpha$ if 
\begin{equation}\label{equ:condition}
\frac{\alpha}{1-\alpha} a - \frac{\rho}{1-\rho}b \geq \frac{2(1-\alpha)c+2 }{(n+1)(1-\alpha)(1-\rho)}.
\end{equation}
This condition holds if the interval uses an asymmetric distance metric with $a > b$ in \eqref{equ:phi_a_b}.  It becomes more difficult to satisfy if $\alpha$ is small and the distance metric is symmetric with $a=b$, for example, with Gaussian or uniform kernels.

\vspace{8pt}
\begin{simulation} \label{sim:pcp_sim}
We next verify \Cref{prop:example} by comparing the mis-coverage rates of the randomized and non-randomized PCP intervals on a synthetic dataset. We generate $X_i\sim \text{Bern}(\rho = 0.4)$ and the residual $R_i\mid X_i
=0 \sim \mathcal N (5,1)$ and $R_i\mid X_i = 1 \sim \mathcal N (10,1).$ Suppose that $\pi(0) = [0.8,0.2]$ and $ \pi(1) = [1,0].$  In \eqref{equ:phi_a_b}, we let
\[
a = \exp\{- D_{\textup{KL}}(\pi(1)\| \pi(0))\} \approx 0.8  \ \text{ and }\ b =  \exp\{- D_{\textup{KL}}(\pi(0)\| \pi(1))\}=0.
\]
We generate 10000 data points $(X_i,R_i)$ to construct the interval $\hat C_{n} (X_{n+1})$ in \eqref{equ:hat_c_n} and
the randomized PCP intervals with the precision parameter \(m=1\). Another 10000 data points are generated to evaluate the coverage rates of these intervals. As shown in \Cref{fig:kl_1}, the non-randomized intervals exhibit coverage gaps across all values of $\alpha$, while the randomized interval maintains a miscoverage rate consistent with $\alpha$.
In  \Cref{fig:kl_2}, we let $\pi(0)=[0.5, 0.5]$ to make both $a$ and $b$ close to 0.8. Using this nearly symmetric distance metric, the non-randomized PCP intervals have significantly smaller coverage gaps. The gap is approximately 0 when $\alpha \approx 0.4  = \rho$. This result aligns with our lower bound in \eqref{equ:coverage_rate} when $\alpha =\rho$ and $a = b$.
\end{simulation}

\textbf{PCP intervals.} Thus far, we have shown that non-randomized weighted intervals may fail to achieve distribution-free marginal validity. We next give a heuristic comparison between randomized and non-randomized PCP intervals as the precision parameter \(m\) increases. For any fixed \(\pi\), define the non-randomized and randomized weights by
\[
w_i \propto \phi_i := \prod_{k=1}^{K}\big[\pi_k(X_i)\big]^{m\pi_k(X_{n+1})},
\quad 
w_i^\star \propto \phi_i^\star := \prod_{k=1}^{K}\big[\pi_k(X_i)\big]^{L_k^\star},
\]
where \(L^\star\sim \mathrm{Multi}(m,\pi(X_{n+1}))\).
Let
\(
\psi=\sum_{i=1}^{n+1}\phi_i
\)
and 
\(
\psi^\star=\sum_{i=1}^{n+1}\phi_i^\star.
\)

Assuming \(\phi_i>0\) for all \(i\), define
\(
a_i:=\phi_i^\star/\phi_i,
\)
and 
\(
T:=\sum_{i=1}^{n+1}w_i a_i = \psi^\star/\psi.
\)
Then
\[
 w_i^\star = \frac{w_i a_i}{T},
\qquad
w_i- w_i^\star = \frac{w_i}{T}(T-a_i).
\]
Therefore,
\begin{align*}
\left|\sum_{i=1}^{n+1}(w_i- w_i^\star)\one\{R_i\ge R_{n+1}\}\right|
=
\frac{1}{T}\left|\sum_{i=1}^{n+1} w_i(T-a_i)\one\{R_i\ge R_{n+1}\}\right|
\le
\frac{1}{T}
\left(
\sum_{i=1}^{n+1} w_i (a_i-T)^2
\right)^{1/2},
\end{align*}
where the inequality follows from Cauchy--Schwarz and the fact that
\(\sum_{i=1}^{n+1} w_i \one\{R_i\ge R_{n+1}\}\le 1\).
The right-hand side is the weighted standard deviation of \(a_i=\phi_i^\star/\phi_i\) divided by its weighted mean \(T\). This suggests that the gap between the randomized and non-randomized weighted tail probabilities is governed by the relative fluctuation of \(\phi_i^\star/\phi_i\). Establishing a rigorous convergence rate as \(m\) increases would require additional assumptions, which we do not pursue here.

\section{Matching algorithm}\label{sect:alg}

Estimating mixture models is a nontrivial nonconvex optimization problem. Computing the PCP interval in \eqref{equ:posterior_p} requires refitting $\hat \pi^y$ for all $y\in \mathcal{Y}$.
In this section, we introduce a mixture learning algorithm to address this computational challenge. 

Our algorithm is motivated by two observations. First, given a pre-fitted model $\hat \mu $, it is easy to impute $R_{n+1}^{y} = |y-\hat \mu (X_{n+1})|$ for all $y\in \mathcal Y$. Second, to compute the weights in \eqref{equ:posterior_p}, we only need to estimate the membership probabilities, not the entire mixture model. The conditional cumulative distribution functions (CDFs) of the residuals are monotone and right-continuous. We learn the probabilities by matching the CDFs at some fixed points in the support of the marginal residual distribution.

In what follows, we introduce our algorithm $\mathcal{A}$ in three steps: gridding, fitting, and clustering. After the introduction, we provide a numerical experiment to confirm that our algorithm's complexity for computing a single PCP interval is linear in $n$.

\textbf{Gridding.}
In the gridding step, we choose a fixed number of points to match the residual CDFs. This process requires re-running the training algorithm $\mathcal{A}'$ for the predictive model $\hat \mu = \mathcal{A}'(\mathcal{D}_{[n]}')$ on the dataset
\begin{equation}\label{equ:d_prime}
\mathcal{D}_{[n]}' = \{(X_1',Y_1'),\dots,(X_n',Y_n') \}.
\end{equation}
We only need to implement the gridding step once. We will use the same grids to compute the intervals for all the test points.

We first generate $n$ residuals via a 20-fold cross-validation on $\mathcal{D}_{[n]}'$. 
More precisely, we divide $[n]$ into 20 subsets
 $\mathcal{I}^{(1)},\dots, \mathcal{I}^{(20)}$ and compute 20 models  
\[
\hat \mu_{-\mathcal{I}^{(1)}} = \mathcal{A}'(\mathcal{D}_{[n]\setminus \mathcal{I}^{(1)}}'), \dots, 
\hat \mu_{-\mathcal{I}^{(20)}} = \mathcal{A}'(\mathcal{D}_{[n]\setminus \mathcal{I}^{(20)}}').
\]
For every $(X_j',Y_j')$ in $\mathcal{D}_{[n]}',$ we let  $ \hat \mu_{-j}(X_j') = \sum_{t=1}^{20} \one \{j \in \mathcal{I}^{(t)}\}\hat \mu_{-\mathcal{I}^{(t)}}(X_j')$ denote the prediction from the model that is not fitted to $Y_j'$.
Then we compute the  residual
\begin{equation}\label{equ:r_prime}
R_j' = |Y_j'-\hat \mu_{-j}(X_j') |.
\end{equation}
Every model $\hat\mu_{-\mathcal{I}^{(t)}}$ is trained on about 5\% less data than $\hat \mu$. 
The residuals $R_{[n]}' = (R_1',\dots, R_{n}')$  have a similar support compared to $R_{[n]}$.
We choose $s$ different quantiles $\xi_{[s]}'$ of the empirical distribution of $R_{[n]}'$ to
represent the support of the distribution of $R_i$. For example, we set \(s=9\) and choose the \(q\)-quantiles for \(q\in\{0.1,\dotsc,0.9\}\).

\textbf{Fitting.}
For each $t\in [s]$ and any value $y\in \mathcal{Y}$, we want to fit a model $\tau_t^y(X)$ to estimate the conditional expectation $\mathbb P\{R\leq  \xi_t'\mid X  \}$ using the imputed data
\begin{equation}\label{equ:data_}
 (X_1,\one\{R_1 \leq \xi_t'\}),\dots, (X_n,\one\{R_n \leq \xi_t'\}),(X_{n+1},\one\{R_{n+1}^y \leq \xi_t'\}).
\end{equation}
The imputed label can only take two values, 0 or 1, so we only need to compute two estimators, $\tau_{t,0}(X_i)$ and $\tau_{t,1}(X_i)$, for all values $y$.  Then, we can define
\begin{equation}\label{equ:tau_t}
\tau_t^y(X_i)  = 
\begin{cases}
    \tau_{t,1}(X_i),\ \text{if } R_{n+1}^y \leq \xi_t',\\
    \tau_{t,0}(X_i),\ \text{otherwise.}
\end{cases}
\end{equation}
To prevent $\tau_{t}^y(X_{n+1})$ from over-fitting to $\one\{R_{n+1}^y \leq \xi_t'\}$, we implement a cross-fitting strategy: divide the imputed data in \eqref{equ:data_} into 20 folds and compute every estimate $\tau_{t}^y(X_{i})$ using the model fitted to the folds that do not contain $(X_{i},\one\{R_{i} \leq \xi_t'\})$.
To ensure our algorithm is efficient,  we construct the estimators using linear models and the least-squares method proposed by \citet{kanamori2009least}. This method is different from logistic regression in the Bregman divergences they minimize; see Section 3.1 in \citet{sugiyama2012density} for more details.  Here, we briefly describe the idea.

For every $t\in [s]$, we follow  Bayes' rule to estimate  $\mathbb P\{R\leq  \xi_t'\mid X=x  \}$ as 
\begin{equation}\label{equ:tau_t_hat}
\hat \tau_t(x) =  [n_t \hat g_t(x)]/[ n-n_t+n_t\hat g_t(x)], 
\end{equation}
where \(n_t=\sum_{i=1}^{n}\one\{R_i\leq \xi_t'\}\), 
\(g_t(x)=f_X(x\mid R\leq \xi_t')/f_X(x\mid R>\xi_t')\) is the conditional density ratio, and \(\hat g_t\) is its estimator.
The mean squared loss  of $\hat g_t$ can be written as 
\begin{align*}
 &\ \int  [\hat g_t(x) -  g_t(x)]^2 f_X(x\mid R>  \xi_t'  )dx	\\
=	&\   \int \hat g_t^2(x)f_X(x\mid R>  \xi_t'  )dx -  2\int \hat g_t(x)f_X(x\mid R\leq   \xi_t'  )dx + \int g_t^2(x)f_X(x\mid R>  \xi_t'  )dx.
\end{align*}
Observe that the third term does not depend on the estimator $\hat g_t$. In other words, we can construct $\hat g_t$ by minimizing the empirical version of the first two terms. When $\hat g_t$ is linear, the objective function on the first $n$ data points in \eqref{equ:data_} is given by 
\begin{equation}\label{equ:least_squares}
\min_{\beta}\frac{1}{n-n_t}\sum_{i=1}^{n}\one\{R_i > \xi_t'\}(X_i^{\top}\beta)^2 - \frac{2}{n_t}\sum_{i=1}^{n}\one\{R_i\leq  \xi_t'\}X_i^{\top}\beta,
\end{equation}
The solution of \eqref{equ:least_squares} can be expressed as
\[
\hat \beta = \frac{n-n_t}{n_t}\left(\sum_{i=1}^{n}\one\{R_i > \xi_t'\}X_i X_i^{\top}\right)^{-1}\sum_{i=1}^{n}X_i\one\{R_{i} \leq \xi_t'\}.
\]
Finally, we define $\hat \tau_t (x)$ in \eqref{equ:tau_t_hat} using $\hat g_t(x) = \max\{\hat \beta ^{\top}x,0\}$.
Given the solution $\hat \beta$, we can efficiently compute the least-squares solutions for all the test points and $y\in \mathcal{Y}$ using algorithms developed for full conformal prediction \citep{vovk2005algorithmic}. The key idea is to update $\hat \beta$ with the last point $(X_{n+1},\one\{R_{n+1}^y \leq \xi_t'\})$, using the Sherman--Morrison--Woodbury (SMW) identity. 
The complexity of updating all the estimators for one test point is 
$O(nd^2+ndc),$ where $d$ is the dimension of $\mathcal{X}$ and $c$ is the total number of linear models to update. When we use 20-fold cross-fitting, $c=  20\times 2\times 9 = 360$. 
As we will explain later, 
the constant $c$ can be reduced by computing our interval as a union when $\alpha$ is small, e.g., 0.1.  Using the identity, we can also efficiently update solutions in 
kernel regression \citep{burnaev2014efficiency} and Lasso \citep{lei2019fast}, as well as in our method with regularization. We can select the regularization hyperparameter by cross-validating the linear models fitted to the residuals from \eqref{equ:r_prime}:
\begin{equation}\label{equ:data_prime}
(X_1',\one\{R_{1}' \leq \xi_t'\}),\dots,(X_{n}',\one\{R_{n}' \leq \xi_t'\}).
\end{equation}

\textbf{Clustering.}
Let $( \xi_{0}',\xi_{s+1}') = (0,\infty)$. The response space $\mathcal{Y}$ is covered by 
\[
\mathcal{Y}^{(j)}
=
\left\{y\in \mathcal{Y}: R_{n+1}^y\in \big[\xi_{j-1}',\xi_j'\big)\right\},
\qquad j\in [s+1].
\]
By definition, 
$\xi_{0}'<\xi_1'< \dots <\xi_s'<\xi_{s+1}'$.
The estimators $\tau_1^y,\dots, \tau_s^y$ are unchanged across $y\in \mathcal{ Y}^{(j)}$ because the binary responses in  
\eqref{equ:data_} remain invariant across all $y\in \mathcal{Y}^{(j)}$. Thus, we let 
$\tau_{[s]}^{(j)}(x) = [\tau_1^y(x),\dots, \tau_s^y(x)]$ for all $y\in \mathcal{Y}^{(j)}$ and $j\in [s+1]$.

We learn the membership probabilities by solving a reconstruction problem,
\begin{equation}\label{equ:opt_2}
\hat\pi_{[n+1],[J]}^{(j)}, \hat\gamma_{[J]}^{(j)}
\in
\argmin_{\pi_{i,[J]}\in \Delta^{J-1},\,\gamma_k\in \mathbb{R}^{s}}
\sum_{i=1}^{n+1} 
\Big\| 
\tau_{[s]}^{(j)}(X_i)	
- \sum_{k=1}^{J}\pi_{i,k}\gamma_{k}	
\Big\|^2,
\end{equation}
where the choice of \(J\) will be discussed in \Cref{sect:hyper}.
For every \(k\in[J]\), \(\hat\gamma_k^{(j)}\) estimates the values of an unconditional residual CDF at the points \(\xi_1',\dots,\xi_s'\). Then, 
\(\hat\pi_{i,[J]}^{(j)}=(\hat\pi_{i,1}^{(j)},\dots,\hat\pi_{i,J}^{(j)})\) combines
\(\hat\gamma_{[J]}^{(j)}=(\hat\gamma_1^{(j)},\dots,\hat\gamma_J^{(j)})\)
to reconstruct the estimators \(\tau_{[s]}^{(j)}(X_i)\).
For \(y\in\mathcal Y^{(j)}\), we set 
\(\hat\pi^y(X_i)=\hat\pi_{i,[J]}^{(j)}\), \(i\in[n+1]\).
To warm start, we initialize $(\gamma_1,\dots, \gamma_{J})$ as the $J$ cluster centers from the $k$-means$++$ algorithm \citep{arthur2007k} applied to the estimators  $\tau_{[s]}^{(j)}(X_1),\dots, \tau_{[s]}^{(j)}(X_{n+1})$. 
The algorithm has complexity $O(Jns)$.
Then, we improve the fit by alternating optimization.  Fixing $\gamma$'s allows us to update each $\pi_{i,[J]}$ via mirror descent. 
Fixing the \(\pi\)'s enables us to update \(\gamma_1,\dots,\gamma_J\) with a least-squares fit, which has complexity $O(Jns^2+J s^3)$. Alternating between the two convex problems will converge to a local minimum of the problem in \eqref{equ:opt_2}.

\textbf{PCP interval.} Finally, we compute our PCP interval by taking a union:
\[
\hat C_{n}^{\textup{PCP}}(X_{n+1})
=
\bigcup_{j=1}^{s+1}\hat C_{n}^{(j)}
\equiv
\bigcup_{j=1}^{s+1}
\left\{
y\in \mathcal{Y}^{(j)}:
P_n^{(j)}(y)
=
\sum_{i=1}^{n+1}
\hat w_{i}^{\star,(j)}
\one {\left\{R_i\geq R_{n+1}^{y}\right\}}
>\alpha
\right\},
\]
where the weights \(\hat w_i^{\star,(j)}\) are computed from \(\hat\pi_{[n+1],[J]}^{(j)}\) using the randomized multinomial construction in \eqref{equ:posterior_p}.
In any implementation, we can compute $\hat C_{n}^{(1)},\dots, \hat C_{n}^{(s+1)}$  in a backward fashion. 
Once we reach \(j^*\) satisfying
\(\hat C_n^{(j^*)}=\mathcal Y^{(j^*)}\),
we can set the remaining $\hat C_{n}^{(j)}= \mathcal{Y}^{(j)} $ for all $j=1,\dots, j^*-1$.
This simplification never yields a smaller interval and often returns the same interval that we want to compute. Indeed, if \(j<j^*\), it often holds that
\[
P_n^{(j)}(y)> P_n^{(j^*)}(y'), \forall y\in \mathcal Y^{(j)} \text{ and } y'\in \mathcal Y^{(j^*)}.
\]
When \(\alpha=0.1\), \(j^*\) is often close to \(s\). Suppose that \(j^*=s=9\). To compute the interval, we only need to compute $\tau_{t,1}(x)$ for all $t\in [9]$ and $\tau_{9,0}(x)$ in \eqref{equ:tau_t}, i.e., update $10\times 20 = 200$ cross-fitted linear models using the solution $\hat \beta$ and the SMW identity.

\textbf{Numerical experiment.}
We verify the computational complexity of our algorithm $\mathcal A$ using the MEPS 19 and 20 datasets \citep{romano2019conformalized} with 139 features. We see from \Cref{fig:compute} that the time required to compute 1000 PCP intervals increases linearly as the size $n$ of the validation set grows from 2000 to 10000.
This verifies our calculation above, which shows that the computational complexity of $\mathcal A$ is linear in $n$. 
When the validation set size $n=10000$, our algorithm $\mathcal A$ can generate 1000 intervals in less than 10 minutes on a MacBook Pro laptop with an Intel Core i7 processor. Given this linearly increasing complexity, we can apply PCP to large datasets, and evaluate its coverage rate and interval length through hundreds of repeated experiments.

\begin{figure}[t]
\centering
         \includegraphics[width=0.8\textwidth]{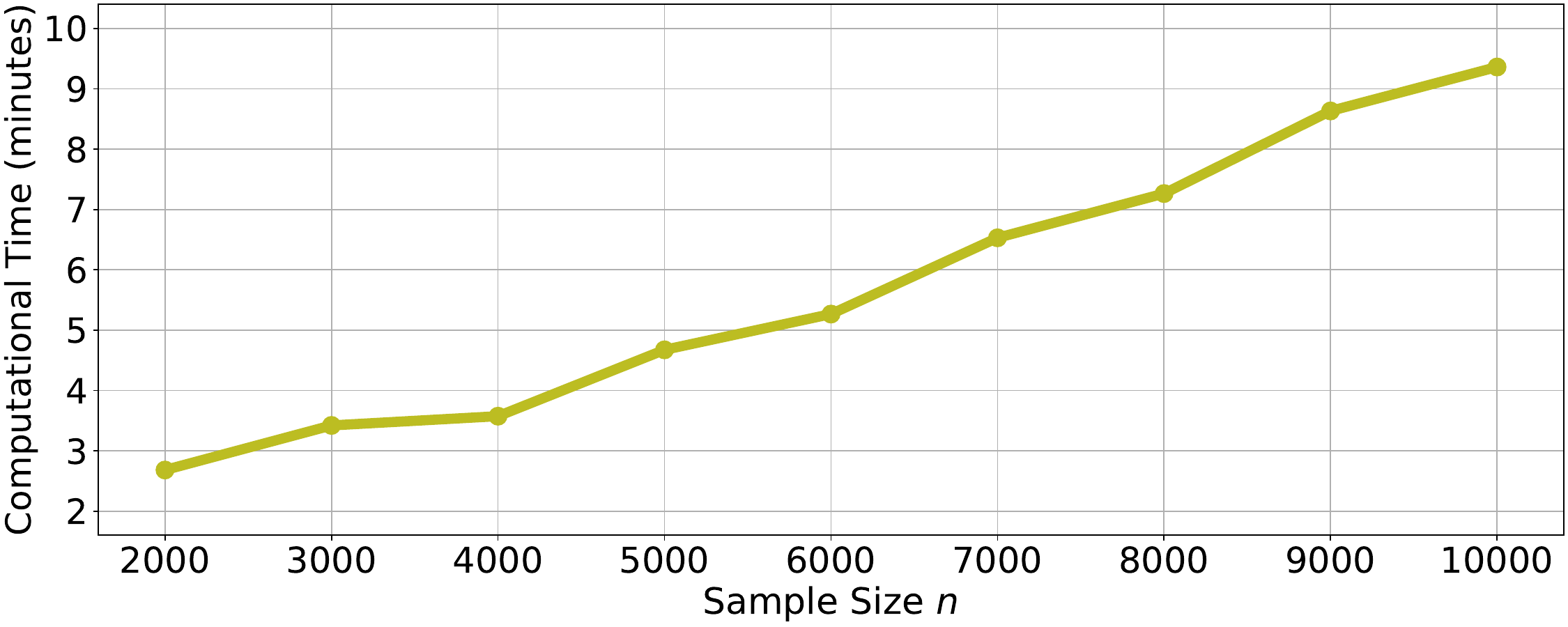}
        \caption{Computational time of generating 1000 PCP intervals.}
        \label{fig:compute}
\end{figure}

\subsection{Hyperparameter selection}\label{sect:hyper}

\textbf{Number of cluster distributions $J$.} We choose the number of cluster distributions $J$ using the dataset $\mathcal{D}_{[n]}'$ in \eqref{equ:d_prime}
	used to pre-fit $\hat \mu$.  As described in \eqref{equ:r_prime}, we can use $\mathcal{D}_{[n]}'$ to generate $n$ extra residuals  $R_{[n]}'$. We apply the algorithm $\mathcal{A}$ on $\mathcal{D}_{[n]}'$
	to learn a membership probabilities estimator. More precisely, using the data in 
\eqref{equ:data_prime}, we fit the same model $\tau_t'$ 
to estimate $\mathbb P\{R\leq  \xi_t'\mid X  \}$ for all $t\in [s]$.
Let $\tau_{[s]}' = (\tau_1',\dots, \tau_s')$. 
We learn the membership probabilities by solving a reconstruction problem,
\[
\min_{\pi_{i,[J]}\in \Delta^{J-1},\,\gamma_k\in \mathbb{R}^{s}}\sum_{i=1}^{n} \Big\| \tau_{[s]}'(X_i')	
- \sum_{k=1}^{J}\pi_{i,k}\gamma_{k}	\Big\|^2.
\]
The loss is difficult to minimize when $J$ is too small. We increment \(J=1,2,\dots\) until the coefficient of determination improves by less than 0.05 from \(J\) to \(J+1\).

\textbf{Precision parameter $m$.} Let \(\hat\pi'\) denote the membership probability estimator using the selected \(J\). For each point $(X_i',Y_i')$ in $\mathcal{D}_{[n]}'$, we compute its PCP interval for \(Y_i'\) using \(\hat\pi'\) and the other points in $\mathcal{D}_{[n]}'$ as a validation set.
More specifically, we let 
\[
\hat L_i^{\star\prime}\sim \textup{Multi}\big(m,\hat\pi'(X_i')\big),
\qquad
\hat\pi^{\star\prime}(X_i')=\hat L_i^{\star\prime}/m.
\]
The PCP interval for the response $Y_i'$ is defined as
\[
\hat C_{n-1}^{\textup{PCP}}(X_i')  =
\left\{y\in \mathcal{Y}:
\sum_{j\in[n]\setminus\{i\}} \hat w_{ij}^\star
\one {\left\{ R_j'\geq R_i^{\prime,y} \right\}}
+\hat w_{ii}^\star
>\alpha
\right\}.
\]
where \(R_i^{\prime,y}=|y-\hat \mu_{-i}(X_i')|\), and the weights
\[
\hat w_{ij}^\star \propto 
\prod_{k=1}^{J}
\big[\hat\pi_k'(X_j')\big]^{m\hat\pi_k^{\star\prime}(X_i')}.
\quad j\in[n],
\]
Let  \(\|\hat w_i^\star\|_2^2 = \sum_{j=1}^{n}(\hat w_{ij}^\star)^2\). We define the \quotes{effective sample size} of \(\hat C_{n-1}^{\textup{PCP}}(X_i')\) as \(1/\|\hat w_i^\star\|_2^2\), 
which is a heuristic formula used in the literature of distributional shifts \citep{gretton2009covariate,reddi2015doubly,tibshirani2019conformal,barber2022conformal}. When the weights $\hat  w_{ij}^\star=1/n$ for all $j\in [n]$, the effective sample size is $n$. When the weights $\hat w_{ij}^\star$ are unequal, the effective sample size decreases as the $\ell^2$-norm increases. 
We use the bisection method to search for the largest \(m\in [5,500]\) that achieves an effective sample size greater than \(100\) and \(n^{-1}\sum_{i=1}^n \hat w_{ii}^\star \leq 1/30\), ensuring both the effective sample size and the precision parameter \(m\) are reasonably large to compute our interval in \eqref{equ:posterior_p}.

Allowing for a smaller sample size may improve conditional coverage at the cost of wider intervals. While this trade-off is present in both our method and localized conformal prediction methods, it tends to be more severe for the latter in high-dimensional feature spaces. For instance, in a 100-dimensional space with a sample size of $n = 3000$, it may be difficult to find a sufficient number of neighbors (e.g., 100) near the test point $X_{n+1}$. In contrast, our method may assign these 3000 data points to a small number of clusters. Then regardless of the cluster membership of $X_{n+1}$, we would have 
 a sufficient number of data points to construct the PCP interval for $Y_{n+1}$.

\section{Technical proofs}\label{sect:proof_B}

\subsection{Proof of Proposition \ref{prop:example}}\label{sect:prop_example}

\begin{proof}
We denote the cluster membership of \(X_i\) by
\[
U_i:=\one\{X_i>0\}\sim \textup{Bern}(\rho),
\]
and let the subgroup sample sizes be
\[
N_1:=\sum_{i=1}^n U_i\sim \textup{Binomial}(n,\rho),
\qquad
N_0:=n-N_1.
\]
When $X_{n+1}\leq 0,$ we have
\begin{align*}
\sum_{i=1}^{n+1}w_i\one {\left\{R_i\geq R_{n+1}\right\}}  & =  \frac{\sum_{i=1}^{n+1}[(1-U_i)+bU_i]\one {\left\{R_i\geq R_{n+1}\right\}}  + c}{N_0+N_1 b +1+c} \\
& = \frac{N_0+1}{N_0+N_1 b +1+c} \times \frac{\sum_{i=1}^{n+1}(1-U_i)\one {\left\{R_i\geq R_{n+1}\right\}} + bN_1+c }{N_0+1}.
\end{align*}
The second equality uses the assumption that, when \(X_{n+1}\le 0\), all residuals from observations with \(X_i>0\) are larger than \(R_{n+1}\) with probability one.
Then, we have
\begin{equation}\label{equ:p_cover_0}
    \begin{split}
&\ \mathbb P\{ Y_{n+1}\not\in \hat C_{n} (X_{n+1}) \mid X_{n+1}\leq 0,N_0 \} \\
= &\  \mathbb P \left\{  \sum_{i=1}^{n+1}w_i\one {\left\{R_i\geq R_{n+1}\right\}}\leq  \alpha \ \Big| \  X_{n+1}\leq 0,N_0 \right\}  \\
= &\   \mathbb P \left\{  \frac{\sum_{i=1}^{n+1}(1-U_i)\one {\left\{R_i\geq R_{n+1}\right\}}  }{N_0+1}  \leq \alpha - \frac{(1-\alpha)(N_1 b+c) }{ N_0+1} \ \Big| \ X_{n+1}\leq 0,N_0 \right\} \\
\geq  &\ \alpha - \frac{(1-\alpha)(N_1 b+c) +1 }{ N_0+1} \\
=   &\ \alpha + (1-\alpha)b - \frac{(1-\alpha)[(n+1)b+c] + 1}{N_0+1}.
    \end{split}
\end{equation}
The inequality is obtained by conditioning on which of the \(N_0\) units satisfy \(X_i\le 0\).
For example, given that these \(N_0\) validation points and the test point all satisfy \(X_i\le 0\), the corresponding residuals
$R_1,\dotsc, R_{N_0},R_{n+1}$ are drawn from the same distribution. 
Then, the empirical average in the second equality is a valid $p$-value, which satisfies the inequality in the last line when all the residuals are distinct \citep[Theorem 2]{lei2018distribution}.  
We also note that
\[
\mathbb P \left\{ \alpha - \frac{(1-\alpha)(N_1 b+c) }{ N_0+1} \geq 0      \right\} = \mathbb P \left\{ N_1 \leq \frac{\alpha(n+1) -(1-\alpha)c}{\alpha + (1-\alpha)b}      \right\}.
\]
Since \(N_1/n\to \rho\), this probability tends to one whenever we choose the example parameters so that
$\rho < \alpha/[\alpha+(1-\alpha)b]$.
Under this choice, the lower bound in \eqref{equ:p_cover_0} is asymptotically tight up to the \(1/(N_0+1)\) discretization term.

Since $N_1=\sum_{i=1}^n U_i\sim \textup{Binomial}(n,\rho)$ and $N_0=n-N_1\sim \textup{Binomial}(n,1-\rho)$
we have
\[
\mathbb E\left[\frac{1}{N_0+1}\right]
=
\frac{1-\rho^{n+1}}{(n+1)(1-\rho)},
\qquad
\mathbb E\left[\frac{1}{N_1+1}\right]
=
\frac{1-(1-\rho)^{n+1}}{(n+1)\rho}.
\]
Then it follows from \eqref{equ:p_cover_0}  that
\begin{equation}\label{equ:x_minus}
\begin{split}
 &\ \mathbb P\{ Y_{n+1}\not\in \hat C_{n} (X_{n+1})  \text{ and } X_{n+1}\leq 0\} \\
\geq &\  (1-\rho) \times  \left(   \alpha + (1-\alpha)b - 
\left[(1-\alpha)[(n+1)b+c] + 1\right]\times \mathbb E\left[\frac{1}{N_0+1}\right]\right) \\
= &\ (1-\rho)\alpha - (1-\alpha) b\rho \left[1-\rho^n\right]  - \frac{(1-\alpha)c+1}{n+1}\left[1-\rho^{n+1}\right].
\end{split}
\end{equation}
Similarly, when \(X_{n+1}>0\), we have
\begin{align*}
\sum_{i=1}^{n+1}w_i\one {\left\{R_i\geq R_{n+1}\right\}}  & =  \frac{\sum_{i=1}^{n+1}[a(1-U_i)+U_i]\one {\left\{R_i\geq R_{n+1}\right\}}  + c}{N_0 a +N_1 +1+c} \\
& = \frac{N_1+1}{N_0 a+N_1 +1+c} \times \frac{\sum_{i=1}^{n+1}U_i\one {\left\{R_i\geq R_{n+1}\right\}} +c }{N_1+1}.
\end{align*}
The mis-coverage rate of $\hat C_{n} (X_{n+1})$ conditional on $X_{n+1}>0$ can be expressed as
\begin{equation}\label{equ:p_cover_1}
    \begin{split}
&\ \mathbb P\{ Y_{n+1}\not\in \hat C_{n} (X_{n+1}) \mid X_{n+1}>0, N_1 \} \\
= &\  \mathbb P \left\{  \sum_{i=1}^{n+1}w_i\one {\left\{R_i\geq R_{n+1}\right\}}\leq  \alpha \ \Big| \  X_{n+1}> 0, N_1 \right\}  \\
= &\   \mathbb P \left\{  \frac{\sum_{i=1}^{n+1}U_i\one {\left\{R_i\geq R_{n+1}\right\}}  }{N_1+1}  \leq  \alpha + \frac{ \alpha N_0 a - (1-\alpha)c }{ N_1+1} \ \Big| \ X_{n+1}> 0, N_1 \right\} \\
\geq  &\ \alpha + \frac{ \alpha N_0 a - (1-\alpha)c -1}{ N_1+1} =
\alpha -\alpha a + \frac{\alpha (n+1)a - (1-\alpha)c-1}{N_1+1}.
    \end{split}
\end{equation}
As explained in \eqref{equ:p_cover_0}, the inequality is tight when the probability
\[
\mathbb P \left\{  \alpha + \frac{ \alpha N_0 a - (1-\alpha)c }{ N_1+1}\geq 0      \right\} = \mathbb P \left\{  N_1 \geq \frac{-\alpha (n a+1) + (1-\alpha)c  }{\alpha (1-a)}   \right\},
\]
goes to 1 as $n$ increases.
Next,
\begin{equation}\label{equ:x_plus}
\begin{split}
 &\ \mathbb P\{ Y_{n+1}\not\in \hat C_{n} (X_{n+1})  \text{ and } X_{n+1}> 0\} 
 \\
\geq  &\ \rho \times  \left(  \alpha -\alpha a + \left[\alpha (n+1)a - (1-\alpha)c-1\right] \times \mathbb E \left[\frac{1}{N_1+1}\right] \right)    \\
= &\ \rho\alpha + \alpha a (1-\rho) \left[1- (1-\rho)^{n}\right]   - \frac{(1-\alpha)c+1}{n+1}\left[1-(1-\rho)^{n+1}\right].
\end{split}
\end{equation}
Using \eqref{equ:x_minus} and \eqref{equ:x_plus}, we can lower bound the mis-coverage rate of $\hat C_{n} (X_{n+1})$ as follows:
\begin{align*}
 \mathbb P\{ Y_{n+1}\not\in \hat C_{n} (X_{n+1}) \} 
\geq  &\ \alpha + \alpha a(1-\rho)\left[1-(1-\rho)^n\right] - 
(1-\alpha)b\rho\left[1-\rho^n\right] \\
& - 2[(1-\alpha)c+1 ]/[n+1].
\end{align*}
Up to terms that vanish exponentially fast in \(n\), this proves the claimed bound for the interval with the point mass included. We next verify the case \(c=-1\), where the point mass is removed.

When $c=-1,$ the interval  $\hat C_{n}(X_{n+1})$ is defined using the first $n$ data points. In other words, we have
\[
\hat C_{n}^{-} (X_{n+1}) = \left[\hat \mu(X_{n+1}) \pm Q_{1-\alpha} \left(\sum_{i=1}^{n}w_i^{-}\delta_{R_i}   \right)\right],
\]
where  \(w_i\) denotes the corresponding weight with \(c=0\).
Let \(D_0=N_0+N_1b\). When \(X_{n+1}\le0\), \(w_i^-=w_i(D_0+1)/D_0\).
Then
\begin{align*}
&\ \mathbb P\{ Y_{n+1}\not\in \hat C_{n}^- (X_{n+1}) \mid X_{n+1}\leq 0,N_0,N_1 \} \\
= &\  \mathbb P \left\{
\sum_{i=1}^{n}w_i^-\one {\left\{R_i\geq R_{n+1}\right\}}\leq  \alpha
\ \bigg| \  X_{n+1}\leq 0,N_0,N_1 \right\}  \\
= &\   \mathbb P \left\{
\sum_{i=1}^{n+1}w_i\one {\left\{R_i\geq R_{n+1}\right\}}\leq
\alpha' :=
\left[\alpha  +  \frac{1}{D_0}\right]\frac{D_0}{D_0+1}
\ \bigg| \  X_{n+1}\leq 0,N_0,N_1 \right\}  \\
\geq &\ \alpha'  - \frac{(1- \alpha')N_1 b +1}{ N_0+1}   \\
= &\ \alpha + (1-\alpha) b -  \frac{(1-\alpha)(n+1)b+\alpha }{N_0+1}.
\end{align*}
The inequality follows from \eqref{equ:p_cover_0} with \(c=0\) and \(\alpha\) replaced by \(\alpha'\).
Similarly, let \(D_1=N_0a+N_1\). When \(X_{n+1}>0\),
\(w_i^-=w_i(D_1+1)/D_1\). Then
\[
\mathbb P\{ Y_{n+1}\not\in \hat C_{n}^- (X_{n+1}) \mid X_{n+1}>0,N_0,N_1 \}
\geq
 \alpha - \alpha a + \frac{\alpha (n+1) a -\alpha }{N_1 + 1},
\]
using \eqref{equ:p_cover_1} with \(c=0\) and \(\alpha\) replaced by
$\alpha'=\left[\alpha+\frac{1}{D_1}\right]\frac{D_1}{D_1+1}.$
Then,
\begin{align*}
&\mathbb P\{ Y_{n+1}\not\in \hat C_{n}^- (X_{n+1}) \} 
= \ \mathbb E\left[
\mathbb P\{Y_{n+1}\not\in \hat C_n^-(X_{n+1})\mid X_{n+1},N_0,N_1\}
\right] \\
\geq &\ (1-\rho)\left[
\alpha+(1-\alpha)b
-
\big\{(1-\alpha)(n+1)b+\alpha\big\}
\mathbb E\left\{\frac{1}{N_0+1}\right\}
\right] \\
&\quad
+\rho\left[
\alpha-\alpha a
+
\{\alpha(n+1)a-\alpha\}
\mathbb E\left\{\frac{1}{N_1+1}\right\}
\right] \\
= &\ \alpha
+\alpha a(1-\rho)\left[1-(1-\rho)^n\right]
-(1-\alpha)b\rho\left[1-\rho^n\right] 
-\frac{\alpha}{n+1}\left[2-\rho^{n+1}-(1-\rho)^{n+1}\right].
\end{align*}
This is at least
\[
\alpha
+\alpha a(1-\rho)\left[1-(1-\rho)^n\right]
-(1-\alpha)b\rho\left[1-\rho^n\right]
-\frac{2\alpha}{n+1}.
\]
Up to terms that vanish exponentially fast in \(n\), this confirms that the proposition holds with \(c=-1\).
\end{proof}

\subsection{Proof of Proposition \ref{prop:oracle_coverage_2}}\label{sect:oracle_coverage_2}

\begin{proof}
By the definition \(w_i\), we have  $w_1,\dots, w_{n+1} \independent X_{n+1}\mid  \pi (X_{n+1})$. By the definition of $\hat C_{n}(X_{n+1} )$ in \eqref{equ:oracle}, we have
\[
Y_{n+1}\not\in \hat C_{n}(X_{n+1} )  \ \Longleftrightarrow\  \hat P_n:= \sum_{i=1}^{n+1} w_{i}\one {\left\{R_i\geq R_{n+1}\right\}}\leq \alpha.
\]
which shows that the miscoverage event $\{Y_{n+1}\not\in \hat C_{n}(X_{n+1} )\}$ only depends on $X_{n+1}$ through $R_{n+1}$ and $ \pi(X_{n+1})$ in the weights $w_1,\dots, w_{n+1}$. 
Under the model \eqref{equ:mixture},
\[
R_{n+1}\independent X_{n+1}\mid \pi(X_{n+1}).
\]
The probability of miscoverage given $\pi(X_{n+1})$ and $X_{n+1} $ is independent of $X_{n+1}:$
\[
\mathbb P \big\{ \hat P_n\leq \alpha
\,\big|\, X_{n+1},\pi(X_{n+1}) \big\}
=
\mathbb P \big\{ \hat P_n\leq \alpha
\,\big|\, \pi(X_{n+1}) \big\}.
\]
Since \(\pi(X_{n+1})\) is a function of \(X_{n+1}\), conditioning on \(X_{n+1}\) already conditions on \(\pi(X_{n+1})\). Hence,
\[
\mathbb P \big\{ \hat P_n\leq \alpha
\,\big|\, X_{n+1},\pi(X_{n+1}) \big\}
=
\mathbb P \big\{ \hat P_n\leq \alpha
\,\big|\, X_{n+1} \big\}.
\]
\end{proof}

\subsection{Proof of Theorem \ref{thm:validity_1} and Proposition  \ref{thm:pcp_cover}}

\begin{proof}
Below, we prove Proposition \ref{thm:pcp_cover}; Theorem \ref{thm:validity_1} follows by removing the hats and the imputation superscript \(y\), and by replacing \(J\) with \(K\).
Let   $z_i = (x_i,y_i)$ and $r_i = |y_i- \hat \mu (x_i)|.$ 
 Let $E_{\bm z}$ denote the event $\{Z_{1},\dotsc, Z_{n+1}\} = \mathcal{D}_{\bm z}:=  \{z_1,\dots, z_{n+1}\}  $.
Let \(l=(l_1,\dots,l_J)\in\mathbb N^J\) satisfy \(\sum_{k=1}^J l_k=m\).
Conditional on \(E_{\bm z}\), we write \(\hat\pi_k(x_i)\) for the membership probability assigned to \(x_i\) by the symmetric algorithm \(\mathcal A\) applied to the unordered dataset \(\mathcal D_{\bm z}\); this is the common value of \(\hat\pi_k^{y_j}(x_i)\) under the event \(Z_{n+1}=z_j\).

For any $j\in [n+1]$,
\begin{align*}
   &\  \mathbb P\{ Z_{n+1}=z_j \mid E_{\bm z}, \hat  \pi^{\star, Y_{n+1}}(X_{n+1}) =  l/m \}  \\
= &\ \frac{\sum_{\sigma :\sigma  (n+1) = j } p_{Z_{[n+1]}}(z_{\sigma (1)},\dots, z_{\sigma (n+1)})\prod_{k=1}^{J}\big[ \hat \pi_k^{y_j}(x_{j})\big]^{ l_k }  }
{ \sum_{i=1}^{n+1}\sum_{\sigma ':\sigma ' (n+1) = i} p_{Z_{[n+1]}}(z_{\sigma '(1)},\dots, z_{\sigma '(n+1)}) \prod_{k'=1}^{J}\big[\hat \pi_{k'}^{y_i}(x_i)\big]^{ l_{k'} }} \quad\quad(\text{\small By Bayes' rule}) \\
= &\  \frac{\sum_{\sigma :\sigma  (n+1) = j } \prod_{k=1}^{J}\big[\hat \pi_k^{y_j}(x_{j})\big]^{ l_k }}
{ \sum_{i=1}^{n+1}\sum_{\sigma':\sigma' (n+1) = i } \prod_{k'=1}^{J}\big[\hat \pi_{k'}^{y_i}(x_i)\big]^{ l_{k'} }}\hspace{25pt} (\text{\small By the exchangeability of  $Z_{[n+1]}$}) \\
 = &\  \frac{\prod_{k=1}^{J}\big[\hat \pi_k(x_{j})\big]^{ l_k }}
{ \sum_{i=1}^{n+1} \prod_{k'=1}^{J}\big[\hat \pi_{k'}(x_i)\big]^{ l_{k'} }}\hspace{40pt} (\text{\small By the symmetry of  $\mathcal{A}$}) \\
=  &\ \hat  w_{j}^\star(\{z_1,\dots, z_{n+1}\})   \hspace{35pt} (\text{\small By the definition of 
$\hat  w_{j}^\star$)}.
\end{align*}
The probability of mis-coverage can be rewritten as
\begin{equation}\label{equ:marginal_validity_step}
    \begin{split}
  &\  \mathbb P\{Y_{n+1}\not\in \hat C_n^{\textup{PCP}}(X_{n+1}) \mid	  E_{\bm z} , \hat  \pi^{\star, Y_{n+1}}(X_{n+1}) =  l/m	\}		\\
 = &\  \sum_{j=1}^{n+1} \mathbb P\{Y_{n+1}\not\in \hat C_n^{\textup{PCP}}(X_{n+1}) \mid	E_{\bm z}, \hat  \pi^{\star, Y_{n+1}}(X_{n+1}) = l/m, Z_{n+1}=z_j	\}		\\
 &\ \hspace{20pt}  \times	\mathbb P\{ Z_{n+1}=z_j \mid  E_{\bm z},	\hat  \pi^{\star, Y_{n+1}}(X_{n+1}) =  l/m \} 	\\
= &\  \sum_{j=1}^{n+1} \hat w_{j}^\star(\{z_1,\dots, z_{n+1}\}) \one \left\{\sum_{i=1}^{n+1}\hat w_{i}^\star ( \{z_1,\dots, z_{n+1}\})\one{ \{r_{i}\geq r_{j}\} } \leq \alpha  \right\} \leq \alpha.
    \end{split}
\end{equation}
The last equality uses the weight expression above and that $Y_{n+1}\not\in \hat C_n^{\textup{PCP}}(X_{n+1})$ is deterministic under the conditional event. 
More precisely, the non-random estimated membership probabilities \(\hat\pi=\hat\pi^{Y_{n+1}}\) are functions of \(\mathcal D_{\bm z}\).
As functions of \(\hat\pi^{\star,Y_{n+1}}(X_{n+1})=l/m\) and \(\mathcal D_{\bm z}\), the weights in our interval are also fixed.
Finally, $ Z_{n+1}=z_j$ fixes the test point, which does not affect the weights of our interval given $\hat  \pi^{\star,Y_{n+1}}(X_{n+1}) =  l/m$.
The last step is achieved by the deterministic inequality in Lemma A.1 in  \citet{harrison2012conservative}.

Following the proof of the same Lemma A.1, we lower bound the probability of miscoverage. 
Without loss of generality, we assume that $r_{1}<\cdots < r_{n+1}$. Otherwise, we can sort them from the smallest to the largest. Let $j^*$ be the smallest $j\in [n+1]$ such that $\sum_{i=1}^{n+1}\hat w_{i}^\star \one{ \{r_{i}\geq r_{j}\} } \leq \alpha $. If $j^*=1$, then
\[
\sum_{j=1}^{n+1} \hat w_j^\star \one \left\{\sum_{i=1}^{n+1}\hat w_{i}^\star \one{ \{r_{i}\geq r_{j}\} } \leq \alpha  \right\} = 1 > \alpha,
\]
so the desired lower bound is immediate. Otherwise, $j^*\geq 2$, and we have
\begin{align*}
&\ \sum_{j=1}^{n+1} \hat w_{j}^\star \one \left\{\sum_{i=1}^{n+1}\hat w_{i}^\star \one{ \{r_{i}\geq r_{j}\} } \leq \alpha  \right\} = \sum_{j=j^*}^{n+1}\hat w_j^\star 	= \sum_{j=j^*-1}^{n+1}\hat w_j^\star
- \hat w_{j^*-1}^\star >  \alpha -  \hat w_{j^*-1}^\star. 
\end{align*}
By the definition of $j^*$, 
$\sum_{j=j^*-1}^{n+1}\hat w_j^\star  > \alpha,$
which gives the inequality above.  Upper bounding $ \hat w_{j^*-1}^\star$ by $\max_{i\in [n+1]}\hat w_{i}^\star$ proves the claim as required. 
Averaging the upper bound in \eqref{equ:marginal_validity_step} over \(E_{\bm z}\) conditional on \(\hat\pi^\star(X_{n+1})=\hat\pi^{\star,Y_{n+1}}(X_{n+1})\) gives \eqref{equ:pcp_guarantee}. Similarly, averaging the deterministic lower bound over \(E_{\bm z}\) conditional on \(\hat\pi^\star(X_{n+1})\) gives \eqref{equ:pcp_guarantee_2}.
\end{proof}

\subsection{Proof of Theorem \ref{thm:length}}\label{sect:length}

The interval \(\hat C_n(X_{n+1})\) in \eqref{equ:oracle} is infinite when the weight assigned to \(\delta_{+\infty}\) satisfies \(w_{n+1}^\star \geq \alpha\). Writing the non-normalized weights as
\(
\phi(X_i,X_{n+1})
=
\prod_{k=1}^{K}\big[\pi_k(X_i)\big]^{L_k^\star}.
\)
This is equivalent to
\[
\phi(X_{n+1},X_{n+1})
\geq
\frac{\alpha}{1-\alpha}
\sum_{i=1}^{n}\phi(X_i,X_{n+1}).
\]
By the expression of $\phi(X_{i},X_{n+1})$ in \eqref{equ:likelihood}, that is when
\begin{equation}\label{equ:equiv}
\hspace{-5pt}
H := -\sum_{k=1}^{K} L_k^\star \log \pi_{k}(X_{n+1})
\leq 
-\log\left(
\frac{\alpha}{1-\alpha}
\sum_{i=1}^{n}
\prod_{k=1}^{K}[\pi_k(X_i)]^{L_k^\star}
\right)
:= \epsilon(\alpha,X_{[n]}, L^\star).
\end{equation}
When $ L^\star$ is fixed, the random variable $H$ does not have a known distribution due to the mutual dependence between $\pi_j(X_{n+1})$ and $\pi_k(X_{n+1})$ for $j\neq k$. Nevertheless, we can compute the moment-generating function (MGF) of 
$H$ as follows:
\begin{align*}
  \mathbb E[e^{tH}\mid  L^\star ]   = &\  \mathbb E\left[\prod_{k=1}^{K} [\pi_{k}(X_{n+1})]^{-t L_k^\star} \ \Bigg| \  L^\star \right]  \\
  = &\    \int       \prod_{k=1}^{K} \pi_k^{-t L_k^\star}            \textup{Dir}( \pi ;\Lambda +  L^\star ) d \pi   \\
= &\ \frac{1}{ \textup{B}(\Lambda + L^\star ) } \int       \prod_{k=1}^{K} \pi_{k}^{-t L_k^\star}           \prod_{k=1}^{K} \pi_{k}^{ \Lambda_k +  L_k^\star-1}   d \pi          \\
= &\  	\textup{B}(\Lambda +[1-t] L^\star )	/\textup{B}(\Lambda + L^\star )		\\
  =&\  \frac{\Gamma( \bar \Lambda + m )}{\Gamma( \bar \Lambda + [1-t]m )}\prod_{k=1}^{K}\frac{\Gamma( \Lambda_k + [1-t] L_k^\star )}{\Gamma( \Lambda_k +  L_k^\star )},
\end{align*}
where $\bar \Lambda =  \sum_{k=1}^{K}\Lambda_k.$ By minimizing  $H$ in terms of $\pi(X_{n+1})$ subject to the constraints of the probability simplex, we can show that
$H$ is lower bounded by 
\begin{equation}\label{equ:h_min}
H_{\min}  := -\sum_{k=1}^{K} L_k^\star \log \pi_{k}^\star(X_{n+1}) =  - m \sum_{k=1}^{K} \pi_{k}^\star(X_{n+1}) \log  \pi_{k}^\star(X_{n+1}),
\end{equation}
which is non-random given $ L^\star$. The MGF of $H - H_{\min}$ is given by
\begin{equation}\label{equ:mgf_h}
\begin{split}
M_{H - H_{\min}}(t) & = \mathbb E\big[e^{t[H-H_{\min}]}\mid  L^\star \big]  \\
& = 
\prod_{k=1}^{K}[ \pi_k^\star(X_{n+1})]^{t L_k^\star}
 \frac{\Gamma( \bar \Lambda + m )}{\Gamma( \bar \Lambda + [1-t]m )}\prod_{k=1}^{K}\frac{\Gamma( \Lambda_k + [1-t] L_k^\star )}{\Gamma( \Lambda_k +  L_k^\star )}.
 \end{split}
\end{equation}
We let $M^{(j)}$ denote the $j$-th order derivative of  $M_{H - H_{\min}}(t) $ as follows:
\[
M^{(0)} \equiv M_{H-H_{\min}}(t) \quad  \text{and} \quad M^{(j)} := d^j M_{H-H_{\min}}(t)/dt^j \  \text{ for } j\in \mathbb N^+. 
\]
The following is the key lemma used to prove Theorem \ref{thm:length}. We defer its proof to \Cref{sect:lemma1}.
Ignoring the error term, the lemma shows that the moments of  $H-H_{\min}$ match the moments of the exponential distribution $\text{Exp}(1+\bar \Lambda/m)$. 
\begin{lemma}\label{lemma:thm_2}
	For $b\in \mathbb N^+$, the $b$-th moment of $H-H_{\min}$ can be written as 
\[
\mathbb E\left[(H-H_{\min})^b\mid  \pi^\star (X_{n+1}) = \Lambda/\bar \Lambda  \right] = \frac{b!}{\left(1+\bar \Lambda/m \right)^{b} } + O(1/m).
\]
\end{lemma}

Using \Cref{lemma:thm_2}, we can write the characteristic function of $H-H_{\min}$ as 
\begin{align*}
\varphi_{H-H_{\min}}(t) = 1 + \sum_{b=1}^{\infty} \frac{i^b \mathbb E\left[(H-H_{\min})^b\mid \pi^\star (X_{n+1}) = \Lambda/\bar \Lambda \right]}{b!}t^b  = \frac{m+\bar \Lambda }{m+\bar \Lambda - m i t} + O (1/m).
\end{align*}
The CDF of $H-H_{\min}$ is given by 
\[
F_{H-H_{\min}}(h)  =\int_{0}^{h} \frac{1}{2\pi }\int_{-\infty}^{\infty}\exp(-ith')	\varphi_{H-H_{\min}}(t)	dtdh' = 1-\exp \{-(1+\bar \Lambda/m)h\} + O(1/m).
\]
Then, by the equivalence and the definition of \(\epsilon(\alpha,X_{[n]}, L^\star)\) in \eqref{equ:equiv}, we have
\begin{equation}\label{equ:exp_approx}
\begin{split}
&\ \mathbb P \big\{  |\hat C_n (X_{n+1})|  < \infty  	\ \big|	\	X_{[n]}, \pi^\star(X_{n+1}) = \Lambda/\bar \Lambda    \big\} \\
= &\  \mathbb P \left\{  H - H_{\min} \geq  \epsilon(\alpha,X_{[n]}, L^\star )	- H_{\min} \mid 	X_{[n]}, \pi^\star(X_{n+1}) = \Lambda/\bar \Lambda   \right\}			\\
= &\  \exp\left\{ (1+\bar \Lambda/m) \log\left(\frac{\alpha}{1-\alpha}\sum_{i=1}^{n}\prod_{k=1}^{K}[\pi_k(X_{i})]^{ L_k^\star} \right) + (1+\bar \Lambda/m)H_{\min}\right\}\wedge 1 + O(1/m)	\\
= &\ \left\{\frac{\alpha}{1-\alpha}\sum_{i=1}^{n}\prod_{k=1}^{K}[\pi_k(X_{i})]^{ L_k^\star} /\prod_{k=1}^{K}[ \pi_k^\star(X_{n+1})]^{ L_k^\star }\right\}^{1+\bar \Lambda/m }\wedge 1+ O(1/m)	\\
= &  \left\{\frac{\alpha}{1-\alpha}\sum_{i=1}^{n}\exp \left(\sum_{k=1}^{K} L_k^\star \log \frac{\pi_k(X_{i})}{ \pi_{k}^\star(X_{n+1})}\right)\right\}^{1+\bar \Lambda/m }\wedge 1+ O(1/m)\\
= &  \min \left\{\frac{\alpha}{1-\alpha}\sum_{i=1}^{n}\exp \left(-m D_{\textup{KL}}\big( \pi^\star(X_{n+1})\| \pi(X_{i})  \big)  \right),1\right\}^{(m+\bar \Lambda)/m }+ O(1/m),
\end{split}
\end{equation}
where the third equality is obtained by the definition of $H_{\min}$ in \eqref{equ:h_min}.

\subsubsection{Numerical experiments for Theorem \ref{thm:length}}\label{sect:sim4}

The proof of Theorem \ref{thm:length} is based on \Cref{lemma:thm_2}, which shows that $H-H_{\min}$ (defined in \eqref{equ:equiv} and \eqref{equ:h_min}) approximately follows an exponential distribution $\text{Exp}(1+\bar \Lambda/m).$
We next verify our theory on three different Dirichlet distributions. When $ \pi^\star (X_{n+1}) = \Lambda/\bar\Lambda,$
the posterior distribution of $\pi (X_{n+1})$ given $ \pi^\star (X_{n+1})$ can be expressed as $\text{Dir}(\Lambda + m\Lambda/\bar\Lambda)$. 
We draw 5000 random samples from the posterior to construct a Monte-Carlo approximation of the distribution of  $H-H_{\min}$. \Cref{fig:hist} shows that the Monte-Carlo approximation aligns well with $\text{Exp}(1+\bar \Lambda/m)$ for $m=20$ and three different values of $\Lambda.$ This result confirms that the first term in \eqref{equ:exp_approx} gives an accurate approximation of the probability that our interval has a finite length. 

\begin{figure}[t]
     \centering
     \begin{subfigure}[b]{0.4\textwidth}
         \centering         \includegraphics[width=\textwidth]{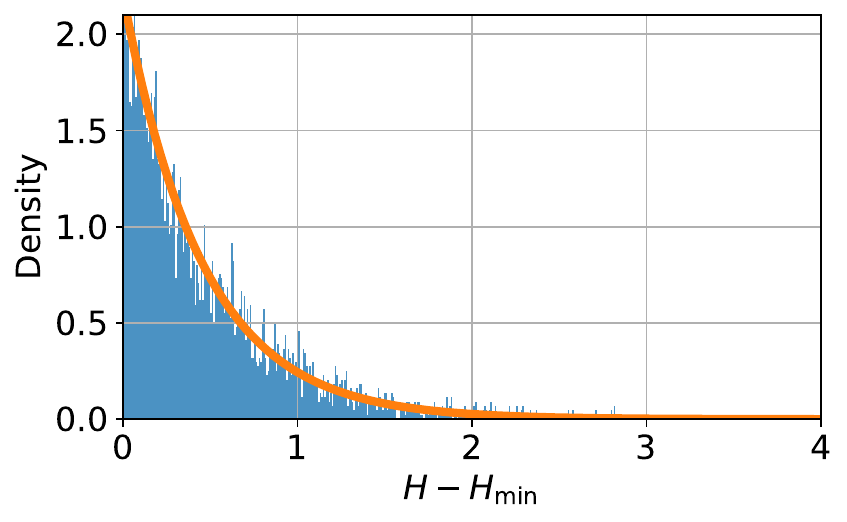}
         \caption{$\Lambda  = (2,2,20).$
         }
           \label{fig:hist0}
     \end{subfigure}
     \vspace{5pt}
     \hspace{15pt}
     \begin{subfigure}[b]{0.4\textwidth}
         \centering
         \includegraphics[width=\textwidth]{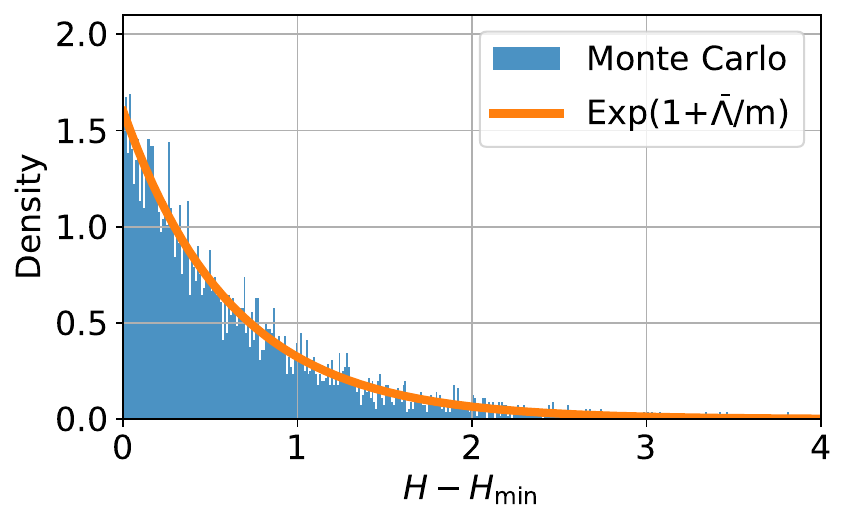}
          \caption{$ \Lambda = (2,5,5).$}
           \label{fig:hist1}
     \end{subfigure}
      \vspace{5pt}
          \begin{subfigure}[b]{0.4\textwidth}
         \centering
         \includegraphics[width=\textwidth]{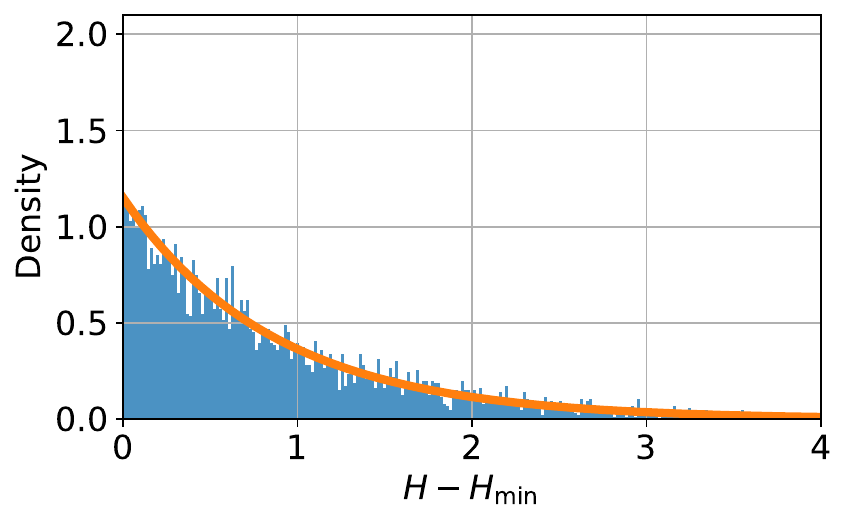}
        \caption{$\Lambda = (1,1,1).$
         }
          \label{fig:hist2}
            \end{subfigure}
        \caption{The exponential distribution $\text{Exp}(1+\bar \Lambda/m)$ and    
        the Monte-Carlo simulations of $H-H_{\min}$  under the Dirichlet posterior $\text{Dir}(\Lambda+ m\Lambda/\bar\Lambda )$ with $m=20$.
         }
        \label{fig:hist}
\end{figure}

\subsubsection{Proof of Lemma \ref{lemma:thm_2}}\label{sect:lemma1}

\begin{proof}
The derivative of each term in $M_{H - H_{\min}}(t) $  in \eqref{equ:mgf_h} involves itself. First,
\[
\frac{d}{dt}\left( [\pi_k^\star(X_{n+1})]^{t L_k^\star}\right) = 
[ \pi_k^\star(X_{n+1})] ^{t L_k^\star} \times  L_k^\star\ln \pi_k^\star(X_{n+1}). 
\]
The derivative of the gamma-ratio term involving \(\bar\Lambda\) is
\[
\frac{d}{dt}\left( \frac{\Gamma( \bar \Lambda + m )}{\Gamma( \bar \Lambda + [1-t]m )}\right) = 
\frac{\Gamma( \bar \Lambda + m )}{\Gamma( \bar \Lambda + [1-t]m )} \times m\psi(\bar \Lambda  + [1-t]m ),
\]
where $\psi(v) = d\ln \Gamma (v)/dv$ is  the digamma function. Similarly, we have
\[
\frac{d}{dt}\left( \frac{\Gamma( \Lambda_k + [1-t] L_k^\star )}{\Gamma( \Lambda_k +  L_k^\star )}\right) = 
-\frac{\Gamma( \Lambda_k + [1-t] L_k^\star )}{\Gamma( \Lambda_k +  L_k^\star )}\times  L_k^\star \psi( \Lambda_k  + [1-t] L_k^\star ).
\]
Taking the logarithmic derivative of \(M(t)=M_{H-H_{\min}}(t)\) in \eqref{equ:mgf_h}, and using the three derivative identities above, gives
\begin{equation}\label{equ:first_dev}
M^{(1)}(t)=G^{(0)}(t)M^{(0)}(t),
\end{equation}
where the function $G^{(0)}$ is defined as
\[
G^{(0)}(t) =\sum_{k=1}^{K} L_k^\star \log \pi_{k}^\star(X_{n+1}) + m \psi(\bar \Lambda + [1 - t]m)- \sum_{k=1}^{K} L_k^\star \psi(\Lambda_k +[1-t]L_k^\star).
\]
It is well-known that $\psi$ can be approximated as follows:
\[
\psi(v) = \ln (v) - \frac{1}{2v} - \frac{1}{12v^2} + \frac{1}{120v^4}-\frac{1}{252v^6}+\cdots \approx \ln (v)-\frac{1}{2v}.
\]
Then, we approximate the above expression of $G^{(0)}$ as 
\begin{equation}\label{equ:g_0}
\begin{split}
G^{(0)}  &   \approx \sum_{k=1}^{K} L_k^\star \log  \pi_{k}^\star (X_{n+1})  + \sum_{k=1}^{K} L_k^\star \ln\left(	\frac{\bar \Lambda + (1-t)m}{\Lambda_k + (1-t)L_k^\star}	\right)	\\	
& \hspace{15pt} +\sum_{k=1}^{K}\left(\frac{ L_k^\star}{2[ \Lambda_k + [1 - t] L_k^\star]} \right) -\frac{m}{2[\bar \Lambda + (1-t)m]} \\
& =  \sum_{k=1}^{K} L_k^\star \log \pi_{k}^\star (X_{n+1})  + \sum_{k=1}^{K} L_k^\star \ln 
\left(
\frac{\bar \Lambda/m + (1-t)}{\Lambda_k/L_k^\star + (1-t)}
\times \frac{1}{\pi_k^\star(X_{n+1})}
\right)
 \\
&\hspace{15pt}  + \frac{(K-1)m}{2\bar \Lambda + 2[1 - t]m}  = \frac{(K-1)m}{2\bar \Lambda + 2[1 - t]m}.
\end{split}
\end{equation}
The last equality is due to $\bar \Lambda/m=\Lambda_k/ L_k^\star$, which follows from the condition that  $\pi_k^\star(X_{n+1}) =   L_k^\star/m =  \Lambda_k/\bar\Lambda$.

The approximation error can be written as
\begin{align*}
G^{(0)} -  \frac{(K-1)m}{2\bar \Lambda + 2[1 - t]m} =& \sum_{j=1}^\infty C_j	 \left\{\sum_{k=1}^{K}\left(\frac{ L_k^\star}{[ \Lambda_k + [1 - t] L_k^\star]^{2j}} \right) -\frac{m}{[\bar \Lambda + (1-t)m]^{2j}} \right\}\\
= & \sum_{j=1}^\infty C_j	 \left\{ \sum_{k=1}^{K}\frac{1}{ (L_{k}^\star)^{2j-1}} - \frac{1}{m^{2j-1}}  \right\}\frac{m^{2j}}{[\bar \Lambda + (1-t)m]^{2j}}\\
= & \sum_{j=1}^\infty C_j	 \left\{ \sum_{k=1}^{K}(\bar \Lambda/\Lambda_k)^{2j-1}- 1  \right\}\frac{m }{[\bar \Lambda + (1-t)m]^{2j}}		=  O(1/m),
\end{align*}
where \(C_1=-1/12, C_2=1/120,\dots\) are the coefficients of the \(v^{-2j}\) terms in the expansion of \(\psi\). The  second equality is obtained by $\bar \Lambda/ m= \Lambda_k/ L_k^\star$:
\begin{align*}
 \frac{ L_k^\star}{[ \Lambda_k + [1 - t] L_k^\star]^{2j}} & =  \frac{1}{ (L_k^\star)^{2j-1}}\times \left[\frac{m / \bar \Lambda}{ 1 + [1 - t]m/ \bar \Lambda}\right]^{2j}  = \frac{1}{ (L_k^\star)^{2j-1}}\times\frac{m^{2j}}{[\bar \Lambda + (1-t)m]^{2j}}.
\end{align*}
When $K=3$, using \eqref{equ:g_0}, the $b$-th order derivative of $G^{(0)} $ is
\begin{equation}\label{equ:G_b}
\begin{split}
G^{(b)} & =    \sum_{j=1}^\infty C_j	 \left\{ \sum_{k=1}^{K}\frac{1}{ (L_{k}^\star)^{2j-1}} - \frac{1}{m^{2j-1}}  \right\}\frac{(2j+b-1)!m^{b+1}}{(2j-1)![\bar \Lambda + (1-t)m]^{2j+b}} \\
& \hspace{15pt} + \frac{d^b}{dt^b}\left( \frac{m}{\bar \Lambda +  [1 - t]m}\right)	\\
& = \frac{b! m^{b+1}}{(\bar \Lambda + [1 - t]m)^{b+1}} + O (1/m).
\end{split}
\end{equation}
Then, $M^{(1)}(t)$ in \eqref{equ:first_dev} can be written as
\begin{equation}\label{equ:m_1}
M^{(1)}(t) = \frac{mM^{(0)} }{\bar \Lambda + [1 - t]m} + O(1/m).
\end{equation}
By $M^{(0)}(0) = M_{H-H_{\min}}(0)=1,$ we can express the first moment of $H-H_{\min}$ as
\[
\mathbb E\left[H-H_{\min}\mid X_{[n]},\pi^\star(X_{n+1}) = \Lambda/\bar \Lambda  \right]= M^{(1)}(0) =  \frac{m }{\bar \Lambda + m} + O(1/m),
\]
which proves the claim for $b=1.$ We next use induction to prove the claim for $b=2,3,\dots.$ 
More precisely, we will show that for any $b\in \mathbb N^+,$
\begin{equation}\label{equ:M_b}
M^{(b)} = b! \left(\frac{m }{\bar \Lambda + [1-t]m}\right)^{b}M^{(0)}  + O(1/m).
\end{equation}
We first compute the second, third and fourth-order derivative of $M_{H-H_{\min }}(t)$ to illustrate the idea.
As shown in  the first and second row in \Cref{fig:gradient}, we differentiate $M^{(0)}$ and $G^{(0)}$ respectively to obtain the expression of $M^{(2)}$ as follows:
\[
M^{(2)}  = G^{(1)}M^{(0)} +  G^{(0)}M^{(1)}   = 2 \left(\frac{m }{\bar \Lambda + [1-t]m}\right)^{2}M^{(0)}  + O(1/m),
\]
by the expressions of $G^{(b)}$ in \eqref{equ:G_b} and $M^{(1)}$ in \eqref{equ:m_1}. Similarly, in the third and fourth rows of the same figure, further differentiating $M^{(b)}$ gives that 
\begin{align*}
M^{(3)} & = G^{(2)}M^{(0)} +  2 G^{(1)}M^{(1)} + G^{(0)}M^{(2)}, \\
M^{(4)} &  = G^{(3)}M^{(0)} + 3G^{(2)}M^{(1)} + 3G^{(1)}M^{(2)} + G^{(0)}M^{(3)}. 
\end{align*}
\begin{figure}[t]
\centering
         \includegraphics[width=0.97\textwidth]{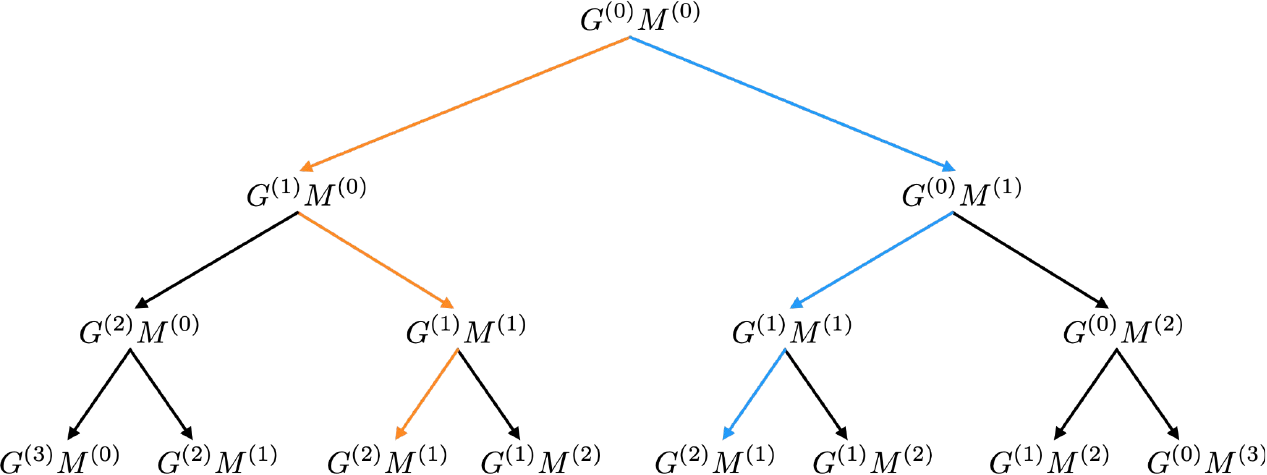}
                  \vspace{7pt}
        \caption{The paths of computing the $j$-th order derivatives of $M_{H-H_{\min}}{(t)}$ for $j\in [4].$}
            \vspace{2pt}
        \label{fig:gradient}
\end{figure}
Because of the  multiplicative structure in \eqref{equ:first_dev}, the expression of $M^{(b)} $ involves terms 
\[
G^{(c)}M^{(d)}, \forall c,d\in \{0,1,2,\dots\} \text{ such that } c+d=b-1.
\]
In other words, we can write $G^{(c)}M^{(d)} = G^{(c)}M^{(b-1-c)}$ in the expression of $M^{(b)}.$ From $G^{(0)}M^{(0)}$ in the first row to one of the $ G^{(c)}M^{(b-1-c)}$'s in the $b$-th row in \Cref{fig:gradient}, we need to
 go through the derivatives $G^{(1)},G^{(2)},\dots, G^{(c-1)}$ in some previous rows. For example, the two coloured paths in the figure highlight two of the paths for obtaining the term   $ G^{(2)}M^{(1)}$ in the expression of $M^{(4)}.$ In the yellow path, $G^{(1)}$ appears in the second row, and $G^{(2)}$ appears in the fourth row, while in the blue path, $G^{(1)}$ appears in the second row and $G^{(2)}$ appears in the fourth row. In general, there are  $\binom{b-1}{c}$ paths to obtain the term $ G^{(c)}M^{(b-1-c)}$ in the $b$-th row. Therefore, we have the expression,
\begin{equation}\label{equ:express_m_b}
 M^{(b)} = \sum_{c=0}^{b-1} \binom{b-1}{c} G^{(c)}M^{(b-1-c)}.
\end{equation}
It is straightforward to show that \eqref{equ:express_m_b} holds for the expression of $ M^{(b)} $ above for $b\in [4].$ Suppose that the expression holds for $M^{(b)}.$ 
Differentiating the expression  gives 
\begin{align*}
 M^{(b+1)}  =  & \sum_{c=0}^{b-1} \binom{b-1}{c} G^{(c+1)}M^{(b-1-c)} +  \sum_{c=0}^{b-1} \binom{b-1}{c} G^{(c)}M^{(b-c)} \\
=  & \sum_{c=1}^{b} \binom{b-1}{c-1} G^{(c)}M^{(b-c)} +  \sum_{c=0}^{b-1} \binom{b-1}{c} G^{(c)}M^{(b-c)} 	\\
=  & \sum_{c=1}^{b-1}\left[ \binom{b-1}{c-1}+  \binom{b-1}{c} \right]G^{(c)}M^{(b-c)} + G^{(b)}M^{(0)}+ G^{(0)}M^{(b)} \\
=  & \sum_{c=0}^{b}\binom{b}{c}G^{(c)}M^{(b-c)}, 
\end{align*}
which proves the expression \eqref{equ:express_m_b} holds for $b+1$. 

Similarly, substituting  $G^{(c)}$
in \eqref{equ:G_b} and $M^{(b-c)}$ in \eqref{equ:M_b}
into \eqref{equ:express_m_b}, we have
\begin{align*}
 M^{(b+1)}  & =  \sum_{c=0}^{b}\binom{b}{c}\Bigg[\frac{c! m^{c+1}}{(\bar \Lambda + [1 - t]m)^{c+1}} \Bigg] \left[(b-c)! \left(\frac{m }{\bar \Lambda + [1-t]m}\right)^{b-c} M^{(0)} \right] + O (1/m) \\
 & = \sum_{c=0}^{b}\binom{b}{c}c!(b-c)!\left(\frac{m }{\bar \Lambda + [1-t]m}\right)^{b+1} M^{(0)} + O (1/m) \\
& = (b+1)!\left(\frac{m }{\bar \Lambda + [1-t]m}\right)^{b+1} M^{(0)} + O (1/m),
\end{align*}
which proves that the expression \eqref{equ:M_b} holds for $b+1$ by induction.
Taking \(t=0\) in \eqref{equ:M_b} and using \(M^{(0)}(0)=1\), we obtain
\[
M^{(b)}(0)
=
b!\left(\frac{m}{m+\bar \Lambda}\right)^b
+O(1/m)
=
\frac{b!}{(1+\bar \Lambda/m)^b}+O(1/m),
\]
which proves the lemma.
\end{proof}

\subsection{Proof of Theorem \ref{thm:approximate_valid_2}}\label{sect:valid_2}

Below, we prove the theorem for the weighted interval
\[
\hat C_n(X_{n+1}) =
\left\{y:\sum_{i=1}^{n+1}\hat w_i \one \{ R_i\geq R_{n+1}^{y}\}>\alpha \right\},
\]
where the weights \(\hat w_i=\hat w_i(X_{[n+1]})\) are nonnegative, sum to one, and are fixed conditional on \(X_{[n+1]}\).
This covers the KL-based weights in Theorem \ref{thm:approximate_valid_2}, 
since the map \(\hat\pi\) used to construct the weights is fixed and does not depend on \(R_{[n+1]}\).
We also use that \(\hat w_{n+1}\ge \hat w_j\) for all \(j\in[n]\), which holds for the KL-based weights used in the theorem.

\begin{proof} By the definition of $\hat C_n(X_{n+1})$ above, we have
\begin{align*}
Y_{n+1}\not\in \hat C_n(X_{n+1})  ~ \Leftrightarrow~  \sum_{i=1}^{n+1} \hat w_i\one\left\{R_i\geq R_{n+1} \right\}	\leq  \alpha.
\end{align*}
By the deterministic inequality in Lemma A.1 in \citet{harrison2012conservative},
we can upper bound the coverage gap $\mathbb P \big\{Y_{n+1}\not\in \hat C_n(X_{n+1}) \ \big|	\ X_{[n+1]}	\big\}-\alpha	$ by
\begin{equation*}
\begin{split}
	&~ \mathbb{E}\Bigg[ \one\left\{\sum_{i=1}^{n+1} \hat w_i\one\left\{R_i\geq R_{n+1} \right\}	\leq  \alpha \right\}  - \sum_{j=1}^{n+1} \hat w_j \one \left\{	\sum_{i'=1}^{n+1} \hat w_{i'}\one\left\{R_{i'}\geq R_{j}\right\}	\leq \alpha 	\right\}	~\Bigg|~ X_{[n+1]}	\Bigg]		\\
= &~ \sum_{j=1}^{n+1}\mathbb{E}\Bigg[ \hat w_j \Bigg(\one\left\{\sum_{i=1}^{n+1} \hat w_i\one\left\{R_i\geq R_{n+1} \right\}	\leq  \alpha \right\}  -  \one \left\{	\sum_{i'=1}^{n+1} \hat w_{i'}\one\left\{R_{i'}\geq R_{j}\right\}	\leq \alpha 	\right\}	\Bigg)~\Bigg|~ X_{[n+1]}	\Bigg]	\\
\equiv&~
\sum_{j=1}^{n+1}
\mathbb E\Big[
\hat w_j(X_{[n+1]})
\epsilon_j(X_{[n+1]},R_{[n]\setminus\{j\}})
\mid X_{[n+1]}
\Big],
\end{split}
\end{equation*}
For \(j\in[n]\), the term $\epsilon_j(X_{[n+1]},R_{[n]\setminus\{j\}})$ can be rewritten as
\begin{equation}\label{equ:gap}
\begin{split}
\epsilon_j(X_{[n+1]},R_{[n]\setminus\{j\}})
& :=
\mathbb E\Bigg[
\one\left\{\sum_{i=1}^{n+1} \hat w_i\one\{R_i\ge R_{n+1}\}\le \alpha\right\}
-
\one\left\{\sum_{i'=1}^{n+1} \hat w_{i'}\one\{R_{i'}\ge R_j\}\le \alpha\right\}
\\
&\hspace{28pt} \,\Bigg|\,
X_{[n+1]},R_{[n]\setminus\{j\}}
\Bigg].
\end{split}
\end{equation}
For \(j=n+1\), we have \(\epsilon_{n+1}=0\).
Under the mixture model \eqref{equ:mixture}, we can rewrite each residual $R_j$ as
\begin{equation}\label{equ:r_representation}
R_j \stackrel{d}{=} \sum_{k=1}^{K}\one\{U_j=k\}\tilde R_{j,k},
\end{equation}
where \(U_j\sim \sum_{k=1}^{K}\pi_k(X_j)\delta_k\) and
\(\tilde R_{j,k}\sim f_k\) independently across \(j\) and \(k\).
Below, \(\tilde R_k\) denotes a generic independent draw from \(f_k\), independent of \(R_{[n]\setminus\{j\}}\).
Denote
\begin{align*}
&\Delta(R_{[n]\setminus\{j\}},r) := \alpha-\sum_{i'\in[n]\setminus\{j\}}\hat w_{i'}\one\{R_{i'}\ge r\}.\\
&   \mathbb{P}_{k',k}^{-j}\left\{\cdot\right\}  = \mathbb{P}\left\{\cdot \hspace{-3pt}\mid U_{n+1}=k',U_{j}=k, X_{[n+1]},R_{[n]	\setminus	\{j\}}\right\}.
\end{align*}
In the following display, we use the convention that
\(\mathbb P_{k',k}^{-j}\) conditions on \(U_{n+1}=k'\) and \(U_j=k\).
The term $\epsilon_j(X_{[n+1]},R_{[n]\setminus\{j\}})$ in \eqref{equ:gap} can be rewritten as 
\begin{align*}
 &	\sum_{k=1}^{K}\pi_k(X_{n+1}) \sum_{k'=1}^{K}\pi_{k'}(X_{j})\mathbb P_{k,k'}^{-j}\left\{ \hat w_{j}\one \big\{R_{j}\geq R_{n+1} \big\} + \hat w_{n+1}	\leq \Delta(R_{[n]\setminus\{j\}},R_{n+1})	\right\}										\\
&	- \sum_{k=1}^{K}\pi_k(X_{j}) \sum_{k'=1}^{K}\pi_{k'}(X_{n+1}) \mathbb P_{k',k}^{-j}\left\{   \hat w_{j} +  \hat w_{n+1}	\one \big\{ R_{n+1} \geq R_j \big\} \leq 		\Delta(R_{[n]\setminus\{j\}},R_j) 		\right\}		\\
\leq &\sum_{k=1}^{K}\sum_{k'=1}^{K}
\big[\pi_k(X_{n+1})\pi_{k'}(X_j)
-\pi_k(X_j)\pi_{k'}(X_{n+1})\big] \\
&\hspace{45pt}\times
\mathbb P_{k',k}^{-j}\left\{
\hat w_j+\hat w_{n+1}\one\{R_{n+1}\ge R_j\}
\leq \Delta(R_{[n]\setminus\{j\}},R_j)
\right\}
\\
\leq & \sum_{k'=1}^{K} | \pi_{k'}(X_{j})-\pi_{k'}(X_{n+1}) |
\mathbb P \left\{
\hat w_j + \hat w_{n+1}\one\{\tilde R_{k'} \ge R_{n+1}\}
\leq \Delta(R_{[n]\setminus\{j\}},R_{n+1})
\ \Big|\ X_{[n+1]},R_{[n]\setminus\{j\}}
\right\} \\
& + \sum_{k=1}^{K} |\pi_k(X_{n+1})-\pi_k(X_j)|
\mathbb P \left\{
\hat w_j + \hat w_{n+1}\one\{R_{n+1}\ge \tilde R_k\}
\leq \Delta(R_{[n]\setminus\{j\}},\tilde R_k)
\ \Big|\ X_{[n+1]},R_{[n]\setminus\{j\}}
\right\}\\
\leq &\ 2\sum_{k=1}^{K} |\pi_k(X_{n+1})	 -\pi_k(X_{j})|  \mathbb P\left\{ 	\hat w_{j}  \leq 		\Delta(R_{[n]\setminus\{j\}},\max\{R_{n+1},\tilde R_k\}) 	\ \Big| \  X_{[n+1]},R_{[n] \setminus \{j\}} \right\} \\
 \leq  &\ 2  \sum_{k=1}^{K} |\pi_k(X_{n+1})	 -\pi_k(X_{j})|.
\end{align*}
The first inequality uses \(\hat w_{n+1}\ge \hat w_j\) for all \(j\), after relabeling the two residual draws according to the convention for \(\mathbb P_{k',k}^{-j}\):
\begin{align*}
&\ \big( \hat w_{n+1}- \hat w_j\big)\big(	1-   \one\big\{ R_{n+1} \geq R_j \big\}    \big) \geq 0  \\
  \Leftrightarrow &\   \hat w_j \one\big\{ R_{n+1} \geq R_j \big\}+  \hat w_{n+1} \geq   \hat w_j+ \hat w_{n+1} \one\big\{ R_{n+1} \geq R_j \big\}.
\end{align*}
The second inequality is achieved by applying
\begin{align*}
&\ \pi_k(X_{n+1}) \pi_{k'}(X_{j})	 -\pi_k(X_{j}) \pi_{k'}(X_{n+1})	\\
\leq &\  |\pi_{k'}(X_{j}) - \pi_{k'}(X_{n+1})|  \pi_k(X_{n+1})  + |\pi_k(X_{n+1})	 -\pi_k(X_{j})|  \pi_{k'}(X_{n+1}), 
\end{align*}
along with \eqref{equ:r_representation}  and the definition of $ \mathbb P_{k',k}^{-j}.$
Finally, substituting the upper bound on $\epsilon_j(X_{[n+1]},R_{[n]\setminus\{j\}})$ into the decomposition above proves the claim.
\end{proof}

\subsection{Proof of Theorem \ref{thm:approximate_valid_3}}

\begin{proof}
The proof below follows the proof of Theorem 6 in \cite{barber2022conformal}.
Under the mixture model \eqref{equ:mixture}, we have
\[
R_j\stackrel{\textup{d}}{=} \sum_{k=1}^{K}\one\{U_j=k\}\tilde R_{j,k},
\]
where \(U_j\sim \sum_{k=1}^{K}\pi_k(X_j)\delta_k\), and
\(\tilde R_{j,k}\sim f_k\) independently across \(j\) and \(k\), independently of the \(U_j\)'s.

Let $B_{i} = \one\{U_i=U_{n+1}\}$.  We have
\begin{equation}\label{equ:def_phi}
\begin{split}
 Y_{n+1}\not\in \hat C_n(X_{n+1})  								
~ \Rightarrow~  & \sum_{i=1}^{n+1} \hat w_{\textup{KL},i}B_i\one\left\{R_i\geq R_{n+1} \right\}	\leq  \alpha	\\							
 ~ \Leftrightarrow~& \sum_{i=1}^{n+1}  w_i B_i\one\left\{R_i\geq R_{n+1} \right\}   \\
 &	\leq  \alpha  	\frac{\sum_{i=1}^{n}\phi_i + 1}{\sum_{j=1}^{n}\phi_j B_j + 1}
 \equiv  \alpha \delta(X_{[n+1]},U_{[n+1]}),
\end{split}
\end{equation}
where
$\phi_i=\phi(X_i,X_{n+1})
=
\exp\{-mD_{\textup{KL}}(\hat \pi(X_{n+1})\|\hat \pi(X_i))\},i\in[n],
$
and the renormalized weights over the points with \(U_i=U_{n+1}\) are
\[
w_i=\frac{\phi_i}{\sum_{j=1}^{n}\phi_jB_j+1},\quad i\in[n],
\qquad
w_{n+1}=\frac{1}{\sum_{j=1}^{n}\phi_jB_j+1}.
\]
Let $\mathcal{U}_k = \{i\in [n+1]:U_i=k\}.$ All $R_i$'s for $i\in \mathcal U_k$ are i.i.d. drawn from $f_k$. Then,
\begin{align*}
 &\ \mathbb P\{		Y_{n+1}\not\in \hat C_n(X_{n+1}) \mid X_{[n+1]},U_{[n]},U_{n+1}=k\}\\
\leq  &\ \mathbb{P}\left\{\sum_{i\in \mathcal{U}_k} w_{i}\one\left\{R_i\geq R_{n+1} \right\}\leq  \alpha \delta(X_{[n+1]},U_{[n+1]}) \ \bigg| \	X_{[n+1]},U_{[n]},U_{n+1}=k \right\} \\
=  &\ \sum_{j\in \mathcal U_k }w_j\mathbb{P}\left\{\sum_{i\in \mathcal{U}_k} w_{i}\one\left\{R_i\geq R_{j} \right\}\leq  \alpha \delta(X_{[n+1]},U_{[n+1]}) \ \bigg| \	X_{[n+1]},U_{[n]},U_{n+1}=k \right\} \\
=  &\ \mathbb{E}\Bigg[ \sum_{j\in \mathcal U_k }w_j \one\left\{\sum_{i\in \mathcal{U}_k} w_{i}\one\left\{R_i\geq R_{j} \right\}\leq  \alpha \delta(X_{[n+1]},U_{[n+1]}) \right\} \ \Bigg| \	X_{[n+1]},U_{[n]},U_{n+1}=k \Bigg] \\
\leq &\ \alpha \delta(X_{[n+1]},U_{[n]},k),
\end{align*}
by the deterministic inequality in Lemma A.1 in \citet{harrison2012conservative}. Then,
\begin{equation}\label{equ:exp_bound}
\mathbb P\{Y_{n+1}\not\in \hat C_n(X_{n+1}) \mid X_{[n+1]},U_{n+1}\}
\leq
\alpha \mathbb{E}\left\{ \delta(X_{[n+1]},U_{[n+1]}) \mid X_{[n+1]},U_{n+1} \right\}.
\end{equation}
Next, we use induction to prove an upper bound for the expectation on the RHS of \eqref{equ:exp_bound}.
Without loss of generality, we assume that the probabilities $\pi_k(X_{i})$ are ordered such that 
$\pi_k(X_{1})\leq \dots \leq \pi_k(X_{n}).$
When $n=1$, for any $b\geq 0,$ 
\begin{equation}\label{equ:case_n_1}
\begin{split}
 \mathbb E\left\{\frac{\phi_1 + b+1 }{\phi_1B_1+  b+1} \ \Bigg| \  X_{[n+1]},U_{n+1} = k \right\}
= &\    \pi_{k}(X_1)+ [1- \pi_{k}(X_1)]\frac{\phi_1+b+1}{b+1} \\
= &\  \frac{[1-\pi_{k}(X_1)]\phi_1+b+1}{b+1}  \\
 \leq &\ \frac{[1-\pi_{k}(X_1)]\phi_1+b+\pi_{k}(X_1)\phi_1}{b+\pi_{k}(X_1)\phi_1} \\
= &\  \frac{\phi_1+b}{\phi_1\pi_{k}(X_1)+b}.
\end{split}
\end{equation}
The inequality is obtained by $\pi_{k}(X_1)\phi_1\leq 1$ and that $f(x) = [x+a]/[x+b]$ is a decreasing function of $x$, with a negative derivative $(b-a)/(b+x)^2$ for any $a>b\geq 0.$

When \(n\geq 2\) and \(U_{n+1}=k\), by the strict positivity assumption and the ordering above,
$0<\pi_k(X_i)/\pi_k(X_n)\le 1, i\in[n-1]$.
Thus the Bernoulli variables \(A_i\) below are well-defined. We have
\(B_i \stackrel{\textup{d}}{=} A_iB_n^{(i)}\) for all \(i\in[n-1]\), where
\[
A_i\sim \textup{Bern}\!\left(\frac{\pi_k(X_i)}{\pi_k(X_n)}\right),
\qquad
B_n^{(1)},\dots,B_n^{(n-1)},B_n
\stackrel{\textup{i.i.d.}}{\sim}
\textup{Bern}(\pi_k(X_n)),
\]
with all these variables independent.
For any $b\geq0 $, 
\begin{align}
&\ \mathbb E\left\{\frac{ \sum_{j=1}^{n}\phi_j  +b+1}{\sum_{i=1}^{n}\phi_i B_i +b+1} \ \Bigg| \  X_{[n+1]},U_{n+1} = k \right\}  \notag \notag \\
= &\  \mathbb E\left\{\frac{ \sum_{j=1}^{n}\phi_j  +b+1}{\sum_{i=1}^{n-1}\phi_i A_i B_n^{(i)} + \phi_n B_n +b+1} \ \Bigg| \  X_{[n+1]},U_{n+1} = k \right\} \notag\\
= &\ \mathbb E\Bigg\{   
\mathbb E\left[ \frac{ \sum_{j=1}^{n-1}\phi_j A_j + \phi_n +b+1}{\sum_{i=1}^{n-1}\phi_i A_i B_n^{(i)} + \phi_n B_n +b+1}  \ \Bigg|\ X_{[n+1]},A_{[n-1]}\right] \notag  \\
&\hspace{135pt} \times \frac{ \sum_{j=1}^{n}\phi_j +b+1}{\sum_{j=1}^{n-1}\phi_j A_j + \phi_n +b+1} \ \Bigg|\ X_{[n+1]},U_{n+1}=k\Bigg\} \notag \\
\leq  &\ \mathbb E\Bigg\{   \frac{ \sum_{j=1}^{n-1}\phi_j A_j + \phi_n +b}{ (\sum_{i=1}^{n-1}\phi_i A_i + \phi_n )\pi_k(X_{n})+b} \times \frac{ \sum_{j=1}^{n}\phi_j +b+1}{\sum_{j=1}^{n-1}\phi_j A_j + \phi_n +b+1} \ \Bigg|\ X_{[n+1]},U_{n+1}=k\Bigg\} \notag \\
\leq  &\ \mathbb E\Bigg\{   \frac{ \sum_{j=1}^{n-1}\phi_j A_j + \phi_n +b}{ (\sum_{i=1}^{n-1}\phi_i A_i + \phi_n )\pi_k(X_{n})+b} \ \Bigg|\ X_{[n+1]},U_{n+1}=k\Bigg\}		\notag \\
&\hspace{135pt} \times \mathbb E\Bigg\{  \frac{ \sum_{j=1}^{n}\phi_j +b+1}{\sum_{j=1}^{n-1}\phi_j A_j + \phi_n +b+1} \ \Bigg|\ X_{[n+1]},U_{n+1}=k\Bigg\}									\notag		\\
\leq  &    \frac{ \sum_{j=1}^{n-1}\phi_j \pi_k(X_j)/\pi_k(X_{n}) + \phi_n +b}{ (\sum_{i=1}^{n-1}\phi_i \pi_k(X_i)/\pi_k(X_{n}) + \phi_n )\pi_k(X_{n})+b}  								 \times \frac{ \sum_{j=1}^{n}\phi_j+b}{\sum_{j=1}^{n-1}\phi_j \pi_{k}(X_j)/\pi_{k}(X_n) + \phi_n+b}			\notag			\\
= &\ \frac{ \sum_{j=1}^{n}\phi_j +b}{\sum_{i=1}^{n}\phi_i \pi_{k}(X_i) +b}.	\label{equ:63}
	\end{align}
In the second equality, only the i.i.d. Bernoulli random variables
\(B_n^{(1)},\dots,B_n^{(n-1)},B_n\) are random given
\(X_{[n+1]}\), \(U_{n+1}=k\), and \(A_{[n-1]}\).
We then apply case 2 of Lemma 3 in \cite{barber2022conformal} (i.e. the inverse binomial lemma in \citet{ramdas2019unified}) to obtain the first inequality. 
In the fourth line, the first term is an increasing function of $\sum_{j=1}^{n-1}\phi_j A_j $ if $\pi_k(X_n) <1 $, otherwise a constant function $1$ if $\pi_k(X_n) =1.$ The second term is a decreasing function of $\sum_{j=1}^{n-1}\phi_j A_j $. The covariance between the first and second terms is zero or negative, which gives the second inequality above. In the last inequality, the first term is obtained by Jensen's inequality. 
When $n=2,$ the second term is obtained by \eqref{equ:case_n_1} with $B_1$ changed to $A_1$ and 
 $b$ changed to $b+\phi_2$. 
 
Once eq. \eqref{equ:63} holds for $n=2$, we can apply it to bound the second term in the last inequality for $n=3$, with $(B_1,B_2)$ changed to $(A_1,A_2)$ and  $b$ changed to $b+\phi_3.$ 
By induction, we can prove that eq. \eqref{equ:63} holds for any positive integer $n$. 
Then using the definition of $ \delta(X_{[n+1]},U_{[n+1]})$ in \eqref{equ:def_phi},
\eqref{equ:exp_bound} and \eqref{equ:63} for $b=0$, we obtain
\begin{align*}
\mathbb{E}\left\{ \delta(X_{[n+1]},U_{[n+1]}) \mid X_{[n+1]},U_{n+1} \right\}  \leq \frac{ \sum_{j=1}^{n}\phi_j}{\sum_{i=1}^{n}\phi_i \pi_{U_{n+1}}(X_i)} \leq  \frac{1}{\sum_{i=1}^n \hat w_{\textup{KL},i} \pi_{U_{n+1}}(X_i)}.
\end{align*}
Marginalizing out $U_{n+1}$ given $X_{n+1}$ gives the bound in the theorem. 
\end{proof}

\section{PAC guarantees}\label{sect:pac_appendix}

In this section, we first review the PAC guarantee for the unweighted SCP interval, and extend the proof to weighted prediction intervals. 
In general, the proofs of PAC guarantees rely on concentration inequalities without requiring any algorithm to be a symmetric function of the data $Z_{[n+1]}$. In the following, we mainly use the notation $Z_{[n+1]}$ instead of  $\mathcal D_{[n+1]}$.  
For the intervals \(\hat C_n(X_{n+1})\) considered in this paper, the miscoverage event 
\(\{Y_{n+1}\not\in \hat C_n(X_{n+1})\}\) is equivalent to \(p(Z_{[n+1]})\leq \alpha\), where
\begin{equation}\label{equ:p_n}
p(Z_{[n+1]}) := \sum_{i=1}^{n+1}w_i(Z_{[n+1]}) \one \{ R_i\geq R_{n+1}\}.   
\end{equation}
We define a slightly different p-value which assigns zero weight to $R_{n+1},$ 
\begin{equation}\label{equ:p_n_2}
p'(Z_{[n+1]}) := \sum_{i=1}^{n}w_i'(Z_{[n+1]})\one \{ R_i\geq R_{n+1}\},
\end{equation}
where  $w_i'$ are renormalized over the first $n$ residuals. 

For the SCP interval  in \eqref{equ:interval_scp}, $w_i(Z_{[n+1]})=1/(n+1)$ and $w_i'(Z_{[n+1]})=1/n$.
We define
\begin{equation}\label{equ:p_n_3}
 \alpha (Z_{[n]}) = \mathbb P\big\{p(Z_{[n+1]})\leq \alpha  \mid Z_{[n]}\big\} 
 \ \text{ and }  \   \alpha' (Z_{[n]}) = \mathbb P\big\{p'(Z_{[n+1]})\leq \alpha  \mid Z_{[n]}\big\}.
\end{equation}
We begin with a simple lemma that will be used in the results below.
\begin{lemma}\label{lemma:zero_gap}
	With probability 1, $p(Z_{[n+1]})\geq p'(Z_{[n+1]}).$ 
\end{lemma}
\begin{proof}
Since $w_i(Z_{[n+1]})$ is normalized over all the data points while $w_i'(Z_{[n+1]})$ is only normalized  over first $n$ points, we have 
$w_i(Z_{[n+1]})\leq w_i'(Z_{[n+1]}), \forall i\in [n]. $ Then, 
\begin{align*}
 p(Z_{[n+1]}) - p'(Z_{[n+1]})  &   =  w_{n+1}(Z_{[n+1]}) +   \sum_{i=1}^{n}\big[w_i(Z_{[n+1]}) - w_i'(Z_{[n+1]}) \big] \one \{ R_i\geq R_{n+1}\}            	\\ 
& \geq w_{n+1}(Z_{[n+1]}) + \sum_{i=1}^{n}\big[w_i(Z_{[n+1]}) - w_i'(Z_{[n+1]}) \big] \\
& = 0,
\end{align*}
which proves the claim as required. 
\end{proof}

\subsection{Recap}\label{sect:recap}

\begin{proposition}[Proposition 2a in \citet{vovk2012conditional}]\label{prop3}
If the data points $Z_{[n+1]}$ are i.i.d., for any $\delta \in (0,0.5],$ with probability at least $1-\delta,$ 
	the SCP interval in \eqref{equ:interval_scp} satisfies
\begin{equation}\label{equ:vovk_condition}
\alpha(Z_{[n]}) = \mathbb P\big\{Y_{n+1}\not\in \hat C_n^{\textup{SCP}}(X_{n+1}) \mid Z_{[n]}\big\} \leq \alpha +\sqrt{\frac{\log (1/\delta )}{2n}}.
\end{equation}
\end{proposition}
The guarantee in \eqref{equ:vovk_condition} means that
the coverage rate over many test points is close to \(1-\alpha\) for large sample size $n$. We next revisit two standard proofs of \Cref{prop3}.
\begin{proof} 
By \Cref{lemma:zero_gap},  it suffices to prove \eqref{equ:vovk_condition} holds for 
$ \alpha' (Z_{[n]})$ in \eqref{equ:p_n_3},
where $\ p'(Z_{[n+1]}) = n^{-1}\sum_{i=1}^{n}\one \{R_i\geq R_{n+1}\}	\leq \alpha.$
Proposition 2a in \citet{vovk2012conditional} shows  that
\begin{equation}\label{equ:beta}
 \alpha' (Z_{[n]}) \sim \text{Beta}\big(n_{\alpha},n+1-n_{\alpha}\big) \text{ for } n_{\alpha} =  \floor{\alpha(n+1)},
\end{equation}
which  is the distribution of the $n_{\alpha}$-smallest sample in $n+1$ i.i.d. standard uniform random variables. 
The event $\{ \alpha' (Z_{[n]}) \leq\alpha\}$ occurs when the sample quantile of $R_{[n]}$ is larger than the true marginal quantile of $R_{n+1}$. We have
\[
F_{ \alpha' (Z_{[n]})}\big(\alpha+\Delta  ;n_{\alpha},n+1-n_{\alpha}\big) = 1-\text{Binomial}(n_{\alpha}-1;n,\alpha+\Delta).
\]
The LHS is the probability that the $n_{\alpha}$-smallest standard uniform random variable $ \alpha' (Z_{[n]}) \leq  \alpha+\Delta$. The RHS is the probability that at least $n_{\alpha}$ standard uniform random variables fall into $[0,\alpha+\Delta]$. The two events are equivalent, so they have the same probability.
For $B_n\sim \text{Binomial}(\cdot ;n,\alpha+\Delta)$ and $\Delta = \sqrt{\log (1/\delta )/(2n)},$  
\begin{align*}
\mathbb P\{ \alpha' (Z_{[n]}) \leq   \alpha +\Delta    \} 
= \mathbb P\{B_n \geq n_{\alpha}   \} 			
 \geq   1- \mathbb P\{B_n -(\alpha+\Delta )n \leq -\Delta n   \} 
 \geq  1- \delta.
\end{align*}
The last step is obtained by  Hoeffding's inequality. 
\end{proof}
\begin{proof}
Denote the marginal CDF $F_{R}(\cdot)$ and $G_{R}(\cdot ) = 1- F_{R}(\cdot)$. Then,
\begin{align*}
 \alpha' (Z_{[n]})  =&\ \mathbb P\big\{	G_{R}(R_{n+1} )		\leq \alpha + G_{R}(R_{n+1}) -p'(Z_{[n+1]}) \mid Z_{[n]} \big\}	\\
 \leq &\ \alpha  + \sup_{(x,y)} \Big\{ G_{R}(|y-\hat \mu(x)|) -p'(Z_{[n]},(x,y) ) \Big\},	
\end{align*}
by the probability integral transform. Then, we have
\[
\mathbb P\big\{ \alpha' (Z_{[n]}) \leq   \alpha +\Delta    \big\} \geq  \mathbb P\bigg\{   \sup_{(x,y)} \big\{ G_{R}(|y-\hat \mu(x)|) -p' (Z_{[n]},(x,y) ) \big\}   \leq  \Delta    \bigg\}	 \geq 1-\delta,
\] 
by the $\Delta$ above and the Dvoretzky-Kiefer-Wolfowitz inequality \citep{massart1990tight}.
\end{proof}
\subsection{PCP guarantee of weighted intervals}
For the weighted p-value in \eqref{equ:p_n_2}, the conditional probability $\alpha' (Z_{[n]})$ in \eqref{equ:p_n_3} does not have a tractable distribution as in \eqref{equ:beta} and $p'(Z_{[n]},(x,y)) $ is not an unbiased estimator of $G_{R}(|y-\hat \mu(x)|)$. The proofs in \Cref{sect:recap} are non-applicable to weighted p-values.

Given marginal validity $\mathbb E[\alpha(Z_{[n]})]\leq \alpha$, the PAC guarantee in 
\eqref{eqi:pac} is implied by 
\begin{equation}\label{equ:pac_2}
\mathbb P \left\{ \alpha(Z_{[n]})- \mathbb E[\alpha(Z_{[n]})]   \leq  o(1) \right\}\geq 1-o(1).
\end{equation}
It is known \citep{talagrand1995concentration} that
a random variable $\alpha (Z_{[n]})$ that smoothly depends on many independent random variables $Z_{[n]}$ concentrates at its expectation.  Following this principle,
we write $\alpha(Z_{[n]})- \mathbb E[\alpha(Z_{[n]})]$  as 
\[
\alpha(Z_{[n]})-\mathbb E[\alpha(Z_{[n]})]
=
\sum_{i=1}^n
\left(
\mathbb E[\alpha(Z_{[n]})\mid Z_{[i]}]
-
\mathbb E[\alpha(Z_{[n]})\mid Z_{[i-1]}]
\right).
\]
The martingale differences in the display above are bounded if we make a stability assumption on the weights in the p-values. Then, we can apply the bounded difference inequality to prove a weighted prediction interval is PAC.

\subsubsection{PAC guarantee of covariate weighted intervals}

To illustrate, we first prove a result  by assuming the weights only depend on the features, e.g.,
$w_i(Z_{[n+1]})   \propto  \phi (X_i,X_{n+1}) = \exp  (-\lambda\|X_i-X_{n+1}\|)$. If $\lambda$ is small,
the sum of distances $\psi(X_{-j},X_{n+1})$ defined below increases at a fast rate $n^{\rho_1}$, which can lead to a PAC guarantee. For example, $\rho_1=1$ for the SCP interval using $\lambda =0.$ 

\begin{theorem}\label{prop:3_x}
Assume that $Z_{[n+1]}$ are i.i.d. and the p-value $p(Z_{[n+1]})$ in \eqref{equ:p_n}, based on $\phi(X_i,X_{n+1})$, has a distribution function that is locally Lipschitz in a neighborhood of $\alpha$ given $Z_{[n]}$ a.s. 
Assume that \(0\le \phi(X_i,X_{n+1})\le 1\) for all \(i\in[n+1]\) a.s.
Then, for \(\alpha (Z_{[n]})\) in  \eqref{equ:p_n_3} and any \(\delta\in(0,1)\), there exists a constant \(C>0\) such that
\[
\mathbb P \left\{\alpha (Z_{[n]})  \leq  \alpha +  C\sqrt{ n^{1-2\rho_1  }\log(1/\delta)}\right \}\geq 1-\delta,
\]
if $\psi(X_{-j},X_{n+1}) :=\sum_{i\neq j }\phi (X_i,X_{n+1}) \gtrsim n^{\rho_1} $ for all $j\in [n]$ and   some $\rho_1>1/2$ a.s.\footnote{We use $\sum_{i\neq j }$ to denote the sum from 1 to $n+1$ without including $j$-th element. We let $ a\lesssim b$ indicate that $a\leq Cb$ for some universal constant $C$.}
\end{theorem}

\begin{proof}
For any $j\in [n],$ we can define a leave-one-out p-value 
\[
p_{-j}(Z_{[n+1]}) :=  \sum_{i\neq j }w_{-j}(X_i,X_{n+1})\one \{ R_i\geq R_{n+1}\},
\]
where $w_{-j}$ is the renormalized weights. By 
\begin{equation}\label{equ:w_expression}
\begin{split}
\hspace{-5pt}        w(X_i,X_{n+1})  &  = \frac{\phi(X_{i},X_{n+1})}{\phi(X_{j},X_{n+1}) + \psi(X_{-j},X_{n+1})}    \\
 		&    = \frac{w_{-j}(X_i,X_{n+1}) }{1+ \phi(X_{j},X_{n+1})/\psi(X_{-j},X_{n+1}) },
\end{split}
\end{equation}
we  can rewrite the event in \eqref{equ:p_n} as 
\begin{equation}\label{equ:p_minus_j}
\begin{split}
p_{-j}(Z_{[n+1]})  \leq  &\ \big[\alpha - w (X_j,X_{n+1})\one\{ R_j\geq R_{n+1}\}\big] \times \left[1+ \frac{\phi(X_{j},X_{n+1})}{\psi(X_{-j},X_{n+1})} \right].
\end{split}
\end{equation}
Observe that  the event in  \eqref{equ:p_minus_j}  implies the following event $E_+:$ 
\begin{align*}
p_{-j}(Z_{[n+1]})   & \leq [\alpha -0]\times \left[1+ 1/\psi(X_{-j},X_{n+1}) \right]  = \alpha + \alpha / \psi(X_{-j},X_{n+1}).	
\end{align*}
Conversely, the event in  \eqref{equ:p_minus_j} is implied by 
\[
 p_{-j}(Z_{[n+1]})   
   \leq \alpha  \times \left[1+ \frac{\phi(X_{j},X_{n+1})}{\psi(X_{-j},X_{n+1})} \right]- \frac{\phi(X_{j},X_{n+1})}{\psi(X_{-j},X_{n+1})},
\]
using $\one\{ R_j\geq R_{n+1}\}\leq 1$ and 
the expression of $ w(X_j,X_{n+1})$ in the first equality in \eqref{equ:w_expression}.
Hence,  the event in  \eqref{equ:p_minus_j} is also implied by the following event $E_-$:
\[
 p_{-j}(Z_{[n+1]})  \leq\alpha  +(\alpha-1) /\psi(X_{-j},X_{n+1}).
\]
Then, we know from the equivalence in \eqref{equ:p_minus_j} that
\[
\mathbb P \big\{ E_- \mid Z_{[n]\setminus\{j\}} \big\} \leq \alpha (Z_{[n]}) \leq \mathbb P \big\{ E_+ \mid Z_{[n]\setminus\{j\}} \big\},
\]
where the upper and lower bound hold because $p_{-j}(Z_{[n+1]})$  does not depend on $Z_j.$
Their difference can be bounded using the rate condition on $\psi(X_{-j},X_{n+1})$,
\begin{align*}
 \mathbb P \big\{ \alpha  +(\alpha-1) /\psi(X_{-j},X_{n+1}) <  p_{-j}(Z_{[n+1]})\leq	 \alpha  +\alpha /\psi(X_{-j},X_{n+1})		 \big\} 
\lesssim  n^{-\rho_1}.
\end{align*}
This shows that $\alpha(Z_{[n]})$ satisfies the bounded difference condition in McDiarmid's inequality.
That is, modifying the value of one argument of $\alpha(Z_{[n]})$ changes the value of $\alpha(Z_{[n]})$ by at most $c n^{-\rho_1}$ for some universal constant $c$. Here we also use the local Lipschitz assumption near $\alpha$ to bound the probability mass in the interval of width $O(n^{-\rho_1})$ above.
Applying McDiarmid’s inequality, we prove the claim as required.
\end{proof}

\subsubsection{PAC guarantee of PCP}\label{sect:main_pac}

We next prove that our PCP interval has a PAC guarantee when the membership probabilities in our mixture model satisfy a stability assumption. Our theorem is inspired by the following series of works on full conformal prediction (FCP) \citep{vovk2005algorithmic}. First, \citet{lei2018distribution}  demonstrate that if the predictive model $\hat{\mu}$ is stable
against pertubation of one response data point in the training set, FCP can achieve an interval length comparable to the interval that uses the true conditional quantile of the residual.
More recently, \citet{bian2023training} shows that FCP is impossible to obtain the PAC guarantee of SCP in distribution-free settings. \citet{liang2023algorithmic} then proves that FCP can have a PAC guarantee if $\hat \mu$ is stable against the change of adding multiple data points into the training set. 

We note that the setup of our problem is different from that of FCP. Regarding the assumption, we will require the weights in our p-values to be stable rather than the predictive model $\hat \mu$. Unlike the results above, we now allow the weights to depend on any data in $Z_{[n+1]}$.
In what follows, we use the general notation introduced in (\ref{equ:p_n}-\ref{equ:p_n_3}). 

Conditional on \(\hat L^{\star,Y_{n+1}}=l\), where
\(l=(l_1,\dots,l_J)\in\mathbb N^J\) and \(\sum_{k=1}^J l_k=m\), we write
\[
\phi_i(Z_{[n+1]})
=
\prod_{k=1}^{J}
\big[\hat\pi_k^{Y_{n+1}}(X_i)\big]^{l_k}.
\]
All arguments below are conditional on the event \(\hat L^{\star,Y_{n+1}}=l\); since the bounds are uniform in \(l\), averaging over \(\hat L^{\star,Y_{n+1}}\) gives the stated result.

For $j\in [n]$, we also write
\[
\psi_{-j} := n^{-1}\sum_{i\in [n]\setminus \{j\}}\phi_i(Z_{[n+1]}),
\qquad
\psi_{-(n+1)} := n^{-1}\sum_{i=1}^{n}\phi_i(Z_{[n+1]}).
\]
\clearpage
\begin{assumption} \label{assume:stable3}
The terms \(\psi_{-j}\), \(j\in[n]\), and \(\psi_{-(n+1)}\) belong to \([C_1,C_2]\) for some constants \(C_1,C_2>0\) a.s.
\end{assumption}

\begin{assumption}(Leave-two-out stability). \label{assume:stable}
	For sufficiently large $n$ and some $\rho > 1/2,$
\[
\mathbb P \left\{ 
\| \hat\pi^{-(j,n+1)}(X_i)  -\hat\pi(X_i) \|_1 
\lesssim  n^{-\rho },
\ \forall i,j\in [n] \text{ and } i \neq j
\right\}  = 1- o(1).
\]
\end{assumption}
The notation $\lesssim$ is introduced in \Cref{prop:3_x}.
Here \(\hat\pi=\hat\pi^{Y_{n+1}}\) denotes the probabilities fitted to \(Z_{[n+1]}\), while \(\hat\pi^{-(j,n+1)}\) denotes the probabilities fitted to \(Z_{[n]\setminus\{j\}}\).
Both of them are fitted to $Z_i$, so they have a small difference when evaluated at $X_i.$ 
\begin{theorem}\label{prop:3}
Suppose \Cref{assume:stable,assume:stable3} holds and the data points $Z_{[n+1]}$ are i.i.d.. Suppose the p-value $p(Z_{[n+1]})$ in \eqref{equ:p_n}, computed using the PCP weights, has a distribution function that is locally Lipschitz in a neighborhood of $\alpha$ given $Z_{[n]}$ a.s.
Then the conditional miscoverage probability
\[
\alpha(Z_{[n]})
=
\mathbb P\big\{Y_{n+1}\not\in \hat C_n^{\textup{PCP}}(X_{n+1})\mid Z_{[n]}\big\}
\]
satisfies the following: for every $\varepsilon>0$, there exists a constant $M_\varepsilon>0$ such that
\[
\limsup_{n\to\infty}
\mathbb P\!\left\{
\alpha(Z_{[n]})
>
\alpha + M_\varepsilon n^{-\min\{\rho-1/2,\,1/2\}}
\right\}
\le \varepsilon.
\]
\end{theorem}

\begin{proof}  
We let $\phi_j =\phi_j(Z_{[n+1]})$ and $\psi_{-j} = \psi_{-j}(Z_{[n+1]}).$
Observe that 
\begin{align*}
\ p'(Z_{[n+1]})\leq\alpha  
\Leftrightarrow  &\ \sum_{i=1}^{n}\phi_i\one \{ R_i\geq R_{n+1}\} \leq \alpha n\psi_{-(n+1)} \\
\Leftrightarrow  &\ p(Z_{[n+1]}) = \frac{\sum_{i=1}^{n}\phi_i\one \{ R_i\geq R_{n+1}\}+\phi_{n+1}}{n\psi_{-(n+1)}+\phi_{n+1}}\\
& \hspace{48pt} \leq \alpha + (1-\alpha )w_{n+1}(Z_{[n+1]}).
\end{align*}
By the definition in \eqref{equ:p_n_3}, the local Lipschitz assumption near $\alpha$, and \Cref{assume:stable3},
\begin{equation}\label{equ:alpha_gap}
\mathbb E[\alpha' (Z_{[n]})] - \mathbb E[\alpha (Z_{[n]})] =  \mathbb P\left\{ \alpha < p(Z_{[n+1]}) \leq  \alpha +  (1-\alpha )w_{n+1}(Z_{[n+1]}) \right\} \lesssim n^{-1}.
\end{equation}

By \Cref{lemma:zero_gap} and marginal validity $\mathbb E[\alpha(Z_{[n]})]\le \alpha$, it suffices to show that for every $\varepsilon>0$, there exists a constant $M_\varepsilon>0$ such that
\[
\limsup_{n\to\infty}
\mathbb P\!\left\{
\alpha'(Z_{[n]})-\mathbb E[\alpha'(Z_{[n]})]
>
M_\varepsilon n^{-\min\{\rho-1/2,\,1/2\}}
\right\}
\le \varepsilon .
\]
Let \(\mathcal E_n\) denote the event in \Cref{assume:stable}. 
The bounded-difference argument below is applied on \(\mathcal E_n\); since
\(\mathbb P(\mathcal E_n^c)=o(1)\), its complement is absorbed into the
\(\limsup\) bound.

We fix the value of \(\hat\pi^\star(X_{n+1})\) at \(l/m\), where
\(l=(l_1,\dots,l_J)\in\mathbb N^J\) and \(\sum_{k=1}^J l_k=m\).
We denote the  non-normalized  weight,
\[
\phi_{i}^{-(j,n+1)} 
= 
\prod_{k=1}^{J}
\big[ \hat\pi_k^{-(j,n+1)}(X_i)\big]^{l_k},
\]
for the leave-two-out estimator \(\hat\pi^{-(j,n+1)}\).
We decompose the difference between the two non-normalized weights into two error terms:
\begin{align*}
\left|\phi_i^{-(j,n+1)}-\phi_i\right|
= &\ \Bigg| 
\prod_{k=1}^{J-1}\big[\hat\pi_k^{-(j,n+1)}(X_i)\big]^{l_k}
\left\{
\big[\hat\pi_J^{-(j,n+1)}(X_i)\big]^{l_J}
-
\big[\hat\pi_J(X_i)\big]^{l_J}
\right\} \\
&\quad
+
\big[\hat\pi_J(X_i)\big]^{l_J}
\left\{
\prod_{k=1}^{J-1}\big[\hat\pi_k^{-(j,n+1)}(X_i)\big]^{l_k}
-
\prod_{k=1}^{J-1}\big[\hat\pi_k(X_i)\big]^{l_k}
\right\} \Bigg| \\
\lesssim &\ \sum_{k=1}^{J}
\left|
\big[\hat\pi_k^{-(j,n+1)}(X_i)\big]^{l_k}
-
\big[\hat\pi_k(X_i)\big]^{l_k}
\right| \\
\lesssim &\ n^{-\rho}.
\end{align*}
The first inequality uses the triangle inequality and a telescoping expansion of the product terms.
For the second inequality,
terms with \(l_k=0\) vanish; for \(l_k\ge1\), we use 
\[
a^q-b^q=(a-b)(a^{q-1}+\cdots+b^{q-1})
\]
for integer \(q\ge1\), and \Cref{assume:stable}.
Let $ w_{-j,i}:=\phi_i/(n\psi_{-j})$ and let $ w_{-j,i}^{-(j,n+1)}$ be the normalized version of $\phi_i^{-(j,n+1)}$, that is,
\[
w_{-j,i}^{-(j,n+1)}:=\phi_i^{-(j,n+1)}/(n\psi_{-j}^{-(j,n+1)}),\quad i\in [n]\setminus\{j\},
\]
where
$
\psi_{-j}^{-(j,n+1)}
:=
n^{-1}\sum_{i\in[n]\setminus\{j\}}\phi_i^{-(j,n+1)}.
$

By \Cref{assume:stable,assume:stable3}, we can bound the difference of the normalized weights:
\begin{equation}\label{equ:weight_gap}
\begin{split}
&\ \sum_{i\in[n]\setminus\{j\}}
\left|w_{-j,i}-w_{-j,i}^{-(j,n+1)}\right| \\
= &\ \ \frac{1}{n}
\sum_{i\in[n]\setminus\{j\}}
\left|
\frac{\phi_i}{\psi_{-j}}
-
\frac{\phi_i^{-(j,n+1)}}{\psi_{-j}^{-(j,n+1)}}
\right| \\
\leq &\ \frac{1}{n}
\sum_{i\in[n]\setminus\{j\}}
\left(
\frac{|\phi_i-\phi_i^{-(j,n+1)}|}{\psi_{-j}}
+
\frac{\phi_i^{-(j,n+1)}}{\psi_{-j}\psi_{-j}^{-(j,n+1)}}
\left|\psi_{-j}^{-(j,n+1)}-\psi_{-j}\right|
\right) \\
= &\  O(n^{-\rho}),
\end{split}
\end{equation}
since $\psi_{-j}$ and $\psi_{-j}^{-(j,n+1)}$ are bounded away from zero and 
\[
\left|\psi_{-j}^{-(j,n+1)}- 	\psi_{-j} \right| \leq n^{-1}\sum_{i\in [n]\setminus\{j\}} \left|  \phi_i -  \phi_i^{-(j,n+1)}  \right| = O(n^{-\rho}).
\]
The remaining steps follow from the proof of \Cref{prop:3_x}. First, 
\begin{align*}
 p'(Z_{[n+1]})\leq \alpha \Leftrightarrow &\  \sum_{i\in [n]\setminus\{j\} }\phi_i\one \{ R_i\geq R_{n+1}\} + \phi_j\one \{ R_j\geq R_{n+1}\} \leq \alpha (n\psi_{-j} + \phi_j)  \\
\Leftrightarrow &\  \sum_{i\in [n]\setminus\{j\} } w_{-j,i}\one \{ R_i\geq R_{n+1}\}\leq \alpha  +  \left[\alpha -\one \{ R_j\geq R_{n+1}\} \right]\phi_j/ (n\psi_{-j}).
\end{align*}
Define
\[
p_j^{-(j,n+1)}
:=
\sum_{i\in[n]\setminus\{j\}}
w_{-j,i}^{-(j,n+1)}
\one\{R_i\ge R_{n+1}\}.
\]
By \eqref{equ:weight_gap} and \(\phi_j/(n\psi_{-j})=O(n^{-1})\), the event in the second line is sandwiched between
\[
p_j^{-(j,n+1)}
\leq \alpha-c_1 n^{-\min\{\rho,1\}}  \qquad \text{and} \qquad p_j^{-(j,n+1)} \leq \alpha+c_2 n^{-\min\{\rho,1\}},
\]
for some constants \(c_1,c_2>0\).
Both events in the last display exclude \(Z_j\).
Hence, by the local Lipschitz argument as in the proof of \Cref{prop:3_x},  changing $Z_j$ can change the conditional probability by at most a constant multiple of $n^{-\min\{\rho,1\}}$. Therefore,
\[
\big|\alpha' (Z_{[n]\setminus \{j\}},z_j) -\alpha' (Z_{[n]\setminus \{j\}},z_j')\big|
\lesssim n^{-\min\{\rho,1\}},
\qquad \forall z_j,z_j'\in \mathcal X\times\mathcal Y,
\]
on \(\mathcal E_n\). Applying McDiarmid's inequality on \(\mathcal E_n\), and using
\(\mathbb P(\mathcal E_n^c)=o(1)\), yields the desired conclusion.
\end{proof}

\section{Additional experiments}\label{sect:D}

In the experiments from \Cref{sect:back,sect:empirical}, we compare posterior conformal prediction (PCP) with split conformal prediction (SCP) \citep{vovk2005algorithmic}, SCP+conditional calibration (CC) \citep{gibbs2023conformal} and randomly-localized conformal prediction (RLCP) \citep{hore2023conformal}. Here, we describe the implementation of these methods and discuss the remaining experiments.

\textbf{Methods.}
The SCP interval is defined in \eqref{equ:interval_scp}. The SCP+CC interval is constructed using a linear quantile regression model fitted to the features and residuals in the validation set. Our implementation also includes the randomization strategy suggested by the authors, which makes the interval have exact coverage at $1-\alpha$. The RLCP interval is described in \Cref{sect:related_work}, where the variance $\sigma^2$ in the kernel weights is chosen to keep the effective sample size of the interval at 100; we refer to equation (21) in \citet{hore2023conformal} for more details on this bandwidth selection method.  We implement PCP and select its hyperparameters as described in \Cref{sect:alg}. As detailed there, our mixture model is based on multiple linear models mapping $X_{n+1}$ to predict whether the residual $R_{n+1}$ is smaller than some quantiles defined below \eqref{equ:r_prime}. Our PCP interval in \eqref{equ:posterior_p} is defined non-parametrically using a weighted empirical distribution. In this sense, we can think of PCP as a semiparametric method, whereas our implementation of SCP+CC is parametric, and RLCP is nonparametric. 

\textbf{Worst-slice conditional coverage rate.} Following the steps outlined in  Section S1.2 of \citet{romano2020classification}, we partition the test data into two subsets, with 20\% used to compute the worst slice and 80\% used to evaluate the conditional coverage. We generate slices $\big\{x: v^{\top}x\in [a,b]\big\}$ for 2500 random vectors $v$ on the unit sphere. For every slice, we choose  $a$ and $b$  to minimize the within-slice coverage rate, subject to the slice containing at least 10\% of the data in the first subset. The worst-slice coverage rate is the minimum within-slice coverage rate across all slices.

\subsection{Setting 2}\label{sect:synthetic_data}

The 6-dimensional features and the drift function $f(V)$ are generated in the same way as in Setting 1. 
We let $V$ denote the first feature in $X$.
Here we generate the response $Y$ with a change point in variance:
\[
Y =  f(V)  + 4(1+3\one\{V\leq  5\})\epsilon,
\qquad \epsilon\sim \mathcal{N}(0,1).
\]
As in Setting 1, we generate independent training, validation, and test sets, each containing 5000 random copies of \((X,Y)\). We use the training set to fit a random forest model $\hat \mu$, which can accurately approximate the function $f(V)$. We can imagine the residual $R=|Y-\hat \mu(X) |$ roughly follows a mixture model in \eqref{equ:mixture}, where $\pi_k(X)$ is either 0 or 1 for $V>5$, leading to the sparse scenario described in \Cref{thm:approximate_valid_3}.

\Cref{fig:sim_2_interval} depicts the test responses and their prediction intervals.
The intervals of SCP, SCP+CC and PCP are stable locally with small variations from the predictions by $\hat \mu.$ In comparison, the interval of RLCP exhibits more variability due to matching all the features and randomizing the location of the test point. 
\Cref{fig:sim_2_coverage} shows the local average coverage rates of the intervals over the 250 nearest test points based on the feature $V$.  SCP uses the same empirical quantile to generate the intervals for all the test points. We can see that these intervals are poorly calibrated with coverage of 
around either  0.8 or 1.  In SCP+CC, the linear quantile model fails to capture the nonlinear changes in the conditional residual distribution. Consequently, its coverage rate drops significantly at the change point $V=5$. RLCP also has coverage gaps because matching all the features makes it difficult to detect the changes in the conditional residual distribution due to the first feature $V$.

\clearpage
\begin{figure}[ht]
\vspace{-20pt}
     \centering
  \begin{subfigure}[b]{1\textwidth}      
	\centering 
         \includegraphics[width=0.99\textwidth]{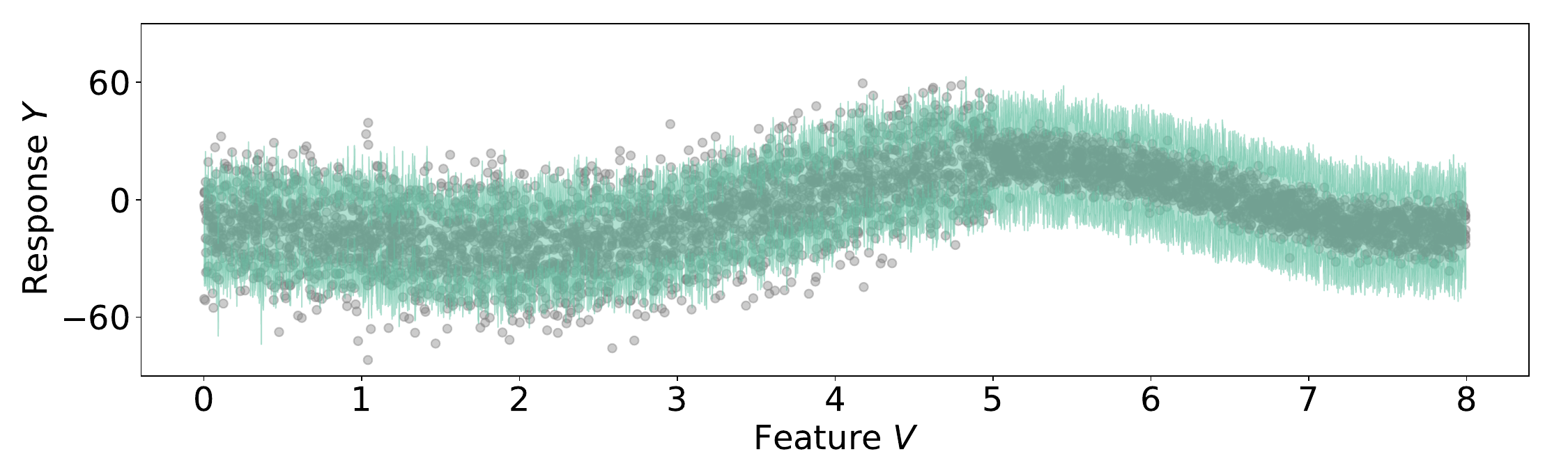}
          \caption{Split conformal prediction (SCP).}
          \vspace{5pt}
          \label{fig:sim_2_SCP}
\end{subfigure}
   \begin{subfigure}[b]{1\textwidth}    
	\centering  
         \includegraphics[width=0.99\textwidth]{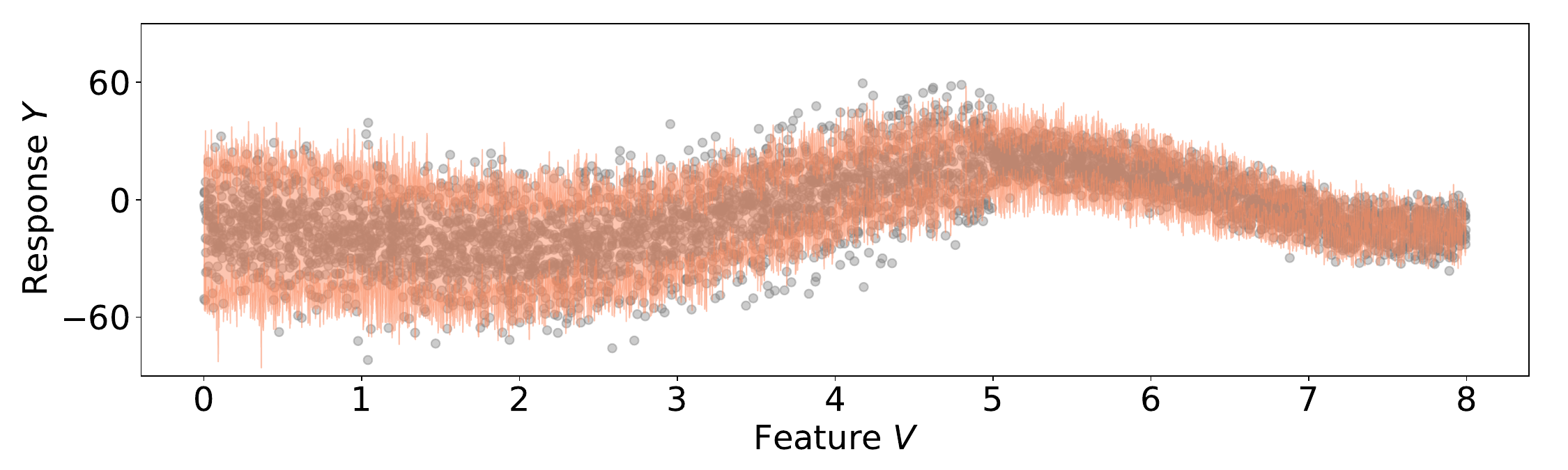}
          \caption{SCP+conditional calibration (SCP+CC).}
                \vspace{5pt}
             \label{fig:sim_2_CC}
	\end{subfigure}	
   \begin{subfigure}[b]{1\textwidth}    
	\centering  
         \includegraphics[width=0.99\textwidth]{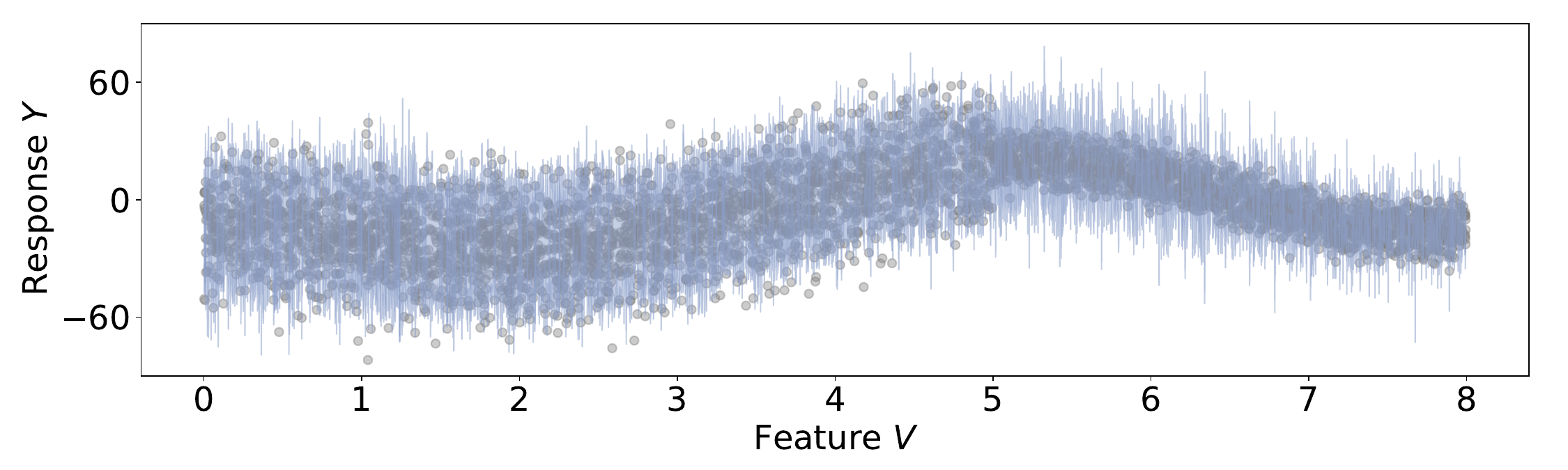}
          \caption{Randomly-localized conformal prediction (RLCP).}
                \vspace{5pt}
             \label{fig:sim_2_RCLP}
	\end{subfigure}	
   \begin{subfigure}[b]{1\textwidth}    
	\centering  
         \includegraphics[width=0.99\textwidth]{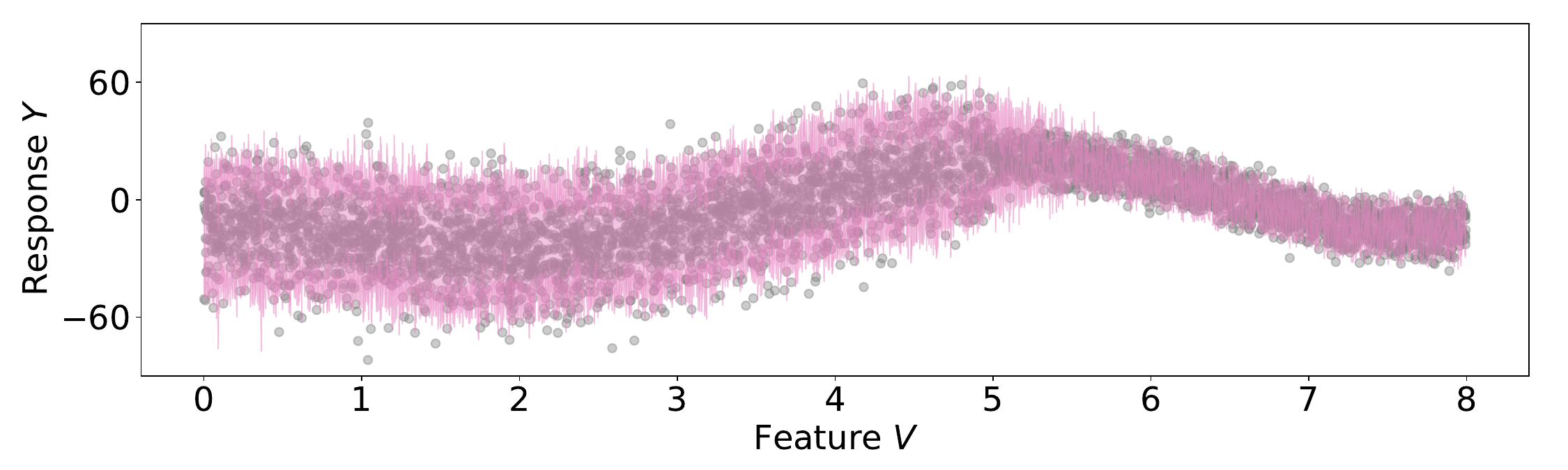}
          \caption{Posterior conformal prediction (PCP).}
                \vspace{5pt}
             \label{fig:sim_2_PCP}
	\end{subfigure}	
\caption{Prediction intervals of conformal prediction methods in Setting 2.}
\label{fig:sim_2_interval}
\end{figure}

\clearpage

\begin{figure}[h] 
	\centering  
         \includegraphics[width=\textwidth]{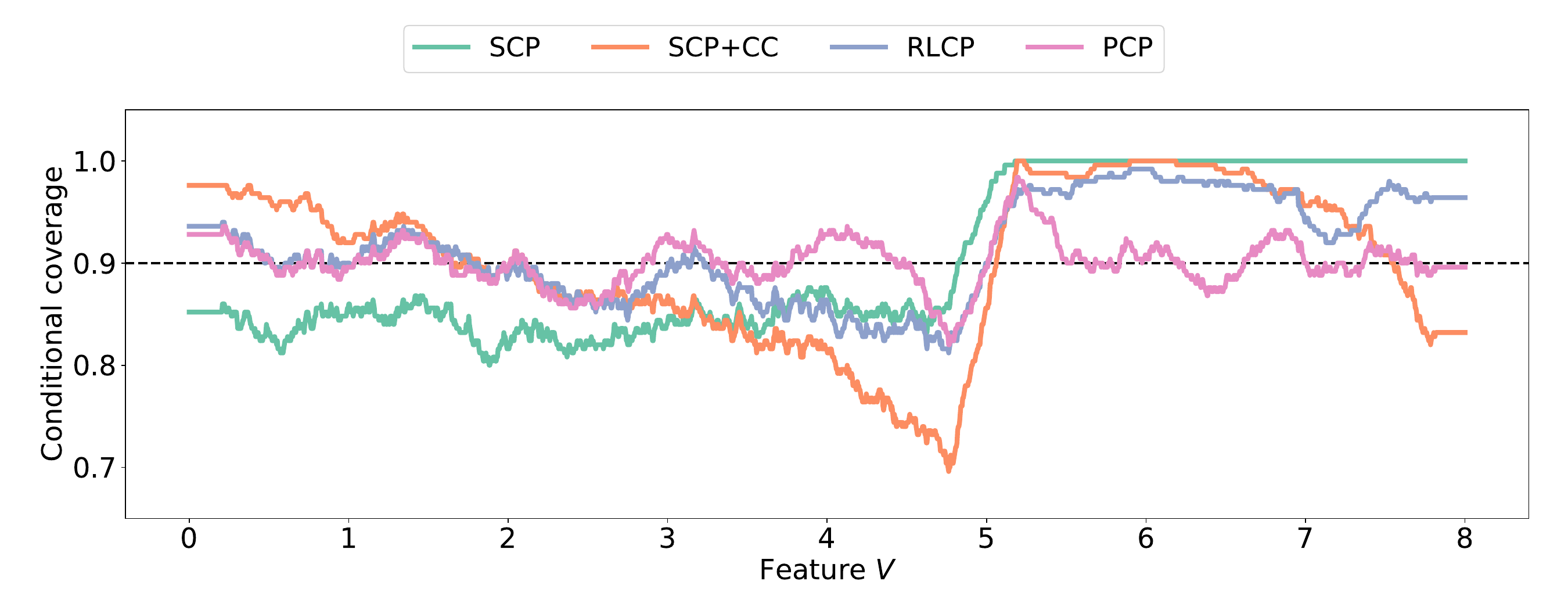}
\caption{Local average coverage rates of conformal prediction methods in Setting 2.}
\label{fig:sim_2_coverage}
\vspace{5pt}
\end{figure}

\begin{figure}[h]
\vspace{25pt}
     \centering
  \begin{subfigure}[b]{0.49\textwidth}      
	\centering 
         \includegraphics[width=0.9\textwidth]{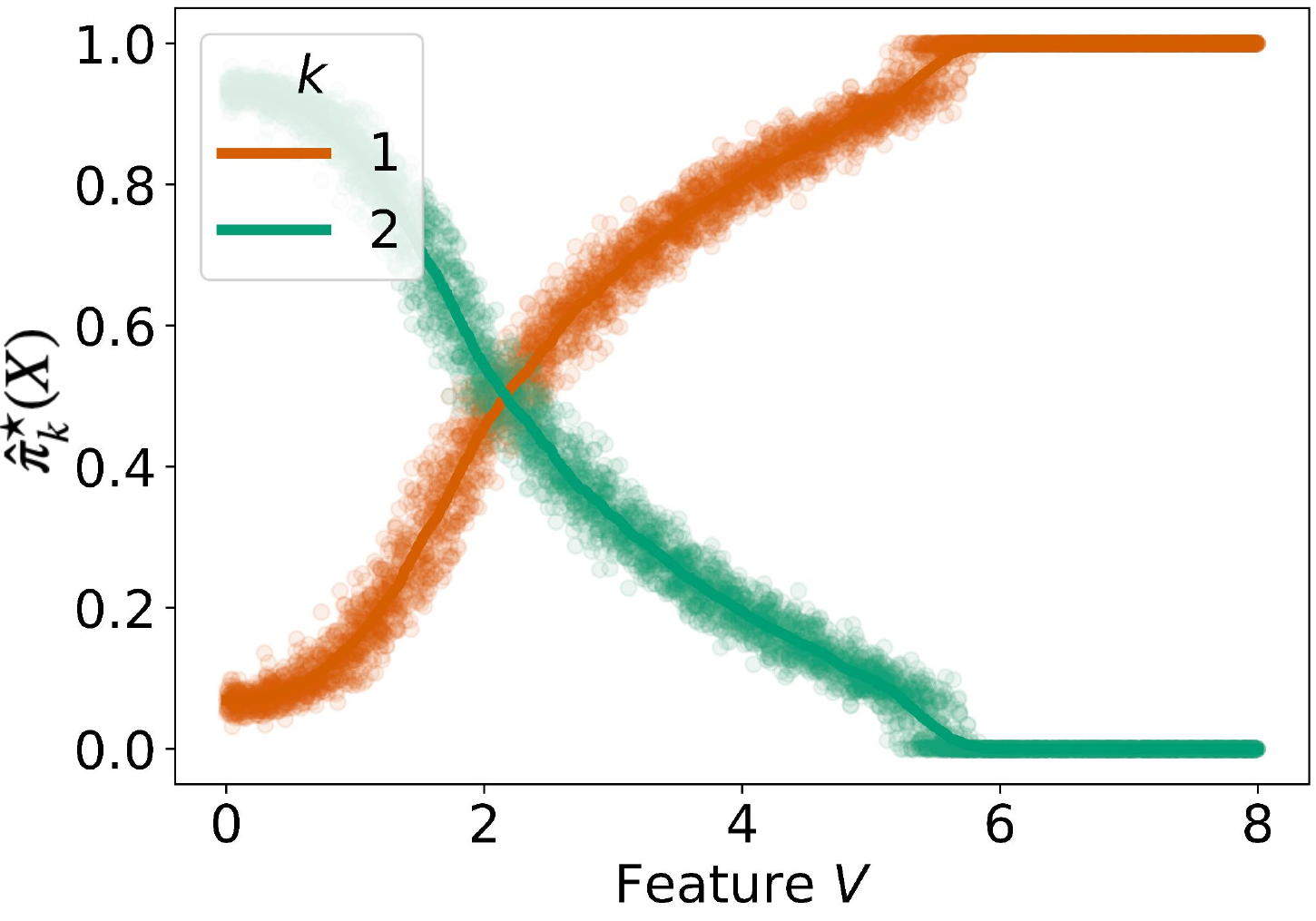}
          \caption{}
          \label{fig:pi_sim_2_1}
\end{subfigure}
  \begin{subfigure}[b]{0.49\textwidth}      
	\centering 
         \includegraphics[width=0.9\textwidth]{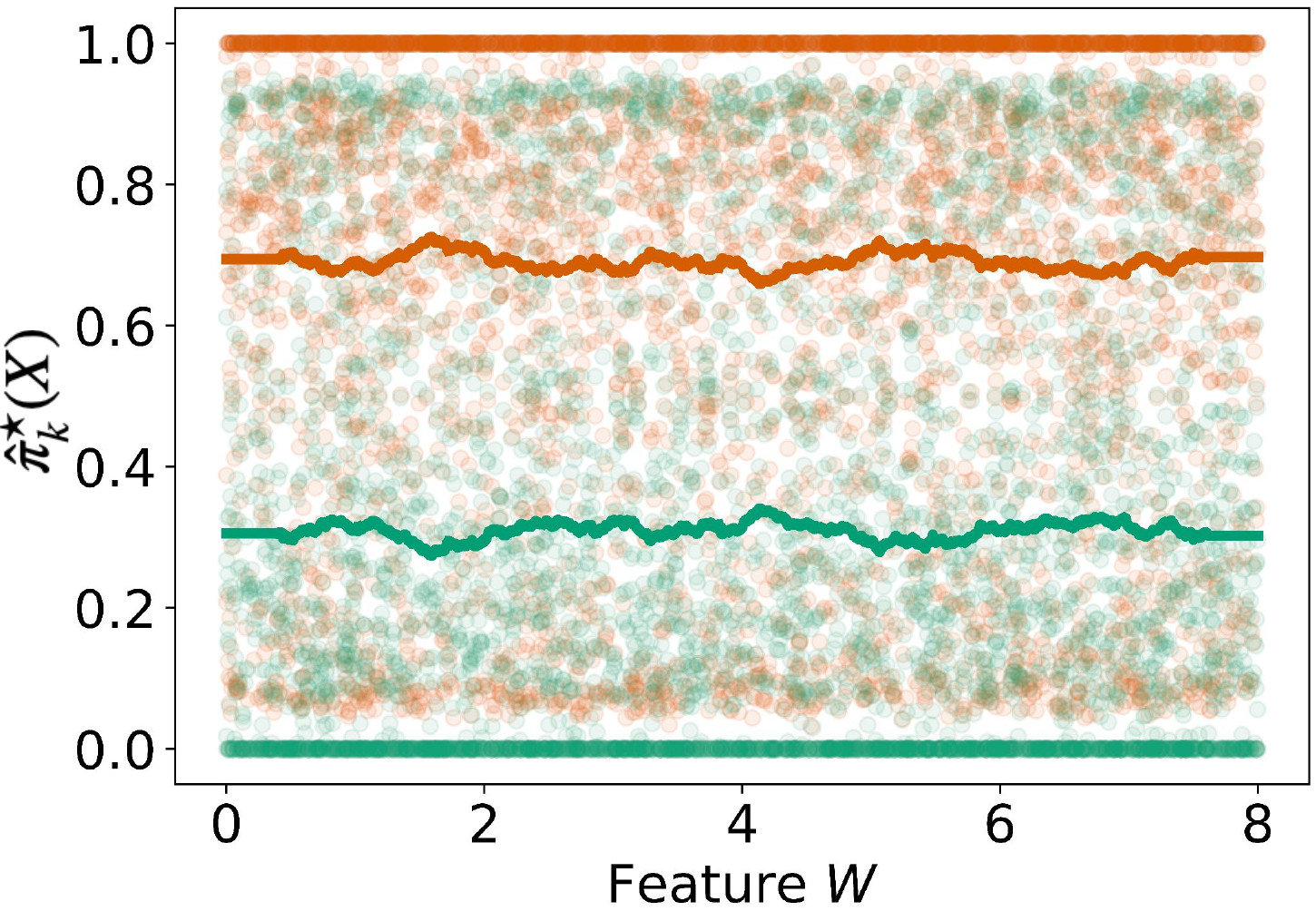}
          \caption{}
           \label{fig:pi_sim_2_2}
\end{subfigure}
\caption{PCP membership probabilities in Setting 2.}
\label{fig:pi_sim2}
\vspace{15pt}
\end{figure}

PCP consistently achieves a coverage rate close to 0.9 at most values of $V$. We interpret this result by visualizing the membership probabilities in \Cref{fig:pi_sim2}. We can see that the probabilities vary continuously over the feature $V$ while changing randomly over the second feature  $W$ in $X$.  This means our mixture model learns that $V$ is the important feature of the conditional residual distribution.
The hyperparameter selection procedure in \Cref{sect:hyper} chooses a large value of the precision parameter, \(m=500\). Thus, the randomized membership probabilities \(\hat \pi_1^\star(X)\) have small deviations from their local average in \Cref{fig:pi_sim_2_1}. With large \(m\) and \(\hat \pi_1^\star(X)\) close to one, PCP achieves approximate conditional validity and short interval length for \(V\in [5,8]\), where the residuals are approximately drawn from the same distribution.

\clearpage

\subsection{
\vspace{-10pt} 
Medical Expenditure Panel Survey dataset
\vspace{10pt}
}\label{sect:repeat_exp}

\begin{figure}[ht]
\hspace{-10pt}
     \centering
     \begin{subfigure}[b]{0.44\textwidth}
         \centering
         \includegraphics[width=0.9\textwidth]{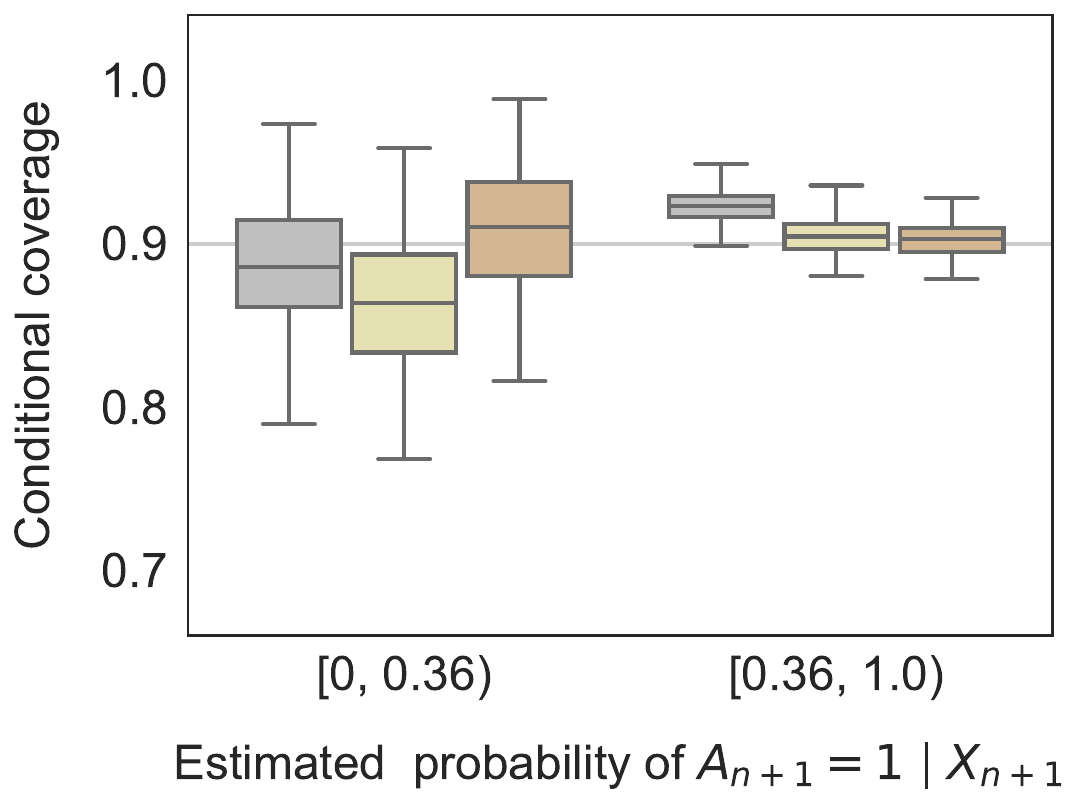}
          \caption{
         \begin{minipage}[t]{0.85\textwidth}
           Coverage rate given $A_{n+1}=1$  and 
           whether  $\hat e^\star(X_{n+1}) <$ or $\geq 0.36.$ 
        \end{minipage}
          }
           \label{fig:fair1_race}
     \end{subfigure}
      \vspace{10pt}
\hspace{5pt}
 \begin{subfigure}[b]{0.44\textwidth}
         \centering
         \includegraphics[width=0.9\textwidth]{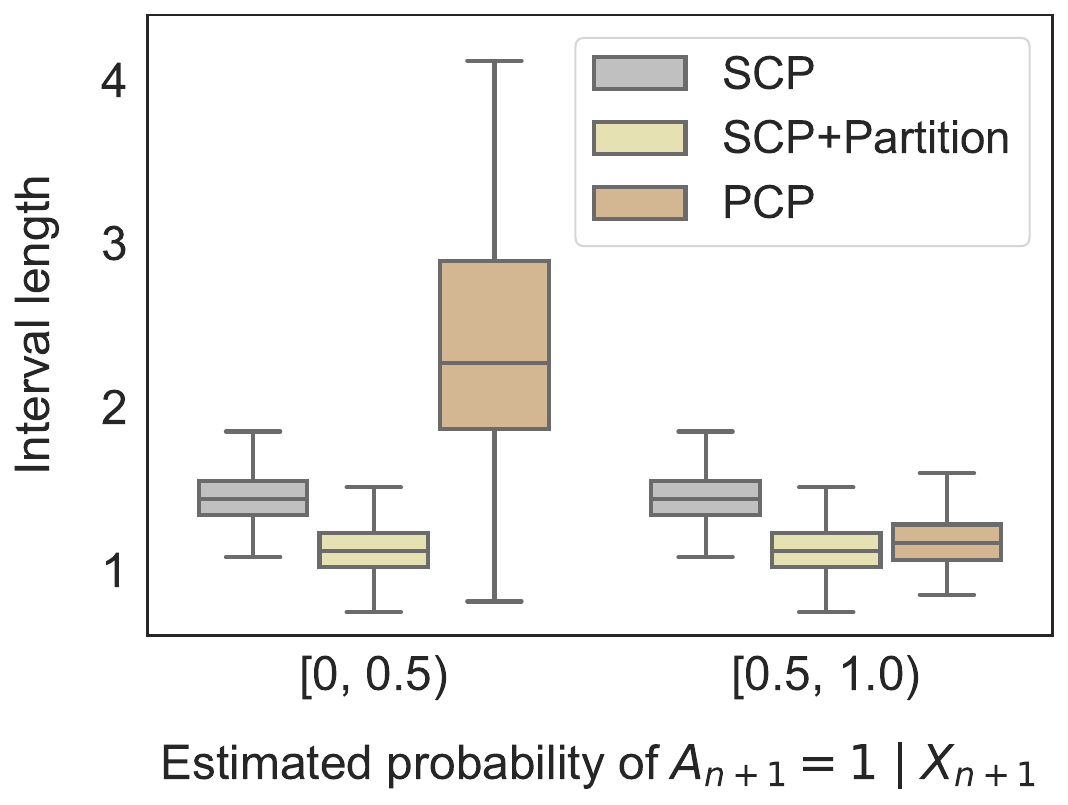}
         \caption{
            \begin{minipage}[t]{0.85\textwidth}
            Interval length given $A_{n+1}=1$  and 
           whether  $\hat e^\star(X_{n+1}) <$ or $\geq 0.36.$  
        \end{minipage}
         }
          \label{fig:fair2_race}
            \end{subfigure}
       \vspace{10pt}
\hspace{-10pt}
     \begin{subfigure}[b]{0.44\textwidth}
          \hspace{5pt}
         \centering
         \includegraphics[width=0.9\textwidth]{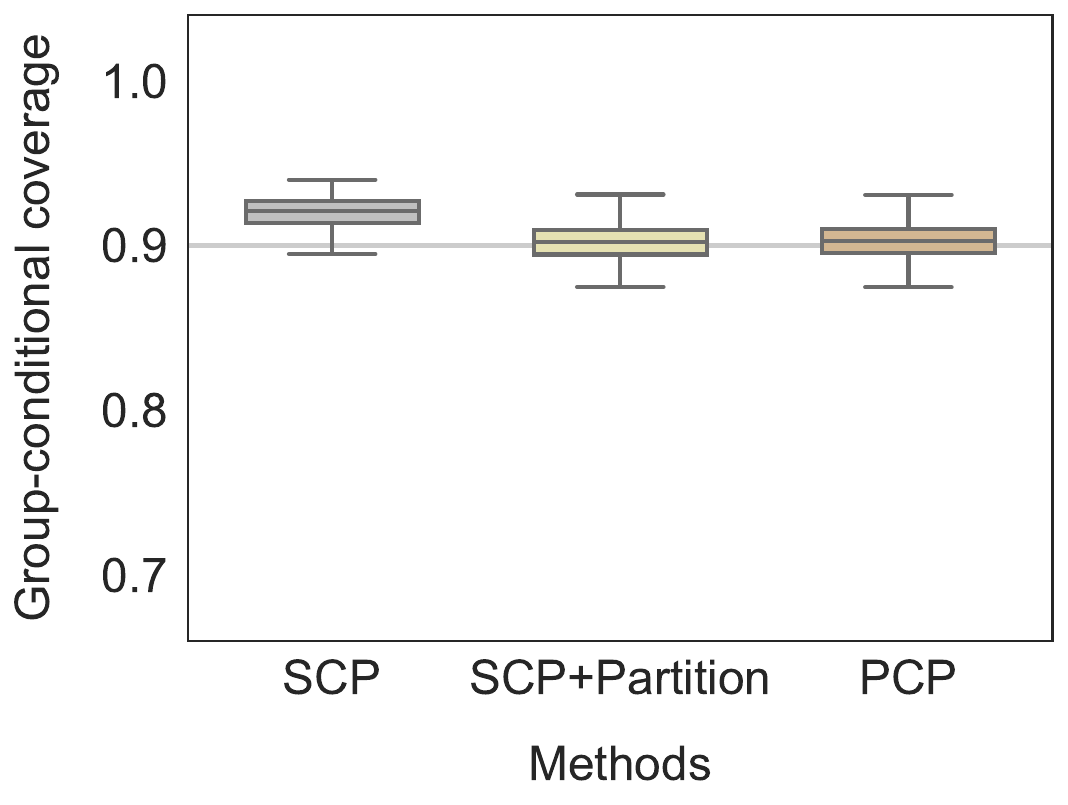}
          \caption{
           \begin{minipage}[t]{0.85\textwidth}
          Group-conditional coverage rate  given $A_{n+1}=1$  (white).     
           \end{minipage}  
          }
     \label{fig:fair3_race}
     \end{subfigure}
           \hspace{5pt}
	\begin{subfigure}[b]{0.44\textwidth}
         \centering
         \includegraphics[width=0.9\textwidth]{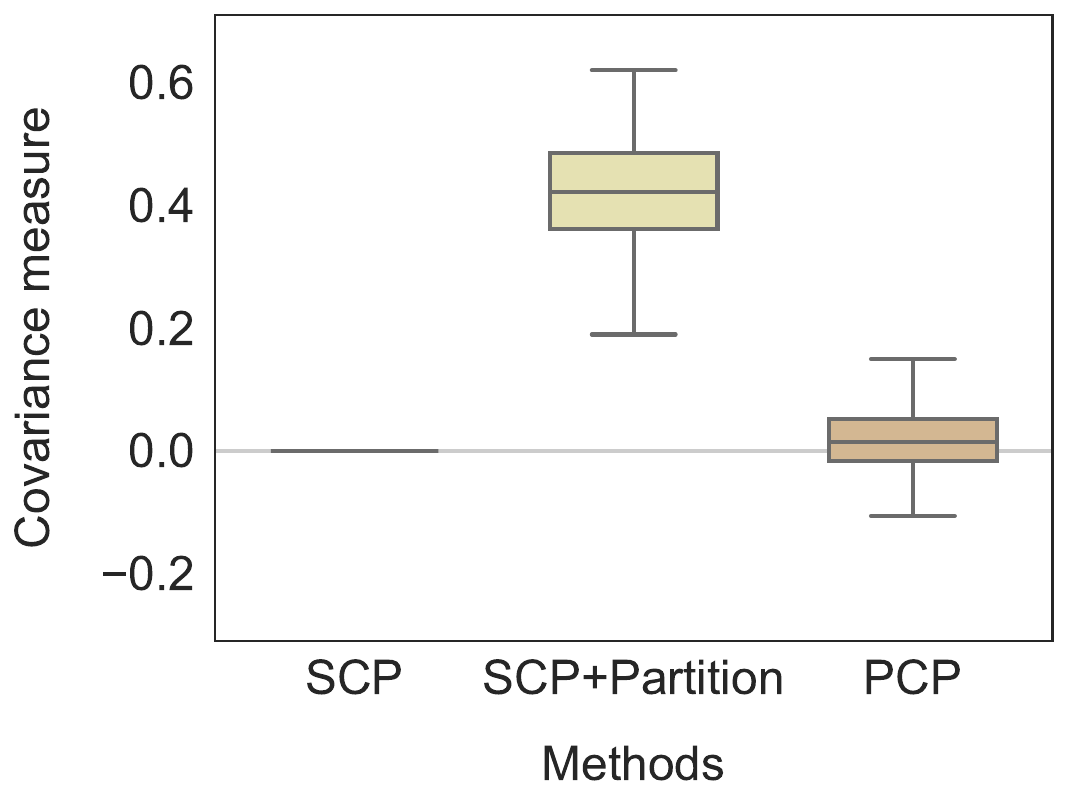}
         \caption{   \begin{minipage}[t]{0.85\textwidth}
         Measure of dependence between intervals and  $A_{n+1}$ given $X_{n+1}$. 
        \end{minipage}              
         }
          \label{fig:fair4_race}
            \end{subfigure}
            \vspace{-5pt}
        \caption{Comparison of conformal prediction methods on a racial feature in the Medical Expenditure Panel Survey (MEPS) 19. }
        \label{fig:fair_race}
            \vspace{2pt}
\end{figure}

We replicate the MEPS experiment in \Cref{sect:fair}, with the variable \(A_{n+1}\) indicating whether the test point is from a white or non-white person. We let \(\theta\approx 0.36\) in \eqref{equ:p_difference}, since white individuals make up approximately \(36\%\) of the full population. We divide the white group \((A_{n+1}=1)\) in the test set into two subgroups based on whether \(\hat e^\star(X_{n+1})<\theta\) or \(\hat e^\star(X_{n+1})\geq\theta\), where \(\hat e^\star(X_{n+1})\) is a randomized estimator defined in \eqref{equ:e_tilde}.
\Cref{fig:fair1_race} shows that SCP+Partition decreases the coverage rate for white individuals with \(A_{n+1}=1\) and \(\hat e^\star(X_{n+1})<0.36\).
PCP closes this gap by weighting, which 
significantly increases the interval length in \Cref{fig:fair2_race}. Nevertheless,  the coverage rate of PCP is around 0.9, meaning it is not overly conservative.
The PCP interval maintains the group-conditional coverage guarantee of SCP+Partition in \Cref{fig:fair3_race}, while being independent of $A_{n+1}$ asymptotically, as shown in \Cref{fig:fair4_race}.

\end{document}